\documentclass[12pt,draftcls,onecolumn]{IEEEtran}

\usepackage{setspace,multirow,tikz}
\usetikzlibrary{arrows,automata}

\usepackage[T1]{fontenc}
\usepackage[latin9]{inputenc}
\usepackage{amsthm}
\usepackage{amsmath}
\usepackage{graphicx}
\usepackage{amssymb}
\usepackage{url}
\usepackage{verbatim}
\usepackage[export]{adjustbox}
\usepackage{framed,caption}

\theoremstyle{plain}
\newtheorem{thm}{Theorem}
\theoremstyle{plain}

\usepackage{cite}\usepackage{amsfonts}\usepackage{times}\usepackage{bm}
\usepackage{amsthm}
\usepackage{subfigure}
\theoremstyle{definition}
\newtheorem{assumption}{Assumption}

\newtheorem{definition}{Definition}

\newtheorem{lemma}{Lemma}
\newtheorem{proced}{Procedure}
\newtheorem{example}{Example}
\newtheorem*{lemma*}{Lemma}

\theoremstyle{remark}
\newtheorem{remark}{Remark}

\newcommand\K{\mathrm{K}}
\newcommand\U{\mathrm{U}}

\title{{\huge Cognitive Access-Transmission Policies under a Primary ARQ process
 via Chain Decoding}}

\author{
\textbf{
Nicol\`o Michelusi\IEEEauthorrefmark{2},
Petar Popovski\IEEEauthorrefmark{3},
Michele Zorzi\IEEEauthorrefmark{1}}
\vspace{-0mm}\\
{\small \IEEEauthorrefmark{2}Ming Hsieh Department of Electrical Engineering,
University of Southern California, USA,\vspace{-2mm}\\ \texttt{\small michelus@usc.edu}} \vspace{-2mm}\\
 {\small \IEEEauthorrefmark{3}Department of Electronic Systems,
Aalborg University, Denmark, \texttt{petarp@es.aau.dk}}
\vspace{-2mm}\\
 {\small \IEEEauthorrefmark{1}Department of Information Engineering,
University of Padova, Italy, \texttt{\small zorzi@dei.unipd.it}}
\vspace{-5mm}
\thanks{This paper was presented in part at 
the \emph{Information  Theory and Applications Workshop (ITA), 2013} \cite{chdecodingfull}.}
}

\graphicspath{%
     {./Figures/}
}

\begin{document}
\maketitle

\vspace{-15mm}
\begin{abstract}
\vspace{-5mm}
This paper introduces a novel technique that enables access by a cognitive secondary user (SU) to a spectrum occupied by an incumbent primary user (PU) that employs Type-I Hybrid ARQ. The technique allows the SU to perform selective retransmissions of SU data packets that have not been successfully decoded in the previous attempts. The temporal redundancy introduced by the PU ARQ protocol and by the selective retransmission process of the SU can be exploited by the SU receiver to perform interference cancellation (IC) over multiple transmission slots, thus creating a "clean" channel for the decoding of the concurrent SU or PU packets. The \emph{chain decoding} technique is initiated by a successful decoding operation of a SU or PU packet and proceeds by an iterative application of IC in order to decode the buffered signals that represent packets that could not be decoded before.  Based on this scheme, an optimal policy is designed that maximizes the SU throughput under a constraint on the average long-term PU performance. The optimality of the chain decoding protocol is proved, which determines which packet the SU should send at a given time.
Moreover, a \emph{decoupling principle} is proved, which establishes the optimality of decoupling the secondary access strategy from the chain decoding protocol.
Specifically, \emph{first}, the SU access policy, optimized via dynamic programming,
specifies whether the SU should access the channel or not, based on a compact state representation of the protocol;
 and \emph{second}, the chain decoding protocol embeds four basic rules that are used to determine which packet should be transmitted by the SU. 
Chain decoding provably yields the maximum improvement that can be achieved by any scheme under our assumptions, and thus it is the ultimate scheme, which completely closes the gap between previous schemes and optimality.
\end{abstract}


\section{Introduction}
The recent proliferation of mobile devices has been exponential in number as well as heterogeneity \cite{CISCO}.
 As mobile data traffic is expected to grow 13-fold, and machine-to-machine traffic will experience a 24-fold increase from 2012 to 2017 \cite{CISCO}, tools for the design and optimization of \emph{agile} wireless networks
are of significant interest \cite{pcast}. Furthermore, network design needs to explicitly consider the resource constraints typical of wireless systems.
Cognitive radio (CR) \cite{Mitola} is a novel paradigm to improve the
spectral efficiency of wireless networks,
by enabling the coexistence of \emph{primary users} (PUs)
and \emph{secondary users} (SUs) in the same spectrum.
SUs are
smart wireless terminals that collect side information
about nearby PUs (\emph{e.g.}, activity, channel conditions, protocols employed, packets exchanged),
and exploit this information to
 adapt their operation in order to opportunistically access the wireless channel
 while generating
bounded interference to the PUs~\cite{FCC,spectrumsharing,Peha}.

In the underlay cognitive radio paradigm \cite{goldsmith},
the PU is a legacy system, oblivious to the presence of the SU, which in turn operates concurrently with the PU and needs to satisfy 
given constraints
on the performance loss caused to the PU.
In this paper,
within this framework, 
we propose a mechanism, termed \emph{chain decoding} (CD),
which exploits the automatic retransmission request (ARQ) protocol implemented by the PU.
In fact,
the PU ARQ mechanism results in replicas of the 
PU packet transmitted over subsequent slots.
 This effectively creates redundancy in the channel, which can be leveraged by the SU
 to implement \emph{interference cancellation} (IC) techniques and boost its own throughput, or alternatively, achieve a target
 throughput with fewer transmissions and less interference to the PU.
Our proposed mechanism leverages opportunistic retransmissions performed by the SU
to improve the spectrum efficiency,
and a 
buffering mechanism at the SU receiver.
Consider, for instance, the following example, depicted in Fig. \ref{figexlabel}.

\begin{figure}
    \centering
    \includegraphics{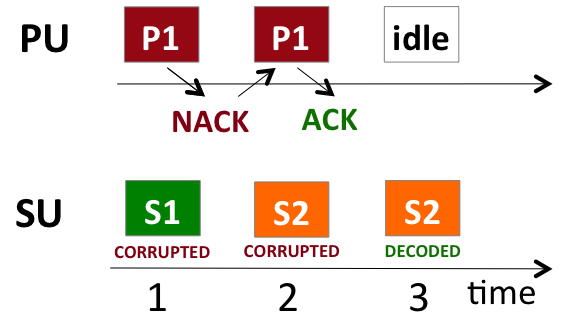}
\caption{Example of chain decoding.}
\label{figexlabel}
\end{figure}

\begin{example}
The PU transmits P1 in slot 1, the transmission is unsuccessful and thus a retransmission occurs in slot 2.
This retransmission is successful, and thus PU remains idle in slot 3, waiting for new data to transmit.
On the other hand, the SU transmits S1 and S2 in slots 1 and 2, respectively, but these transmissions are not successful.
The SU retransmits S2 in slot 3, and successfully decodes it, taking advantage of the fact that the PU is idle in slot 3. \emph{Chain decoding} now starts: 
the interference of S2 is removed from slot 2, and thus the SU can recover P1;
finally, the interference of P1 is removed from slot 1, and thus the SU can recover S1.
That is, IC  is applied in chain,
as SU and PU packets become decodable and their interference is removed.
This gain would not possible if the SU did not apply a clever retransmission and buffering mechanism (specifically, retransmission of S2 in slot 3,
and buffering of the signals received in slots 1 and 2).\qed
\end{example}

More in general, a successful retransmission of a SU packet
may be exploited to perform IC
in the previous transmission attempt of the same packet,
thus potentially enabling the decoding of the interfering PU packet.
In turn, knowledge of the PU packet released via IC
may be exploited to perform IC
in the corresponding ARQ retransmission window of the same packet,
thus potentially enabling the decoding of previously failed
SU transmission attempts, and so on.
 Overall, the decoding of a SU packet releases the decoding of 
the interfering PU packet, which in turn releases the decoding of the SU packets transmitted over the corresponding ARQ window, and so on,
 hence the name \emph{chain decoding}.

Chain decoding opens up intriguing questions. Which signals should be
buffered? Which packet is optimal to be transmitted at a given time instant?
Due to the number of 
possibilities for secondary access
by the SU (remain idle, 
transmit a new data packet, or retransmit some previous data packet),
and to the potentially large number of corrupted packets buffered at the SU receiver, the description of
the CD scheme may require
a very large and possibly unbounded number of states, resulting in prohibitive complexity.
Indeed, a secondary transmission protocol consists of two decisions: (1) \emph{secondary access scheme}: determining whether the SU should transmit or stay idle and (2) \emph{packet selection}: which packet should be sent if  a transmission is made. In general, these two decisions should be made jointly; however, we will show that it is optimal to apply a decoupling principle and separate the decisions, which leads to a simplified protocol specification.
In particular, the CD  protocol
specifies which packet is transmitted by the SU based on four basic rules (Theorem \ref{CD}), whereas the
secondary access scheme determines whether the SU should remain idle or transmit at any given time, depending on the state of the CD  protocol.
Under the CD protocol, we show that the SU throughput admits a closed-form expression and achieves the upper bound (Theorem \ref{thm4} and \ref{thm5}),
obtained under the genie-aided case where the transmission sequence of the SU is generated with non-causal information on the channel state
and on the PU transmission sequence.
Moreover, we prove that this CD protocol defines a compact state space representation of the system,
which is amenable to numerical optimization of the SU access policy via dynamic programming (Theorem \ref{thm6}).
Based on this compact state space representation,
 we model the state evolution of the CD  protocol
as a Markov decision process \cite{Bertsekas,DJWhite},
induced by the specific access policy used by the SU,
which determines its access probability in each state
 of the network. As an application of the proposed CD scheme,
we study the problem of designing optimal secondary access policies 
that maximize the average long-term SU throughput,
while causing 
a bounded average long-term throughput loss to the PU.

There is significant prior work on CR; here, we focus on the literature that is most relevant to our current problem framework.  
The work in \cite{Popovski07} explores the benefits of decoding the PU packet
 at the SU receiver to enable IC. However, no ARQ is assumed.
The idea of exploiting the primary ARQ process to perform IC
 on future packets was proposed by \cite{Nosratinia}.
Therein, the PU employs hybrid ARQ with incremental redundancy
 and the ARQ mechanism is limited to at most one retransmission.
The SU receiver exploits the knowledge of the PU packet,
possibly acquired in the first primary transmission attempt,
to enable IC in case of retransmission,
thus enhancing its own throughput.
In \cite{MichelusiJSAC},
a technique is proposed to exploit the 
 knowledge of the current
PU packet collected at the SU receiver to perform IC 
within the corresponding primary ARQ window where PU transmissions occur.
In particular, \emph{Forward IC} (FIC) enables IC in the subsequent slots
corresponding to primary retransmission attempts, if these occur.
Moreover, previously failed secondary transmission attempts 
may be recovered by using \emph{Backward IC} (BIC)
on the corresponding buffered received signals.
In this work, we further extend these ideas,
by allowing the SU to opportunistically perform retransmissions
of previously failed SU transmission attempts,
so as to introduce redundancy in the secondary channel as well,
which may then be exploited to enable IC \emph{across} different ARQ windows, with the overall effect of improving the secondary throughput
via CD.

Paper
\cite{IT_ARQ}
investigates the interaction between the 
ARQ protocol of the PU and the access scheme of the SU,
but does not exploit the temporal redundancy of ARQ to enable IC.
Paper
 \cite{Nosratinia2} devises an opportunistic sharing scheme with channel probing
 based on the ARQ feedback from the PU receiver.
Compared to \cite{Jovicic}, where the SU transmitter has non-causal 
knowledge of the PU packet, in our work we explicitly model the dynamic acquisition 
of the PU packet at the SU receiver, which enables IC.
In this paper, we assume that the retransmission state of the PU is known at the SU pair, by overhearing the ARQ feedback from the PU receiver.
 The case where the spectrum occupancy is unknown can be analyzed using tools developed in
 \cite{MicheISIT,MicheGlobalsip,Scaglione}, where the state of the PU network is inferred via distributed spectrum sensing.

This paper is organized as follows. In 
Sec.~\ref{sec:sys_model}, we introduce the system model.
In Sec. \ref{CDtech},  we describe the CD technique implemented by the SU pair.
In Sec. \ref{perfmetropt}, we present the optimization problem.
In Sec.~\ref{sec:chdecprot}, we present the four rules of the CD protocol and prove their optimality,
followed by the description of the compact state space representation of the protocol
in Sec.~\ref{compactstsp}.
In Sec.~\ref{sec:numres}, we present some numerical results.
Finally, in Sec.~\ref{sec:remarks}, we conclude the paper. 
The proofs of the analytical results are provided in the Appendix.

\section{System Model}\label{sec:sys_model}
We consider a two-user interference network, depicted in Fig.~\ref{fig:cog_net},
 where a primary and a secondary transmitter, denoted by PUtx 
 and SUtx, respectively,
transmit to their receivers, PUrx and SUrx, over 
the direct links PUtx$\rightarrow$PUrx and SUtx$\rightarrow$SUrx.
 Their transmissions generate mutual interference over the links
PUtx$\rightarrow$SUrx and SUtx$\rightarrow$PUrx.

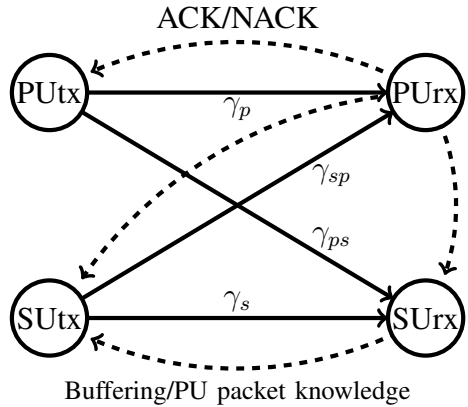
\begin{figure}
    \centering
    {
\begin{tikzpicture}
\draw [ultra thick, ->] (0,0) -- (5-0.5,0);
\draw [ultra thick, ->] (0,3) -- (5-0.5,3);
\draw [ultra thick, ->] (0,0) -- (5-5*0.0857,3-3*0.0857);
\draw [ultra thick, ->] (0,3) -- (5-5*0.0857,0+3*0.0857);
\draw [dashed, ultra thick, ->] (5,3) arc [radius=5, start angle=60, end angle= 113];
\draw [dashed, ultra thick, ->] (5,3) arc [radius=5.83, start angle=92, end angle= 145];
\draw [dashed, ultra thick, ->] (5,0) arc [radius=5, start angle=-60, end angle=-113];
\draw [dashed, ultra thick, ->] (5,3) arc [radius=3, start angle=30, end angle=-18];
\draw [fill=white, ultra thick] (0,0) circle [radius=0.5];
\draw [fill=white, ultra thick] (5,0) circle [radius=0.5];
\draw [fill=white, ultra thick] (0,3) circle [radius=0.5];
\draw [fill=white, ultra thick] (5,3) circle [radius=0.5];
\node at (0,0) {SUtx};
\node at (5,0) {SUrx};
\node at (0,3) {PUtx};
\node at (5,3) {PUrx};
\node at (2.5,+0.22) {$\gamma_s$};
\node at (2.5,3-0.22) {$\gamma_p$};
\node at (5-1.25,3-1.1) {$\gamma_{sp}$};
\node at (5-1.25,1.05) {$\gamma_{ps}$};
\node at (2.5,3+1) {ACK/NACK};
\node at (2.5,-1) {\small Buffering/PU packet knowledge};
\end{tikzpicture}}
\caption{System model}
\label{fig:cog_net}
\end{figure}

Time is divided into slots of fixed duration $\Delta$. Each slot matches the length of 
the PU and SU packets,
and the transmissions of the PU and SU are assumed to be perfectly synchronized.
We adopt the block-fading channel model, \emph{i.e.}, the channel gains are
constant within each slot duration but varies across different slots.
Assuming that the SU and the PU transmit with constant power $P_s$ and $P_p$,
 respectively, and that zero mean Gaussian noise with unit variance is added at the receivers,
we define the signal to noise ratios (SNR) in slot $n$
of the links SUtx$\rightarrow$SUrx, PUtx$\rightarrow$PUrx, SUtx$\rightarrow$PUrx and PUtx$\rightarrow$SUrx,
 as $\gamma_s(n)$, $\gamma_p(n)$, $\gamma_{sp}(n)$ and $\gamma_{ps}(n)$, respectively.
 We model the joint SNR process $\{(\boldsymbol{\gamma}_P(n),\boldsymbol{\gamma}_S(n)),\ n\geq 0\}$, where 
 $\boldsymbol{\gamma}_P(n)=(\gamma_p(n),\gamma_{sp}(n))$
 and
 $\boldsymbol{\gamma}_S(n)=(\gamma_s(n),\gamma_{ps}(n))$,
 as i.i.d. over time, with probability distribution $\mathbb P_{\boldsymbol{\gamma}}(\boldsymbol{\gamma}_S,\boldsymbol{\gamma}_P)$,
 so that the links may be spatially correlated.
 The following analysis can be extended to the case where the SNR process is stationary ergodic with finite first and second order moments.

The SU and PU employ a packet based system, where each packet 
consists of a fixed number of bits $N_s$ and $N_p$, corresponding to fixed transmission
rates $R_{s}$ and $R_p$ bits/s/Hz, respectively.
Both devices may transmit or remain idle in each slot.
We denote the access decision of the SU and PU in slot $n$ as $a_{S,n}\in\{0,1\}$ and $a_{P,n}\in\{0,1\}$, respectively,
where $a_{S,n}=1$ ($a_{P,n}=1$) if the SU (PU) accesses the channel in slot $n$, and
$a_{S,n}=0$ ($a_{P,n}=0$) if it decides to remain idle.
The access decisions are made independently by the SU and PU
according to access policies $\mu_S$ and $\mu_P$, respectively,
introduced in Secs. \ref{suaccesspol} and \ref{puaccesspol}, respectively.

No channel state information (CSI) is available at the transmitters, so that
the latter cannot adjust their transmission rates or power levels
based on the instantaneous link quality $(\boldsymbol{\gamma}_P(n),\boldsymbol{\gamma}_S(n))$.
Additionally, the simultaneous transmissions of the PU and SU generate mutual interference at the respective receivers.
Thus, transmissions
may undergo outage if the transmission rate is not supported by the current
channel quality. 

We now introduce the models for the PU and SU systems.

\subsection{PU system}
\label{puaccesspol}
Herein, we describe the model for the PU system, which specifies the decoding outcomes at PUrx as a function of the activity of the SU pair,
the ARQ scheme, the packet labeling and buffering, the description of the internal state of the PU, the PU access scheme and the internal PU state evolution.
 \subsubsection{Decoding outcome at PUrx}
 Due to the interference generated by SUtx to PUrx, the outcome of the PU transmission (failure or success)
 depends on the SU access decision $a_{S,n}\in\{0,1\}$.
 Additionally, the PU pair is oblivious  to the activity of the SU pair, so that it treats the interfering signal as noise.
Therefore,
the transmission of the PU in slot $n$ is successful if and only if $\boldsymbol{\gamma}_P(n)\in\Gamma_P(a_{S,n})$,
where
\begin{align}
\label{Px1}
\Gamma_P(a_{S,n})\equiv\left\{\boldsymbol{\gamma}_P:R_p<C\left(\frac{\gamma_p}{1+a_{S,n}\gamma_{sp}}\right)\right\}.
\end{align}
In (\ref{Px1}) and hereafter, we have assumed
the use of Gaussian signaling and capacity-achieving coding with sufficiently long codewords,
and we have defined $C(x)\triangleq\log_{2}(1+x)$ as the (normalized) capacity
of the Gaussian channel with SNR $x$ at the receiver \cite{Cover}.

\subsubsection{ARQ scheme}
In order to improve reliability, the PU employs Type{-}I HARQ~\cite{Comroe}
with deadline $R_{\max}>1$,
\emph{i.e.}, at most $R_{\max}$ transmissions of the same PU packet can be performed,
after which the packet is discarded
and a new transmission may be performed.
To this end, the PU receiver, at the end of slot $n$, feeds back the packet $y_{P,n}\in\{\text{ACK},\text{NACK}\}$
to inform the PU transmitter of the transmission outcome, where
 $y_{P,n}=\text{ACK}$ (respectively, $y_{P,n}=\text{NACK}$) indicates that the PU transmission was successful (unsuccessful) in slot $n$.
 If the PU remains idle in slot $n$, then the PU receiver remains idle and
$y_{P,n}=\emptyset$.
 We assume that the feedback packet $y_{P,n}$ is received with no error
by both PUtx and the SU pair.
We define the
\emph{primary ARQ state} $t_{P,n}\in\mathbb{N}(0,R_{\max}-1)$\footnote{We define $\mathbb{N}(n_{0},n_{1})=\left\{ t\in\mathbb{N},n_{0}\leq t\leq n_{1}\right\}$
for $n_{0}\leq n_{1}\in\mathbb{N}$} as the number of 
retransmission attempts for the current PU packet,
and the \emph{ARQ delay} $d_{P,n}\in\mathbb{N}(0,D_{\max}-1)$, with $D_{\max}\geq R_{\max}$,
as the number of slots
 since the current packet was transmitted for the first time,
 where $D_{\max}$ is the maximum tolerable delay for the PU packets.
 Namely, 
if a PU packet is transmitted for the first time in slot $n$, then
$t_{P,n}=0$ and $d_{P,n}=0$;
 the counter $t_{P,n}$ is increased by one unit at each ARQ retransmission,
 and $d_{P,n}$ is increased by one unit in each slot,
until either the ARQ deadline $R_{\max}$ is reached when $t_{P,n}=R_{\max}-1$ and $a_{P,n}=1$ (\emph{i.e.}, the $(R_{\max}-1)$th retransmission attempt is performed),
or the maximum tolerable delay $D_{\max}$ is reached when $d_{P,n}=D_{\max}-1$.
If, in slot $n$, either the ARQ deadline $R_{\max}$ is reached ($t_{P,n}=R_{\max}-1$ and $a_{P,n}=1$), or the delay deadline $D_{\max}$ is reached ($d_{P,n}=D_{\max}-1$),
the packet is, possibly, retransmitted in slot $n$ and then dropped at the end of the slot, irrespective of the transmission outcome.
In case of no active session, we let $t_{P,n}=d_{P,n}=0$.
We let $\nu_P(j),\ j\geq 0$ be the slot index corresponding to the beginning of the $j$th primary ARQ cycle;
  mathematically, $\nu_P(0)=0$  and, for $j>0$, $\nu_P(j)=\min\{n:t_{P,n}=0,d_{P,n}=0,a_{S,n}=1,n>\nu_P(j-1)\}$.
  For $\nu_P(j)\leq n<\nu_P(j+1)$,
  the ARQ delay can thus be expressed as $d_{P,n}=n-\nu_P(j)$.

\subsubsection{PU packet labeling}
Without loss of generality,
each PU packet is univocally labeled
with the slot-index when it was transmitted for the first time, \emph{i.e.}, if the current PU packet is transmitted for the first time in slot $n$ (so that $t_{P,n}=0$),
such packet is assigned the label $l_{P,n}=n_P$,\footnote{We use the subscripts "P" and "S" to refer to PU and SU packets, respectively.} which is used for all future retransmissions of the same packet.
We let $l_{P,n}=n_P$ if $a_{P,n}=0$.

\subsubsection{Packet buffering}
The packets arrive from the upper layer and are stored in a buffer of size $Q_{\max}>0$. 
Packets are served from the data queue according to a first in first out scheme.
The packet arrival process $\{b_{P,n},\ n\geq 0\}$, where $b_{P,n}\in \mathbb{N}(0,B_{\max})$ for some $B_{\max}<\infty$,
is modeled as an i.i.d. process, independent of the SNR process $\{
(\boldsymbol{\gamma}_P(n),\boldsymbol{\gamma}_S(n)),\ n\geq 0\}$,
with probability distribution $\mathbb P_B(b_P)$. The following analysis can be extended
to the case where the data arrival process is stationary ergodic with finite first and second order moments.
We denote the state of the 
queue in slot $n$, \emph{i.e.}, the number of packets stored in the buffer including the current packet under transmission, as $q_{P,n}\in\mathbb{N}(0,Q_{\max})$.
The queue evolution is modeled as
\begin{align}
q_{P,n+1}=\min\{q_{P,n}-o_{P,n}+b_{P,n},Q_{\max}\},
\end{align}
where $o_{P,n}$ takes values $o_{P,n}=1$ if the transmission is successful or the
packet is dropped (due to reaching either  the ARQ or the delay deadlines); otherwise, $o_{P,n}=0$.
Note that $o_{P,n}=0$ when $q_{P,n}=0$, since no packets can be transmitted from an empty queue.
Additionally, if $d_{P,n}=D_{\max}-1$,
  then $q_{P,n}>0$ (since no active session exists with an empty data queue) and necessarily $o_{P,n}=1$, since the packet is dropped, independently of the transmission outcome;
  if $d_{P,n}<D_{\max}-1$ and $y_{P,n}=\emptyset$,
  then $a_{P,n}=0$ and $o_{P,n}=0$ since no PU transmission is performed;
    if $d_{P,n}<D_{\max}-1$ and $y_{P,n}=\text{ACK}$,
  then $q_{P,n}>0$, $a_{P,n}=1$, $\boldsymbol{\gamma}_{P,n}\in\Gamma_P(1)$
  and $o_{P,n}=1$ since the PU transmission is successful;
finally, if $d_{P,n}<D_{\max}-1$ and $y_{P,n}=\text{NACK}$,
  then  $q_{P,n}>0$, $a_{P,n}=1$, $\boldsymbol{\gamma}_{P,n}\notin\Gamma_P(1)$
  and $o_{P,n}=\chi(t_{P,n}=R_{\max}-1)$,  where $\chi(\cdot)$ denotes the indicator function, since the PU transmission is unsuccessful and the packet is discarded only if the ARQ deadline has been reached.
Mathematically, we can write the expression of $o_{P,n}$ as
 \begin{align}
 \label{o}
 o_{P,n}=&
 (1-a_{P,n})\chi(d_{P,n}=D_{\max}-1)\chi(q_{P,n}>0)
 +
 a_{P,n}\chi(\boldsymbol{\gamma}_P\in\Gamma_P(a_{S,n}))\chi(q_{P,n}>0)
 \\&
 \nonumber
+a_{P,n}[1-\chi(\boldsymbol{\gamma}_P\in\Gamma_P(a_{S,n}))]\chi(q_{P,n}>0)\chi(t_{P,n}=R_{\max}-1)\chi(d_{P,n}<D_{\max}-1)
 \\&
 \nonumber
+a_{P,n}[1-\chi(\boldsymbol{\gamma}_P\in\Gamma_P(a_{S,n}))]\chi(q_{P,n}>0)\chi(d_{P,n}=D_{\max}-1).
 \end{align}
 Note that we can express $o_{P,n}$ as a function of $(t_{P,n},d_{P,n},y_{P,n})$, denoted as 
 \begin{align}\label{sigma}
 o_{P,n}=\sigma(t_{P,n},d_{P,n},y_{P,n}).
 \end{align}

\subsubsection{Internal PU state}
\label{sec:intstate}
We denote the internal state of the PU at the beginning of slot $n$ as
\begin{align}
\mathbf s_{P,n}=(t_{P,n},d_{P,n},q_{P,n}),
\end{align}
 where 
$t_{P,n}$ is the ARQ state,
 $d_{P,n}$ is the ARQ delay, and $q_{P,n}$ is the data queue size.
 
 \subsubsection{PU access scheme}
The access decision of PUtx, $a_{P,n}\in\{0,1\}$, is made according to the stationary policy $\mu_P(\mathbf s_P)=\mathbb P(a_{P,n}=1|\mathbf s_{P,n}=\mathbf s_P)$,
representing the probability of choosing action $a_{P,n}=1$
when the internal state of the PU is $\mathbf s_P$.
Clearly, $\mu_P(t_{P,n},d_{P,n},0)=0$, since no transmissions can be performed if the data queue is empty.
This  probabilistic transmission model is general enough to capture, \emph{e.g.}, back-off mechanisms implemented
by the PU.

In this paper, $\mu_P$ is given and is not part of our design. In fact, the PU is 
 oblivious to the activity of the SU. Additionally, $\mu_P$ does not fully specify higher layer specifications of the PU, which are hidden to the SU. 
 Therefore, $\mu_P$ describes only those features of the PU activity which are relevant to the SU access scheme.

 \subsubsection{Internal PU state evolution}
\label{sec:evol}
The internal state of the PU evolves over time as data packets arrive from the upper layer and as a function of the transmission outcome and access decisions.

From state $\mathbf s_{P,n}=(0,0,0)$, \emph{i.e.}, no packets are waiting for transmission in the data queue,
the internal state becomes $\mathbf s_{P,n+1}=(0,0,\min\{b_{P,n},Q_{\max}\})$ in the next slot, since the PU remains idle and $o_{P,n}=0$.

From state $\mathbf s_{P,n}=(0,0,q_{P,n})$ with $q_{P,n}>0$, \emph{i.e.}, $q_{P,n}$ packets are waiting for transmission in the data buffer, and no packet is currently under an active retransmission session,
the internal state becomes: $\mathbf s_{P,n+1}=(0,0,\min\{q_{P,n}+b_{P,n},Q_{\max}\})$, if $a_{P,n}=0$;
$\mathbf s_{P,n+1}=(0,0,\min\{q_{P,n}-1+b_{P,n},Q_{\max}\})$, if $a_{P,n}=1$ and $o_{P,n}=1$ ($l_{P,n}=n_P$ is transmitted successfully at the first attempt);
$\mathbf s_{P,n+1}=(1,1,\min\{q_{P,n}+b_{P,n},Q_{\max}\})$, if $a_{P,n}=1$ and $o_{P,n}=0$ (the transmission of $l_{P,n}=n_P$ is unsuccessful, hence the ARQ state and delay are increased).

From state  $\mathbf s_{P,n}=(t_{P,n},d_{P,n},q_{P,n})$, with $q_{P,n}>0$ and $d_{P,n}\geq t_{P,n}>0$,
the internal state becomes: 
$\mathbf s_{P,n+1}=(0,0,\min\{q_{P,n}-1+b_{P,n},Q_{\max}\})$, if $o_{P,n}=1$ (the transmission is successful or the packet is dropped);
 $\mathbf s_{P,n+1}=(t_{P,n}+1,d_{P,n}+1,\min\{q_{P,n}+b_{P,n},Q_{\max}\})$,
 if $a_{P,n}=1$ and $o_{P,n}=0$ (the transmission is unsuccessful, but the packet is not dropped);
$\mathbf s_{P,n+1}=(t_{P,n},d_{P,n}+1,\min\{q_{P,n}+b_{P,n},Q_{\max}\})$, if $a_{P,n}=0$ and $o_{P,n}=0$ (no retransmission is performed, and the ARQ delay deadline has not been reached yet).

We can combine these cases and write
the internal state $\mathbf s_{P,n+1}=(t_{P,n+1},d_{P,n+1},q_{P,n+1})$ as a function of
  $\mathbf s_{P,n}=(t_{P,n},d_{P,n},q_{P,n})$, $b_{P,n}$, $a_{P,n}$ and $o_{P,n}$ as
\begin{align}
&q_{P,n+1}=\min\{q_{P,n}-o_{P,n}+b_{P,n},Q_{\max}\},
\\&t_{P,n+1}=(1-o_{P,n})(t_{P,n}+a_{P,n}),\label{tp}
\\&d_{P,n+1}=(1-o_{P,n})\left[d_{P,n}+\chi(t_{P,n}>0)+\chi(t_{P,n}=0)a_{P,n}\right].\label{dp}
\end{align}
Since $o_{P,n}$ is a function of $\mathbf s_{P,n}$ and $y_{P,n}$ via (\ref{sigma}), and $a_{P,n}=\chi(y_{P,n}\neq\emptyset)$, we denote the internal state update as
 \begin{align}
 \label{internalstate}
 \mathbf s_{P,n+1}=\phi(\mathbf s_{P,n},b_{P,n},y_{P,n}),
 \end{align}
 where $b_{P,n}$ is i.i.d. over time with probability mass function $\mathbb P_B(b_{P,n})$,
 and $y_{P,n}$ is independent over time, given $\mathbf s_{P,n}$, with probability mass function
 \begin{align}
 \label{pmfY}
& \mathbb P(y_{P,n}=\text{ACK}|\mathbf s_{P,n},a_{S,n})=\mu_P(\mathbf s_{P,n})
 \mathbb P\left(\boldsymbol{\gamma}_P(n)\in\Gamma_P(a_{S,n})\right),
 \nonumber\\&
 \mathbb P(y_{P,n}=\text{NACK}|\mathbf s_{P,n},a_{S,n})=\mu_P(\mathbf s_{P,n})
\left[1- \mathbb P\left(\boldsymbol{\gamma}_P(n)\in\Gamma_P(a_{S,n})\right)\right],
 \nonumber\\&
\mathbb P(y_{P,n}=\emptyset|\mathbf s_{P,n},a_{S,n})=1-\mu_P(\mathbf s_{P,n}).
  \end{align}

\subsection{SU system}
\label{suaccesspol}
Herein, we describe the model for the SU system, which specifies the decoding outcomes at SUrx as a function of the activity of the SU and PU pairs and the knowledge of the current PU packet at SUrx, the feedback message provided by SUrx to SUtx, the buffering mechanism implemented at SUrx,
the labeling of SU packets, and the SU access and labeling policies.
 \subsubsection{Decoding outcomes at SUrx}
SUrx attempts to decode both the PU and SU packets.
If the current PU packet has been decoded at SUrx in a previous 
 slot, its interference can be removed via \emph{Forward Interference Cancellation} (FIC), thus achieving an interference free channel at SUrx.
Therefore, the outcome of the SU transmission in slot $n$ depends on the PU access decision $a_{P,n}\in\{0,1\}$,
and on whether the current PU packet is known or unknown at SUrx.

In order to implement these IC schemes,
the SU pair needs to be able to track the activity of the PU pair (PU access decision $a_{P,n}$ in slot $n$)
and the retransmission process (ARQ state $t_{P,n}$ and delay $d_{P,n}$).
These features can be inferred from the PU feedback sequence $y_{P,0}^{n-1}$, overheard by the SU pair, as detailed in
Lemma \ref{lem1} in Appendix \ref{applemma1}.
Therefore the SU pair knows $(t_{P,n},d_{P,n})$ at the beginning of slot $n$, hence whether the PU will perform 
a retransmission or a new transmission in slot $n$.
However, it does not know in advance the access decision of the PU ($a_{P,n}\in\{0,1\}$),
due to the probabilistic access scheme $\mu_P(\mathbf s_{P,n})\in[0,1]$, and the partial knowledge of $\mathbf s_{P,n}$.

At the end of slot $n$, the SU pair overhears the feedback $y_{P,n}$, and thus infers the value of the PU access decision $a_{P,n}$.
Based on that, SUrx attempts to decode the PU and SU packets jointly (if $a_{P,n}=1$)
or the SU packet only  (if $a_{P,n}=0$). We now analyze the decoding outcomes at SUrx.

\paragraph{Decoding outcomes at SUrx when $a_{P,n}=1$, PU packet unknown}
\label{subsecdecevent}
We denote the current SU and PU packets with their labels $l_S$ and  $l_P$, respectively.
Note that SUtx, PUtx and SUrx form a multiple access channel \cite{Cover}.
Therefore, the region of achievable rates for a given channel quality is as depicted in  Fig. \ref{decevents}.
We have the following possible outcomes:

\begin{figure}[t]
    \centering
{\begin{tikzpicture}
\draw [ultra thick, ->] (0,0) -- (5,0);
\draw [ultra thick, ->] (0,0) -- (0,5);
\node [below] at (2.5,0) {PU rate, $R_p$};
\node [above,rotate=90] at (0,2.5) {SU rate, $R_s$};
\draw [ultra thick, -] (2.5,0) -- (2.5,1.25);
\draw [ultra thick, -] (0,2.5) -- (1.25,2.5);
\draw [ultra thick, -] (1.25,2.5) -- (2.5,1.25);
\draw [ultra thick, -] (2.5,5) -- (2.5,1.25);
\draw [ultra thick, -] (5,2.5) -- (1.25,2.5);
\draw [ultra thick, -] (1.25,5) -- (1.25,2.5);
\draw [ultra thick, -] (5,1.25) -- (2.5,1.25);
\node at (1.25,1) {$\mathcal R_{S,1}(\boldsymbol{\gamma}_S)$};
\node at (1.25+0.625+0.27,1.25+0.625+0.2) {$\mathcal R_{S,7}$};
\node at (0.625,3.75) {$\mathcal R_{S,3}$};
\node at (3.75,0.625) {$\mathcal R_{S,2}(\boldsymbol{\gamma}_S)$};
\node at (1.25+0.625,3.75) {$\mathcal R_{S,6}$};
\node at (3.75,1.25+0.625) {$\mathcal R_{S,5}(\boldsymbol{\gamma}_S)$};
\node at (3.75,3.75) {$\mathcal R_{S,4}$};
\node[right] at (5,2.5) {$C\left(\gamma_s\right)$};
\node[right] at (5,1.25) {$C\left(\frac{\gamma_s}{1+\gamma_{ps}}\right)$};
\node[right,rotate=90] at (2.5,5) {$C\left(\gamma_{ps}\right)$};
\node[right,rotate=90] at (1.25,5) {$C\left(\frac{\gamma_{ps}}{1+\gamma_s}\right)$};
\end{tikzpicture}}
\caption{Decoding regions at SUrx for a given realization of $(\gamma_s,\gamma_{ps})$.
The SU and PU rate pair $(R_s,R_p)$
is a fixed point in the plot. In contrast,
the boundaries of the decoding regions vary as a function of 
$(\gamma_s,\gamma_{ps})$, so that the decoding outcome varies randomly over time depending on
which region $(R_s,R_p)$ falls within.}
\label{decevents}
\end{figure}
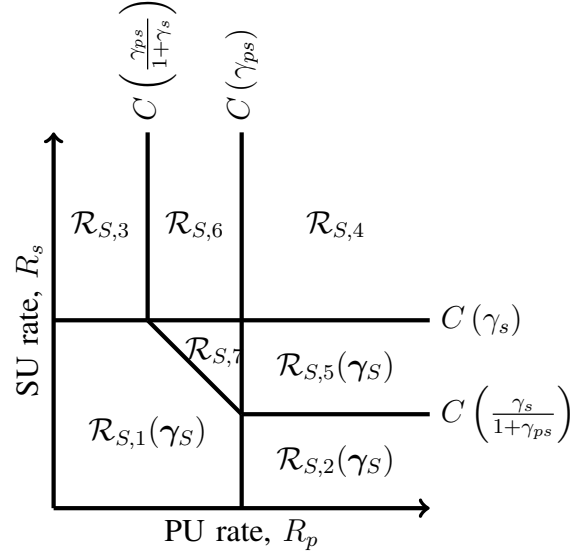

\begin{description}
 \item[O-1] SUrx successfully decodes jointly both $l_S$ and $l_P$;
this event occurs if  $\boldsymbol{\gamma}_S(n)\in\Gamma_{S,1}(R_s,R_p)$, where
\begin{align*}
\Gamma_{S,1}(R_s,R_p)\equiv&\left\{\boldsymbol{\gamma}_S:R_s< C\left(\gamma_s\right),\ 
R_p< C\left(\gamma_{ps}\right),
\right.\\&\left.
R_s+R_p< C\left(\gamma_s+\gamma_{ps}\right)\right\};
\end{align*}
we denote the probability of this event as $\delta_{sp}\triangleq\mathbb P(\boldsymbol{\gamma}_S\in\Gamma_{S,1}(R_s,R_p))$;\footnote{"$\delta$" denotes "decoded",
 with the subscript indicating whether the SU or PU packets are decoded (or both).}
 \item[O-2] SUrx successfully decodes only $l_S$, treating $l_P$ as background noise;
 however, $l_P$ is not decodable, even after removing the interference from $l_S$;
this event occurs if $\boldsymbol{\gamma}_S(n)\in\Gamma_{S,2}(R_s,R_p)$, where
\begin{align*}
\!\!\!\!\!\!\!\!
\Gamma_{S,2}(R_s,R_p)\equiv&\left\{\boldsymbol{\gamma}_S:
R_s< C\left(\frac{\gamma_s}{1+\gamma_{ps}}\right),
R_p\geq C\left(\gamma_{ps}\right)\right\};
\end{align*}
we denote the probability of this event as $\delta_{s}\triangleq\mathbb P(\boldsymbol{\gamma}_S\in\Gamma_{S,2}(R_s,R_p))$;
 \item[O-3] SUrx successfully decodes only $l_P$, treating $l_S$ as background noise;
  however, $l_S$ is not decodable, even after removing the interference from $l_P$;
this event occurs if $\boldsymbol{\gamma}_S(n)\in\Gamma_{S,3}(R_s,R_p)$, where
\begin{align*}
\!\!\!\!\!\!\!\!
\Gamma_{S,3}(R_s,R_p)\equiv&\left\{\boldsymbol{\gamma}_S:
R_s\geq C\left(\gamma_s\right),\ 
R_p< C\left(\frac{\gamma_{ps}}{1+\gamma_s}\right)\right\};
\end{align*}
we denote the probability of this event as $\delta_{p}\triangleq\mathbb P(\boldsymbol{\gamma}_S\in\Gamma_{S,3}(R_s,R_p))$;
 \item[O-4] both $l_S$ and $l_P$ cannot be decoded by SUrx,
 even after removing the interference from the other packet;
this event is denoted as  $l_S\not\leftrightarrow l_P$ and occurs if 
$\boldsymbol{\gamma}_S(n)\in\Gamma_{S,4}(R_s,R_p)$, where
\begin{align*}
\Gamma_{S,4}(R_s,R_p)\equiv&\left\{\boldsymbol{\gamma}_S:
R_s\geq C\left(\gamma_s\right),\ 
R_p\geq C\left(\gamma_{ps}\right)\right\};
\end{align*}
we denote the probability of this event as $\upsilon_{\emptyset}\triangleq\mathbb P(\boldsymbol{\gamma}_S\in\Gamma_{S,4}(R_s,R_p))$;\footnote{"$\upsilon$" denotes "undecoded",
with the subscript indicating whether the SU or PU packets (or none, or both) can be decoded after removing the interference from the other packet.}
 \item[O-5] both  $l_S$ and  $l_P$ cannot be decoded by SUrx; however,
the channel quality is such that, after removing the interference from  $l_P$,
 $l_S$ can be decoded,
or vice versa. In this case, we use an arrow $\rightarrow$ to indicate the
decoding dependence between the two packets. In particular,
 $l_P\rightarrow l_S$ indicates that 
 $l_S$ can be decoded only after removing the interference from  $l_P$,
but  $l_P$ cannot be decoded after removing the interference from  $l_S$,
\emph{i.e.},
$\boldsymbol{\gamma}_S(n)\in\Gamma_{S,5}(R_s,R_p)$, where
\begin{align*}
\!\!\!\!\!\!\!\!
\Gamma_{S,5}(R_s,R_p)\equiv&\left\{\boldsymbol{\gamma}_S:
C\left(\frac{\gamma_s}{1{+}\gamma_{ps}}\right){\leq}R_s{<}C\left(\gamma_s\right),
R_p{\geq}C\left(\gamma_{ps}\right)\!\right\}\!;
\end{align*}
\item[O-6]similarly, the dual event $l_S\rightarrow l_P$ occurs if
$\boldsymbol{\gamma}_S(n)\in\Gamma_{S,6}(R_s,R_p)$, where
\begin{align*}
\!\!\!\!\!\!\!\!
\Gamma_{S,6}(R_s,R_p)\equiv&\left\{\boldsymbol{\gamma}_S:
R_s{\geq}C\left(\gamma_s\right),
C\left(\frac{\gamma_{ps}}{1{+}\gamma_s}\right){\leq}R_p{<}C\left(\gamma_{ps}\right)\!\right\}\!;
\end{align*}
\item[O-7] finally, the event  $l_P\leftrightarrow l_S$ (knowledge of  $l_P$ enables the decoding of  $l_S$, and vice versa)
occurs if 
$\boldsymbol{\gamma}_S(n)\in\Gamma_{S,7}(R_s,R_p)$, where
\begin{align*}
\Gamma_{S,7}(R_s,R_p)\equiv&\left\{\boldsymbol{\gamma}_S:
R_s<C\left(\gamma_s\right),
R_p< C\left(\gamma_{ps}\right),
\right.\nonumber\\&\left.
R_s+R_p\geq C\left(\gamma_s+\gamma_{ps}\right)
\right\};
\end{align*}
we denote the probability that $l_P\rightarrow l_S$ as $\upsilon_{s}\triangleq\mathbb P(\boldsymbol{\gamma}_S\in\Gamma_{S,5}(R_s,R_p))$,
 that $l_S\rightarrow l_P$ as $\upsilon_{p}\triangleq\mathbb P(\boldsymbol{\gamma}_S\in\Gamma_{S,6}(R_s,R_p))$,
 and 
  that $l_P\leftrightarrow l_S$ as $\upsilon_{sp}\triangleq\mathbb P(\boldsymbol{\gamma}_S\in\Gamma_{S,7}(R_s,R_p))$.
\end{description}
For a given SNR $\boldsymbol{\gamma}_S$,
for each $j=1,2,\dots, 7$, we denote the set of rate values $(r_s,r_p)$ such that 
$\boldsymbol{\gamma}_S\in\Gamma_{S,j}(r_s,r_p)$
as
\begin{align}
\label{}
\mathcal R_{S,j}(\boldsymbol{\gamma}_S)\equiv\left\{(r_s,r_p):\boldsymbol{\gamma}_S\in\Gamma_{S,j}(r_s,r_p)\right\},
\end{align}
depicted in  Fig. \ref{decevents}.
In the following treatment, for convenience, we omit the dependence of $\Gamma_{S,j}$ on $(R_s,R_p)$.

\paragraph{Decoding outcomes at SUrx when $a_{P,n}=0$}
If the PU remains idle,
the transmission of the SU is successful if and only if 
$\boldsymbol{\gamma}_S\in\Gamma_{S,1}\cup\Gamma_{S,2}\cup\Gamma_{S,5}\cup\Gamma_{S,7}$,
with probability $\delta_{sp}+\delta_{s}+\upsilon_{s}+\upsilon_{sp}$.

\paragraph{Decoding outcomes at SUrx when $a_{P,n}=1$, PU packet known}
When the current PU packet is known at SUrx
as a result of a previous
PU retransmission of the same packet and
 successful decoding operation at SUrx, its interference can be removed from the received signal,
 thus creating a clean channel.
Therefore, the outcome is the same as in the previous case where $a_{P,n}=0$,
\emph{i.e.}, the transmission of the SU is successful if and only if 
$\boldsymbol{\gamma}_S\in\Gamma_{S,1}\cup\Gamma_{S,2}\cup\Gamma_{S,5}\cup\Gamma_{S,7}$,
with probability $\delta_{sp}+\delta_{s}+\upsilon_{s}+\upsilon_{sp}$.

\subsubsection{Decoding feedback from SUrx}
Let $y_{S,n}\in\{1,\dots,7\}$ be the \emph{decoding outcome} at SUrx, indicating one of the regions depicted in Fig. \ref{decevents},
where $y_{S,n}=j$ if and only if $\boldsymbol{\gamma}_S(n)\in\Gamma_{S,j}$. 
 At the end of each slot, $y_{S,n}$ is fed back from SUrx to SUtx, and received without error by SUtx.
 We emphasize that $y_{S,n}$ represents a feedback that is richer than the ACK, NACK and idle, used by PUrx. 

\subsubsection{Buffering at SUrx and chain decoding}
\label{buffering}
When  $l_P\rightarrow l_S$,  $l_S\rightarrow l_P$ or  $l_P\leftrightarrow l_S$,
occurring with probability $\upsilon_{s}$, $\upsilon_{p}$ and $\upsilon_{sp}$, respectively,
SUrx buffers the corresponding received signals.
In fact, if  $l_P\rightarrow l_S$ or  $l_P\leftrightarrow l_S$, the underlying primary ARQ process
may enable the recovery 
of  $l_S$ in a future slot,
if  $l_P$ is successfully decoded in a subsequent ARQ retransmission,
by removing its interference from the buffered received signals.

Similarly, if  $l_S\rightarrow l_P$ or  $l_P\leftrightarrow l_S$,
the SU may also exploit retransmissions as follows.
It may opportunistically retransmit the buffered  $l_S$, so that,
in the event of a successful decoding operation of  $l_S$ in a future slot,
its interference can be removed from the previously buffered received signal, thus recovering  $l_P$.
In turn, the recovered  $l_P$ may potentially be exploited to recover other SU packets from previously buffered received signals,
as described above.
For analytical tractability, we assume that SUrx is provided with an infinite buffer to store
the received signals.

The process of subsequently decoding a PU or SU packet and removing its interference from previously buffered signals,
thus enabling the decoding of other SU or PU packets, and so on, until no further successive IC operations are possible, is denoted as \emph{chain decoding} (CD).
We term \emph{Forward Interference Cancellation} (FIC) the technique by which 
the current  $l_P$ is decoded in some slots, and its interference is removed in the following slots within its retransmission window,
thus creating a clean channel for SU transmissions.
Finally, we term \emph{Backward Interference Cancellation} (BIC) the technique by which 
the current  $l_P$ is decoded in some slots, and its interference is removed
from signals previously buffered within the current retransmission window.
With BIC and FIC, IC  is limited within the ARQ window where the current PU packet is transmitted, as in \cite{MichelusiJSAC}.
On the other hand, CD enables the use of IC over multiple stages and across multiple ARQ retransmission windows,
by exploiting opportunistic retransmissions by the SU pair.

\subsubsection{SU packet labeling}
The packets transmitted by SUtx are univocally labeled
with the slot-index when they were first transmitted, \emph{i.e.}, if a new packet is transmitted in slot $n$,
it is labeled as 
 $l_{S,n}=n_S$, which is used for all future retransmissions of the same packet. If the SU is idle in slot $n$, we let $l_{S,n}=n_S$.
 
\subsubsection{SU access policy}
The SU, at the beginning of slot $n$, given the PU and SU feedback sequences  $y_{P,0}^{n-1}$,  $y_{S,0}^{n-1}$ collected up to slot $n$, the SU access sequence
$a_{S,0}^{n-1}$,
and the SU label sequence $l_{S,0}^{n-1}$,
decides whether to access the channel or remain idle,
 according to the access policy $\mu_{S,n}(y_{P,0}^{n-1},y_{S,0}^{n-1},a_{S,0}^{n-1},l_{S,0}^{n-1})$,
 representing the probability of choosing $a_{S,n}=1$ in slot $n$.
 
 \subsubsection{SU labeling policy}
 Moreover, if $a_{S,n}=1$, the SU  selects which packet to transmit
 according to the labeling policy $\lambda_{S,n}(l_S|y_{P,0}^{n-1},y_{S,0}^{n-1},a_{S,0}^{n-1},l_{S,0}^{n-1})$,
 representing the probability that the label $l_S$ is chosen in slot $n$. In particular,
 $l_S=n_S$ if a new packet is transmitted, and $l_S<n_S$ if the packet first transmitted in slot $l_S$ (and thus assigned label $l_S$) is retransmitted.

\section{Chain decoding (CD)}
\label{CDtech}
In Sec. \ref{buffering}, we have explained how SUrx buffers the received signals 
when  $l_{P,n}\rightarrow l_{S,n}$,  $l_{S,n}\rightarrow l_{P,n}$ or  $l_{P,n}\leftrightarrow l_{S,n}$,
in order to make it possible to recover these buffered packets in the future via CD.
The decoding relationship between the SU and PU packets buffered at SUrx can be described by a graph,
termed CD graph,
with the set of undecoded SU and PU packets buffered as vertices, and the decoding relationship between them as edges.
For instance, if $l_{S,n}\rightarrow l_{P,n}$, and $l_{S,n}$ and $l_{P,n}$ have not been decoded yet,
then $l_{S,n}$ and $l_{P,n}$ are vertices in the CD graph, 
connected by a directed edge from $l_{S,n}$ to  $l_{P,n}$.
In the following, we describe the construction of the CD graph.

Let $\kappa_{P,n}(l_{P})=1$ (respectively, $\kappa_{S,n}(l_{S})=1$) if the PU packet (SU packet) with label $l_{P}$ ($l_{S}$) has been decoded by SUrx up to slot $n$ (not included),
as a consequence of a  direct decoding operation or via CD,
and $\kappa_{P,n}(l_{P})=0$ ($\kappa_{S,n}(l_{S})=0$) otherwise.
Let $\mathcal V_{P,n}$ and $\mathcal V_{S,n}$ be the set of PU and SU packets
still not decoded by SUrx up to slot $n$ (including the potential transmission of a new PU or SU packet with label $l_P=n_P$ or $l_S=n_S$ in slot $n$). Mathematically,
\begin{align}\label{vpvs}
&\mathcal V_{P,n}=\{l_P\in\{0_P,1_P,\dots,n_P\}:\kappa_{P,n}(l_{P})=0\},
\nonumber\\&
\mathcal V_{S,n}=\{l_S\in\{0_S,1_S,\dots,n_S\}:\kappa_{S,n}(l_{S})=0\}.
\end{align}
Note that these sets may potentially include labels of packets never transmitted (\emph{e.g.}, if SUtx remains idle in slot $k$, then the label $k_S$ is never used for
an SU packet, due to the labeling scheme employed).
Then, the CD graph at the beginning of slot $n$, denoted as
$\mathcal G_n=(\mathcal V_n,\mathbf A_n)$, is a bipartite graph 
with nodes $\mathcal V_n\equiv \mathcal V_{P,n}\cup \mathcal V_{S,n}$, and adjacency matrix
\begin{align}
\label{an}
\mathbf A_n=\left[\begin{array}{cc}\mathbf 0 & \mathbf A_{P\rightarrow S,n}\\
\mathbf A_{S\rightarrow P,n}&\mathbf 0\end{array}\right],
\end{align}
where $\mathbf A_{P\rightarrow S,n}\in\{0,1\}^{|\mathcal V_{P,n}|\times |\mathcal V_{S,n}|}$
is the matrix of edge weights connecting PU packets 
$l_P\in \mathcal V_{P,n}$ to SU packets $l_S\in \mathcal V_{S,n}$,
and 
$\mathbf A_{S\rightarrow P,n}\in\{0,1\}^{|\mathcal V_{S,n}|\times |\mathcal V_{P,n}|}$
is the matrix of edge weights connecting SU packets 
$l_S\in \mathcal V_{S,n}$ to PU packets $l_P\in \mathcal V_{P,n}$.
The edge weight $\mathbf A_{P\rightarrow S,n}(l_P,l_S)$ is set to $1$ if the successful decoding of $l_P$ enables the
decoding of $l_S$ via IC in a previously buffered signal, and to $0$ otherwise.
Similarly, the edge weight $\mathbf A_{S\rightarrow P,n}(l_S,l_P)$ is set to $1$ if the successful decoding of $l_S$ enables the
decoding of $l_P$ via IC in a previously buffered signal, and to $0$ otherwise. Mathematically,
using the notation of Sec. \ref{subsecdecevent},
for each pair $(l_S,l_P)\in\mathcal V_{S,n}\times \mathcal V_{P,n}$:
\begin{itemize}
\item if $\exists k\in\mathbb N(0,n-1):
y_{S,k}\in\{5,7\},a_{P,k}=a_{S,k}=1,l_{S,k}=l_S,l_{P,k}=l_P$,
then $\mathbf A_{P\rightarrow S,n}(l_P,l_S)=1$; otherwise, $\mathbf A_{P\rightarrow S,n}(l_P,l_S)=0$;
\item if $\exists k\in\mathbb N(0,n-1):y_{S,k}\in\{6,7\},a_{P,k}=a_{S,k}=1,l_{S,k}=l_S,l_{P,k}=l_P$,
then $\mathbf A_{S\rightarrow P,n}(l_S,l_P)=1$; otherwise, $\mathbf A_{S\rightarrow P,n}(l_S,l_P)=0$.
\end{itemize}
Note that there are no edges connecting nodes in $\mathcal V_{P,n}$ to nodes in $\mathcal V_{P,n}$, 
nor nodes in $\mathcal V_{S,n}$ to nodes in $\mathcal V_{S,n}$. In fact, in each slot, at most one SU packet and one PU packet are transmitted, rather than
a combination of multiple SU and PU packets.
The packets already decoded at SUrx up to slot $n$ are not included in the graph. In fact, since their interference has been already removed, 
they no longer take part in the CD process.
Those packets never transmitted are isolated nodes in the graph, 
having neither incoming nor outgoing edges.
The CD graph $\mathcal G_n$ captures all information about the state of the buffer at SUrx, since it represents the decoding relationship between
the SU and PU packets transmitted so far.

\subsection{CD graph evolution and instantaneous SU throughput analysis}
\label{secdemecoioni}
In this section, we describe the construction of the CD graph, and we analyze the instantaneous SU throughput accrued via CD.
We let $r_{S,n}$ be the \emph{instantaneous SU throughput} in slot $n$, \emph{i.e.}, the number of SU packets decoded by SUrx in slot $n$.

At the beginning of slot $0$, no packets have been transmitted, and thus $\mathcal V_{P,0}\equiv\mathcal V_{S,0}\equiv\{0\}$,
$\mathbf A_{P\rightarrow S,0}=\mathbf A_{S\rightarrow P,0}=0$,
thus defining the CD graph $\mathcal G_0=(\mathcal V_0,\mathbf A_0)$ via (\ref{vpvs}) and  (\ref{an}).

The evolution of $\mathcal G_n$ over time depends on the outcome at the end of slot $n$
and on which packets are transmitted in slot $n$ by PUtx and SUtx, denoted by their labels $l_{P,n}\in\{0_P,1_P,\dots,n_P\}$ and $l_{S,n}\in\mathcal V_{S,n}$.
Note that the set of SU packets $\{0_S,1_S,\dots,n_S\}\setminus\mathcal V_{S,n}$ are those already decoded by SUrx, and therefore are not retransmitted by SUtx.
On the other hand, PUtx may retransmit a PU packet already decoded by SUrx, if such packet has not been decoded by PUrx yet.
We have different cases, analyzed herein.

\subsubsection{Case $a_{P,n}=0$, $a_{S,n}=0$}
\label{emtpyset}
In this case, both SUtx and PUtx remain idle in slot $n$ and no packets are decoded by SUrx, so that $r_{S,n}=0$, $\kappa_{P,n+1}(l_P)=\kappa_{P,n}(l_P),\forall l_P$
and $\kappa_{S,n+1}(l_S)=\kappa_{S,n}(l_S),\forall l_S$. In the next slot, we thus have
\begin{align}
&\mathcal V_{P,n+1}=\mathcal V_{P,n}\cup\{n+1\},
\nonumber\\&
\mathcal V_{S,n+1}=\mathcal V_{S,n}\cup\{n+1\}.
\end{align}
The sub-matrices $\mathbf A_{P\rightarrow S,n+1}$ and $\mathbf A_{S\rightarrow P,n+1}$ of the adjacency matrix $\mathbf A_{n+1}$ are
given by
\begin{align}
&\mathbf A_{P\rightarrow S,n+1}=
\left[
\begin{array}{cc}
\mathbf A_{P\rightarrow S,n} & \mathbf 0 \\
\mathbf 0 & 0
\end{array}
\right],
\\
&\mathbf A_{S\rightarrow P,n+1}=
\left[
\begin{array}{cc}
\mathbf A_{S\rightarrow P,n} & \mathbf 0 \\
\mathbf 0 & 0
\end{array}
\right].
\end{align}
Note that $\mathbf A_{P\rightarrow S,n+1}$ (respectively, $\mathbf A_{S\rightarrow P,n+1}$) is obtained from $\mathbf A_{P\rightarrow S,n}$ ($\mathbf A_{S\rightarrow P,n}$) by adding a row and a column of zeros,
corresponding to the inclusion of the new (untransmitted) SU and PU packets with label $n+1$.

\subsubsection{Case $a_{P,n}=0$, $a_{S,n}=1$}
\label{cd1}
In this case, PUtx remains idle and SUtx transmits the packet with label $l_{S,n}$ (if $l_{S,n}=n_S$, it is the first transmission attempt).
We distinguish the two cases $y_{S,n}\in\{3,4,6\}$ and $y_{S,n}\in\{1,2,5,7\}$.

If $y_{S,n}\in\{3,4,6\}$, then $l_{S,n}$ cannot be successfully decoded by SUrx, so that $r_{S,n}=0$. The updates of
$\mathcal G_{n+1}$, $\kappa_{P,n+1}$ and $\kappa_{S,n+1}$ are the same as in Sec. \ref{emtpyset}.

On the other hand, if $y_{S,n}\in\{1,2,5,7\}$, then $l_{S,n}$ is successfully decoded by SUrx and the CD technique is initiated.
It works as follows:
starting from node $l_{S,n}$,
 SU and PU packets, previously buffered at SUrx, are decoded subsequently via CD, following the direction of the edges in the graph.
 Mathematically, letting $\mathbf e_S(l_{S,n})$ be a row vector of zeros, except at the position
  corresponding to packet $l_{S,n}$ in the adjacency matrix $\mathbf A_n$ (a similar definition applies to $\mathbf e_P(l_P)$ for a PU packet $l_P\in\mathcal V_{P,n}$),
  after one step of CD the packets recovered are those corresponding to the non-zero elements
  of the vector $\mathbf e_S(l_{S,n})\mathbf A_n$, \emph{i.e.}, $\{l_P\in\mathcal V_{P,n}:[\mathbf A_{S\rightarrow P,n}]_{l_{S,n},l_P}=1\}$.
  The procedure is applied again to each packet recovered, so that
  the PU and SU packets recovered at the $k$th iteration of CD are those corresponding to the non-zero entries
  of the vector $\mathbf e_S(l_{S,n})\mathbf A_n^k$, \emph{i.e.}, 
  $\{l_S\in\mathcal V_{S,n}:\mathbf e_S(l_{S,n})\mathbf A_n^k\mathbf e_S(l_S)^T\geq 1\}\cup\{l_P\in\mathcal V_{P,n}:
  \mathbf e_S(l_{S,n})\mathbf A_n^k\mathbf e_P(l_P)^T\geq 1\}$.
  Therefore, after $i$ iterations of CD,   the PU and SU packets recovered are those corresponding to the non-zero elements
  of
  \begin{align}
  \label{xyz}
\mathbf v_{S,n}^{(i)}=\mathbf e_S(l_{S,n})\chi\left(\sum_{k=0}^i\mathbf A_n^k>0\right),
  \end{align}
  where the indicator function of vectors is applied entry-wise.
  In fact, the PU and SU packets recovered are those corresponding to the non-zero entries of 
  $\mathbf e_S(l_{S,n})\mathbf A_n^k$, for each $k=0,1,\dots,i$.
  The inclusion of $\mathbf A_n^0=\mathbf I$ in (\ref{xyz}) guarantees that also the SU packet $l_{S,n}$ which initiates CD is counted in the throughput accrual. 
This procedure is repeated until no more packets can be decoded, \emph{i.e.}, $\mathbf v_{S,n}^{(i+1)}=\mathbf v_{S,n}^{(i)}$. Overall, when CD is initiated from $l_{S,n}\in\mathcal V_{S,n}$ after a successful decoding operation of $l_{S,n}$, the PU and SU packets recovered
after termination of CD
 are those corresponding to the non-zero elements
of
\begin{align}
\label{vni}
\mathbf v_{S}^*(l_{S,n};\mathcal G_n)\triangleq\lim_{i\to\infty}\mathbf v_{S,n}^{(i)}
=\lim_{i\to\infty}
\chi\left(\mathbf e_S(l_{S,n})\sum_{k=0}^i\mathbf A_n^k>0\right).
\end{align}
This limit exists, since the argument within the function $\chi(\cdot)$ in (\ref{vni}) is a vector with non-decreasing entries (in  the iteration index $i$),
and $0\leq\chi(\cdot)\leq 1$.

Therefore, we have that
\begin{align}
&\kappa_{P,n+1}(l_P)=\mathbf v_{S}^*(l_{S,n};\mathcal G_n)\mathbf e_P(l_P)^T,\forall l_P\in\mathcal V_{P,n},
\\
&\kappa_{S,n+1}(l_S)=\mathbf v_{S}^*(l_{S,n};\mathcal G_n)\mathbf e_S(l_S)^T,\forall l_S\in\mathcal V_{S,n},
\end{align}
where we have used the fact that $\mathbf v_{S}^*(l_{S,n};\mathcal G_n)\mathbf e_X(l_X)^T,X\in\{S,P\}$ equals one if and only if 
packet $l_X$ has been decoded by the end of the CD scheme.

\begin{definition}
Given $\mathcal G_n$,
 we define the \emph{CD potential} of node $l_{S,n}\in\mathcal V_{S,n}$, $v_S(l_{S,n};\mathcal G_n)$, as the number of SU packets that can be decoded by initiating CD
 from the SU packet $l_{S,n}$ (including $l_{S,n}$ itself). Mathematically,
 \begin{align}
 v_S(l_{S,n};\mathcal G_n)=
\mathbf v_{S}^*(l_{S,n};\mathcal G_n)\sum_{l_S^\prime\in\mathcal V_{S,n}}\mathbf e_S(l_S^\prime)^T.
 \end{align}
 \qed
\end{definition}
With this definition, the instantaneous SU throughput accrued in slot $n$ is given by $r_{S,n}=v_S(l_{S,n};\mathcal G_n)$, which includes packet $l_{S,n}$ itself.
In the next slot, the CD graph becomes $\mathcal G_{n+1}=(\mathcal V_{n+1},\mathbf A_{n+1})$, 
obtained by pruning from $\mathcal G_n$ the nodes and the edges corresponding to those PU and SU packets recovered via CD,
and adding the new unconnected SU packet $(n+1)_S$ and PU packet $(n+1)_P$.

\subsubsection{Case $a_{P,n}=1$, $l_{P,n}\in\mathcal V_{P,n}$, $a_{S,n}=0$}
\label{cd2}
In this case, PUtx transmits a PU packet still undecoded by SUrx and SUtx remains idle.
We distinguish the two cases $y_{S,n}\in\{2,4,5\}$ and $y_{S,n}\in\{1,3,6,7\}$.

If $y_{S,n}\in\{2,4,5\}$, then $l_{P,n}$ cannot be successfully decoded by SUrx, so that $r_{S,n}=0$. The updates of
$\mathcal G_{n+1}$, $\kappa_{P,n+1}$ and $\kappa_{S,n+1}$ are the same as in Sec. \ref{emtpyset}.

On the other hand, if $y_{S,n}\in\{1,3,6,7\}$, then $l_{P,n}$ is successfully decoded by SUrx and CD is initiated.
Similarly to the case analyzed in Sec. \ref{cd1},
the PU and SU packets recovered
after termination of CD are those corresponding to the non-zero elements
of
\begin{align}
\mathbf v_{P}^*(l_{P,n};\mathcal G_n)\triangleq \lim_{i\to\infty}\chi\left(\mathbf e_P(l_{P,n})\sum_{k=0}^i\mathbf A_n^k>0\right).
\end{align}
Therefore, we have that
\begin{align}
&\kappa_{P,n+1}(l_P)=\mathbf v_{P}^*(l_{P,n};\mathcal G_n)\mathbf e_P(l_P)^T,\forall l_P\in\mathcal V_{P,n},
\\
&\kappa_{S,n+1}(l_S)=\mathbf v_{P}^*(l_{P,n};\mathcal G_n)\mathbf e_S(l_S)^T,\forall l_S\in\mathcal V_{S,n}.
\end{align}

\begin{definition}
Given $\mathcal G_n$,
 we define the \emph{CD potential} of node $l_{P,n}\in\mathcal V_{P,n}$, $v_P(l_{P,n};\mathcal G_n)$, as the number of SU packets that can be decoded by initiating CD
 from the PU packet $l_{P,n}$. Mathematically,
 \begin{align}
 v_P(l_{P,n};\mathcal G_n)=
\mathbf v_{P}^*(l_{P,n};\mathcal G_n)\sum_{l_S^\prime\in\mathcal V_{S,n}}\mathbf e_S(l_S^\prime)^T.
 \end{align}
 \qed
\end{definition}
With this definition, the instantaneous SU throughput accrued in slot $n$ is given by $r_{S,n}=v_P(l_{P,n};\mathcal G_n)$.
In the next slot, the CD graph becomes $\mathcal G_{n+1}=(\mathcal V_{n+1},\mathbf A_{n+1})$, 
obtained by pruning from $\mathcal G_n$ the nodes and the edges corresponding to those PU and SU packets recovered via CD,
and adding the new unconnected SU packet $(n+1)_S$ and PU packet $(n+1)_P$.

\subsubsection{Case $l_{P,n}\in\mathcal V_{P,n}$, $l_{S,n}\in\mathcal V_{S,n}$}
In this case, both PUtx and SUtx transmit. Moreover, the PU packet transmitted is still unknown to SUrx.
The outcome depends on the value of $y_{S,n}$, as detailed below:
\begin{itemize}
\item $y_{S,n}=1$: both $l_{P,n}$ and $l_{S,n}$ are jointly decoded and CD is initiated from both packets, thus combining the cases analyzed in
Secs. \ref{cd1} and \ref{cd2}. In particular,
the PU and SU packets recovered
after termination of CD are those corresponding to the non-zero elements of
\begin{align}
\mathbf v^*(l_{S,n},l_{P,n};\mathcal G_n)
\triangleq
\lim_{i\to\infty}
\chi\left(\left[\mathbf e_P(l_{P,n})+\mathbf e_S(l_{S,n})\right]\sum_{k=0}^i\mathbf A_n^k>0\right).
\end{align}

Therefore, we have that
\begin{align}
&\kappa_{P,n+1}(l_P)=\mathbf v^*(l_{S,n},l_{P,n};\mathcal G_n)\mathbf e_P(l_P)^T,\forall l_P\in\mathcal V_{P,n},
\\
&\kappa_{S,n+1}(l_S)=\mathbf v^*(l_{S,n},l_{P,n};\mathcal G_n)\mathbf e_S(l_S)^T,\forall l_S\in\mathcal V_{S,n}.
\end{align}

\begin{definition}
Given $\mathcal G_n$,
 we define the \emph{joint CD potential} of nodes $l_{S,n}\in\mathcal V_{S,n}$
 and $l_{P,n}\in\mathcal V_{P,n}$, $v(l_{S,n},l_{P,n};\mathcal G_n)$, as the number of SU packets that can be decoded by initiating CD
 from the SU packet $l_{S,n}$  (including $l_{S,n}$ itself) and PU packet $l_{P,n}$. Mathematically,
 \begin{align}
 v(l_{S,n},l_{P,n};\mathcal G_n)=
\mathbf v^*(l_{S,n},l_{P,n};\mathcal G_n)\sum_{l_S^\prime\in\mathcal V_{S,n}}\mathbf e_S(l_S^\prime)^T.
 \end{align}
 \qed
\end{definition}
With this definition, the instantaneous SU throughput accrued in slot $n$ is given by $r_{S,n}= v(l_{S,n},l_{P,n};\mathcal G_n)$.
In the next slot, the CD graph becomes $\mathcal G_{n+1}=(\mathcal V_{n+1},\mathbf A_{n+1})$, 
obtained by pruning from $\mathcal G_n$ the nodes and the edges corresponding to those PU and SU packets recovered via CD,
and adding the new unconnected SU packet $(n+1)_S$ and PU packet $(n+1)_P$.
\item $y_{S,n}=2$: $l_{S,n}$ is decoded by treating $l_{P,n}$ as noise, whereas $l_{P,n}$ cannot be decoded. 
This case is the same as the one analyzed in Sec. \ref{cd1}.
\item $y_{S,n}=3$: $l_{P,n}$ is decoded by treating $l_{S,n}$ as noise, whereas $l_{S,n}$ cannot be decoded. 
This case is the same as the one analyzed in Sec. \ref{cd2}.
\item $y_{S,n}=4$:  neither $l_{P,n}$ nor $l_{S,n}$ can be decoded, even after removing the mutual interference,
due to poor channel quality. 
This case is the same as the one analyzed in Sec. \ref{emtpyset}.
\item $y_{S,n}\in\{5,6,7\}$: neither $l_{P,n}$ nor $l_{S,n}$ can be decoded, but they are buffered since they may be decoded in the future
by removing the mutual interference. Therefore, $r_{S,n}=0$ since CD cannot be initiated, so that
$\kappa_{P,n+1}(l_P)=\kappa_{P,n}(l_P),\forall l_P$
and $\kappa_{S,n+1}(l_S)=\kappa_{S,n}(l_S),\forall l_S$.
The next CD graph $\mathcal G_{n+1}$ is obtained in two intermediate steps. First,
the new sets $\mathcal V_{P,n+1}$ and $\mathcal V_{S,n+1}$ are defined as 
\begin{align}
&\mathcal V_{P,n+1}=\mathcal V_{P,n}\cup\{n+1\},
\nonumber\\&
\mathcal V_{S,n+1}=\mathcal V_{S,n}\cup\{n+1\}.
\end{align}
Then, the intermediate adjacency matrix $\tilde{\mathbf A}_{n+1}$ is defined as
\begin{align}
\tilde{\mathbf A}_{n+1}=\left[\begin{array}{cc}\mathbf 0 & \tilde{\mathbf A}_{P\rightarrow S,n+1}\\
\tilde{\mathbf A}_{S\rightarrow P,n+1}&\mathbf 0\end{array}\right],
\end{align}
with sub-matrices
\begin{align}
&\tilde{\mathbf A}_{P\rightarrow S,n+1}=
\left[
\begin{array}{cc}
\tilde{\mathbf A}_{P\rightarrow S,n} & \mathbf 0 \\
\mathbf 0 & 0
\end{array}
\right],
\\
&\tilde{\mathbf A}_{P\rightarrow S,n+1}=
\left[
\begin{array}{cc}
\tilde{\mathbf A}_{P\rightarrow S,n} & \mathbf 0 \\
\mathbf 0 & 0
\end{array}
\right],
\end{align}
corresponding to the inclusion of the new SU and PU packets with label $n+1$.
Then, the sub-matrices $\mathbf A_{P\rightarrow S,n+1}$ and $\mathbf A_{S\rightarrow P,n+1}$ of the adjacency matrix $\mathbf A_{n+1}$ are
defined as
\begin{align}
&[\mathbf A_{P\rightarrow S,n+1}]_{l_P^\prime,l_S^\prime}=[\tilde{\mathbf A}_{P\rightarrow S,n+1}]_{l_P^\prime,l_S^\prime},\forall (l_P^\prime,l_S^\prime)\in\mathcal V_{P,n+1}\times\mathcal V_{S,n+1}\setminus\{(l_{P,n},l_{S,n})\}
\\
&[\mathbf A_{P\rightarrow S,n+1}]_{l_{P,n},l_{S,n}}=
\left\{
\begin{array}{ll}
1&\text{if }y_{S,n}\in\{5,7\}\\
\left[\tilde{\mathbf A}_{P\rightarrow S,n+1}\right]_{l_{P,n},l_{S,n}}&\text{if }y_{S,n}=6,
\end{array}
\right.
\end{align}
and
\begin{align}
&[\mathbf A_{S\rightarrow P,n+1}]_{l_S^\prime,l_P^\prime}=[\tilde{\mathbf A}_{S\rightarrow P,n+1}]_{l_S^\prime,l_P^\prime},\forall (l_S^\prime,l_P^\prime)\in\mathcal V_{S,n+1}\times\mathcal V_{P,n+1}\setminus\{(l_{S,n},l_{P,n})\}
\\
&[\mathbf A_{S\rightarrow P,n+1}]_{l_{S,n},l_{P,n}}=
\left\{
\begin{array}{ll}
\left[\tilde{\mathbf A}_{S\rightarrow P,n+1}\right]_{l_{S,n},l_{P,n}}&\text{if }y_{S,n}=5\\
1&\text{if }y_{S,n}\in\{6,7\},
\end{array}
\right.
\end{align}
\emph{i.e.}, edges are added corresponding to the decoding relationship between $l_{S,n}$ and $l_{P,n}$.
\end{itemize}

\subsubsection{Case $a_{P,n}=1$, $l_{P,n}\in\{0_P,1_P,\dots,n_P\}\setminus\mathcal V_{P,n}$}
In this case, PUtx transmits a packet which is known by SUrx due to a previous successful decoding operation. In turn, SUrx can remove its interference from the received signal.
After the interference from the PU transmission has been removed, this case becomes the same as the one analyzed in Secs. \ref{emtpyset} and \ref{cd1}, depending on whether 
SUtx remains idle ($a_{S,n}=0$) or transmits ($a_{S,n}=1$).

We now provide an example of construction of the CD graph.
\begin{example}\label{ex1}
Consider a sequence of 4 slots $\{0,1,2,3\}$. PUtx transmits packets $0_P$, $0_P$, $2_P$, in sequence;
SUtx transmits packets $0_S$, $1_S$, $1_S$, in sequence. The decoding outcome at SUrx is such that
$0_P\rightarrow 0_S$ in slot 0,
$1_S\rightarrow 0_P$ in slot 1 and
$2_P\rightarrow 1_S$ in slot 2.
The corresponding CD graph thus evolves as in Fig. \ref{exconstruc}.
Correspondingly, at the beginning of slot $1$ (end of slot $0$) we have
\begin{align}
\mathbf A_{P\rightarrow S,1}=
\left[
\begin{array}{cc}
1 & 0\\
0 & 0
\end{array}
\right],
\quad
\mathbf A_{S\rightarrow P,1}=
\left[
\begin{array}{cc}
0 & 0\\
0 & 0
\end{array}
\right];
\end{align}
at the beginning of slot $2$ 
\begin{align}
\mathbf A_{P\rightarrow S,2}=
\left[
\begin{array}{ccc}
1 & 0 & 0\\
0 & 0 & 0\\
0 & 0 & 0
\end{array}
\right],
\quad
\mathbf A_{S\rightarrow P,2}=
\left[
\begin{array}{ccc}
0 & 0 & 0\\
1 & 0 & 0\\
0 & 0 & 0
\end{array}
\right];
\end{align}
at the beginning of slot $3$ 
\begin{align}
\mathbf A_{P\rightarrow S,3}=
\left[
\begin{array}{cccc}
1 & 0 & 0 & 0\\
0 & 0 & 0 & 0\\
0 & 1 & 0 & 0\\
0 & 0 & 0 & 0
\end{array}
\right],
\quad
\mathbf A_{S\rightarrow P,3}=
\left[
\begin{array}{cccc}
0 & 0 & 0 & 0\\
1 & 0 & 0 & 0\\
0 & 0 & 0 & 0\\
0 & 0 & 0 & 0
\end{array}
\right].
\end{align}
\begin{figure}
    \centering
    \includegraphics[width=\linewidth,trim = 1mm 1mm 1mm 1mm,clip=false]{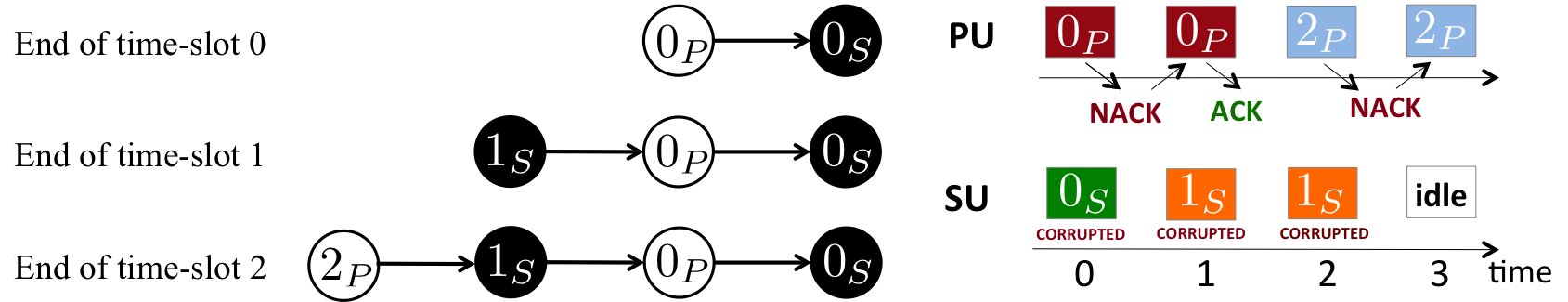}
\caption{Example of CD graph construction. The white and black nodes denote PU packets and SU packets, respectively, numbered with the corresponding label.
On the right, we show the transmission sequence for the PU and SU.}
\label{exconstruc}
\end{figure}
Now, assume PUtx retransmits $2_P$ in slot 3 while SUtx remains idle, and SUrx successfully decodes $2_P$.
The successful decoding of $2_P$ triggers CD over the graph:
in fact,
the interference of $2_P$ is removed from the signal received in slot 2 (previously buffered by SUrx),
 thus recovering $1_S$; then, the interference of $1_S$ is removed from the signal received in slot 1,
 thus recovering $0_P$; finally, the interference of $0_P$ is removed from the signal received in slot 0,
 thus recovering $0_S$.
 The SU is thus able to recover all previously failed transmissions via CD, so that the CD potential $2$ is released.
The CD outcome is thus
 obtained by following the direction of the arrows in the CD graph,
starting from node $2_P$, which initiates it.
Correspondingly, at the beginning of slot $4$ we obtain
$\mathbf A_{P\rightarrow S,4}=\mathbf 0_{3\times 3}$ ($3\times 3$ matrix of zeros)
and $\mathbf A_{S\rightarrow P,4}=\mathbf 0_{3\times 3}$,
corresponding to the untransmitted SU packets with labels $2_S$, $3_S$ and $4_S$,
and untransmitted PU packets with labels $1_P$, $3_P$ and $4_P$.

In this example, the CD graph has a linear structure.
However, the following treatment includes more general graph structures, as the one provided in Example \ref{example2} in Sec. \ref{sec:chdecprot}.
\qed
\end{example}
 
 \subsection{Reachability and root of the CD graph}
We define the \emph{reachability} between a pair of packets
 in the CD graph, and the \emph{root} of the CD graph as follows.
\begin{definition}
Consider the CD graph $\mathcal G_n=(\mathcal V_{n},\mathbf A_{n})$,
and two packets $l_1,l_2\in\mathcal V_{n}$.
  We say that $l_2$
is reachable from $l_1$
 (we write $l_1\Rightarrow l_2$) if there is a directed path
connecting the two packets in the graph, \emph{i.e.},
$\mathbf v_{X}^*(l_{1};\mathcal G_n)\mathbf e_Y(l_{2})^T=1$, $X,Y\in\{S,P\}$,
where $X=S$ (respectively, $Y=S$)  if $l_1$ ($l_2$) is a SU packet
and $X=P$ ($Y=P$) otherwise.
If $l_2$ is not reachable from $l_1$, then we write $l_1\not\Rightarrow l_2$.
More in general,  a set of packets $\tilde{\mathcal V}\subseteq\mathcal V_n$ is reachable from $l_1$ if each packet in $\tilde{\mathcal V}$ is reachable (we write $l_1\Rightarrow\tilde{\mathcal V}$).
\qed
\end{definition}
According to this definition, if $l_1\Rightarrow\tilde{\mathcal V}$
and $l_1$ is decoded, then all $l_1\in\tilde{\mathcal V}$ are recovered via CD.

\begin{definition}
We define the \emph{root} of the CD graph $\mathcal G_n$ as the SU packet with the highest CD potential, \emph{i.e.},
\begin{align}
\rho_S(\mathcal G_n)=\arg\max_{l_S\in\mathcal V_{S,n}}v_S(l_S;\mathcal G_n),
\end{align}
and its CD potential as $v_S^*(\mathcal G_n)=v_S(\rho_S(\mathcal G_n);\mathcal G_n)$.
This may not be unique; to resolve ties, we let $\rho_S(\mathcal G_n)$ be the one with the largest label value, \emph{i.e.}, the most recent SU packet with highest CD potential.
\qed
\end{definition}
\begin{remark}
The choice of the root in case of ties is arbitrary. As we will see, the selected root is occasionally retransmitted by SUtx as part of the chain decoding protocol.
Thus, our choice to select the most recent one favors the retransmission of SU packets with fresher information,
whereas older SU packets with possibly outdated information are dropped.
\end{remark}

According to this definition, we have that $v_S^*(\mathcal G_n)=1$ if and only if 
  $\rho_S(\mathcal G_n)=n$. In fact, $v_S(l_S;\mathcal G_n)=1$ implies that only $l_S$ is decoded;
since the SU packet with label $n_S$ has no edges in $\mathcal G_n$ (it has not been transmitted yet), its CD potential is $v_S(n_S;\mathcal G_n)=1$ (\emph{i.e.}, when decoded, it decodes only itself and no other packets in the CD graph),
and thus $n_S$ is the most recent packet with CD potential $1$.
 
\section{Performance metrics and optimization problem}
\label{perfmetropt}
We define the expected reward incurred by the PU when its internal state is $\mathbf s_{P,n}=\mathbf s_P$ and
the packet arrival in the $n$th slot takes value $b_{P,n}=b_P$,
as a function of the access decision of the PU ($a_{P,n}=a_P\in\{0,1\}$) and of the SU ($a_{S,n}=a_S\in\{0,1\}$),
and of the channel quality $\boldsymbol{\gamma}_{P,n}=\boldsymbol{\gamma}_P$,
as 
\begin{align}
\mathbf r_P(\mathbf s_P,b_P,\boldsymbol{\gamma}_P,a_P,a_S)
=
\left[
\begin{array}{c}
r_{P,1}(\mathbf s_P,b_P,\boldsymbol{\gamma}_P,a_P,a_S)\\
r_{P,2}(\mathbf s_P,b_P,\boldsymbol{\gamma}_P,a_P,a_S)\\
\vdots\\
r_{P,q}(\mathbf s_P,b_P,\boldsymbol{\gamma}_P,a_P,a_S)\\
\end{array}
\right]\in\mathbb R^q.
\end{align}
Note that $\mathbf r_P(\mathbf s_P,b_P,\boldsymbol{\gamma}_P,a_P,a_S)$ is a vector of rewards, in order to model multiple
performance metrics of interest. Moreover, negative rewards are used to model costs for the PU.
For instance, $\mathbf r_{P,i}(\mathbf s_P,b_P,\boldsymbol{\gamma}_P,a_P,a_S)=-P_pa_P$  models the power consumption incurred by the PU;
$\mathbf r_{P,i}(\mathbf s_P,b_P,\boldsymbol{\gamma}_P,a_P,a_S)=-\max\{q_{P}-o_{P}+b_{P}-Q_{\max},0\}$
models the number of packets dropped due to data buffer overflow, where
$o_{P}=\sigma(t_{P},d_{P},y_{P})$ from (\ref{sigma}), and $y_P$ is a function of $a_P$ and $\boldsymbol{\gamma}_P$;
$\mathbf r_{P,i}(\mathbf s_P,b_P,\boldsymbol{\gamma}_P,a_P,a_S)=a_P\chi\left(\boldsymbol{\gamma}_P\in\Gamma_P(a_{S})\right)
$ models the instantaneous throughput achieved by the PU;
$\mathbf r_{P,i}(\mathbf s_P,b_P,\boldsymbol{\gamma}_P,a_P,a_S)=-q_P$ models the queuing delay experienced by the PU packets.
Importantly, the PU reward function $\mathbf r_P(\mathbf s_P,b_P,\boldsymbol{\gamma}_P,a_P,a_S)$ 
is independent of the specific packet transmitted by the SU (\emph{i.e.}, it is independent of the SU label $l_{S,n}$),
but does depend on the SU access decision $a_S\in\{0,1\}$.
This is a practical assumption, since the PU is oblivious to the SU in our setting.

We define the average reward of the PU, under the SU access and labeling policies $\mu_S=(\mu_{S,0},\mu_{S,1},\mu_{S,2},\dots)$ and $\lambda_S=(\lambda_{S,0},\lambda_{S,1},\lambda_{S,2},\dots)$,
over a time horizon of length $N$, as 
\begin{align}
\label{perfPU}
\bar{\mathbf R}_P^N(\mu_S,\lambda_S)=\frac{1}{N}\mathbb E\left[\left.\sum_{n=0}^{N-1}
\mathbf r_P(\mathbf s_{P,n},b_{P,n},\boldsymbol{\gamma}_P(n),a_{P,n},a_{S,n})\right|
\mathbf s_{P,0}=(0,0,0)
\right],
\end{align}
where the internal PU state follows the dynamics  $\mathbf s_{P,n+1}=\phi(\mathbf s_{P,n},b_{P,n},y_{P,n})$ as in (\ref{internalstate}), and the expectation is with respect to the SNR process $\{(\boldsymbol{\gamma}_P(n),\boldsymbol{\gamma}_S(n)),\ n=0,1,\dots,N-1\}$,
the decision of the SU to transmit or remain idle, drawn according to policy $\mu_{S,n}(y_{P,0}^{n-1},y_{S,0}^{n-1},l_{S,0}^{n-1})$,
the SU labeling sequence, drawn according to policy $\lambda_{S,n}(\cdot|y_{P,0}^{n-1},y_{S,0}^{n-1},l_{S,0}^{n-1})$,
and the PU access decision, drawn according to $\mu_P(\mathbf s_{P,n})$.

Similarly, we define the average throughput of the SU,
over a time horizon of length $N$, as 
\begin{align}
\label{perfSU}
\bar T_S^N(\mu_S,\lambda_S)=\frac{1}{N}\mathbb E\left[\left.\sum_{n=0}^{N-1}
r_{S,n}\right|
\mathbf s_{P,0}=(0,0,0)
\right],
\end{align}
where $r_{S,n}$ is the instantaneous expected throughput, defined in Sec. \ref{secdemecoioni}.

In this paper, we focus on the average long-term performance $N\to\infty$, so that (\ref{perfPU}) and (\ref{perfSU}) become
\begin{align}
&\bar{\mathbf R}_P(\mu_S,\lambda_S)\triangleq\lim\inf_{N\to\infty}\bar{\mathbf R}_P^N(\mu_S,\lambda_S),
\nonumber\\&
\bar T_S(\mu_S,\lambda_S)\triangleq\lim\inf_{N\to\infty}\bar T_S^N(\mu_S,\lambda_S).
\end{align}

The goal of the SU is to define a secondary access policy $\mu_S$,
which determines \emph{whether} the SU should access the channel or remain idle at any given time,
 and a labeling policy $\lambda_S$,
which determines \emph{what} the SU should transmit (new data packet
or retransmission of a specific previously failed and buffered SU packet),
so as to maximize the average long-term SU throughput $\bar T_S(\mu_S,\lambda_S)$,
subject to a constraint on the 
minimum average long-term reward $\bar{\mathbf R}_P(\mu_S,\lambda_S)$ incurred by the PU,
\emph{i.e.},
\begin{align}\label{opt}
\mathbf{P1:}\ (\mu_S^*,\lambda_S^*)=\arg\max_{\mu_S,\lambda_S} \bar T_S(\mu_S,\lambda_S)\ \text{s.t.\ }
\bar{\mathbf R}_P(\mu_S,\lambda_S)
\geq
\bar{\mathbf R}_{P,\min}.
\end{align}
Herein, we assume that the reward for the PU is maximized if the SU remains idle, \emph{i.e.}, 
letting $\mu_S=\emptyset$ be the idle SU policy $\mu_{S,n}(\cdot)=0,\ \forall n$,
we have that
$\bar{\mathbf R}_P(\mu_S,\lambda_S)\leq\bar{\mathbf R}_P(\emptyset,\lambda_S),\ \forall \mu_S$.
Then, the optimization problem (\ref{opt}) is feasible if and only if $\bar{\mathbf R}_P(\emptyset,\lambda_S)
\geq
\bar{\mathbf R}_{P,\min}$.

\begin{remark}
Importantly, the average long-term performance for the PU, $\bar{\mathbf R}_P(\mu_S,\lambda_S)$, is a function of the policy implemented by the SU only through
the access scheme $\mu_S$, since the instantaneous expected reward $\mathbf r_P(\mathbf s_P,b_P,\boldsymbol{\gamma}_P,a_P,a_S)$ is independent of the SU packet label $l_S$.
 Therefore, if two labeling policies $\lambda_S^\prime$ and $\lambda_S^{\prime\prime}$ generate the same access sequence
$\{a_{S,n},n\geq 0\}$, the performance for the PU will be the same, \emph{i.e.}, $\bar{\mathbf R}_P(\mu_S,\lambda_S^\prime)=\bar{\mathbf R}_P(\mu_S,\lambda_S^{\prime\prime})$.
\end{remark}

Note that the state space of the system 
may be infinitely large,
since the CD graph may grow arbitrarily
large, and the optimal policy may depend on the specific CD graph available in each slot,
thus challenging the numerical optimization of {\bf P1}.
In the next section, we present the CD protocol and prove its optimality.
Such protocol specifies, at any given time, whether
the SU should transmit a new data packet
or perform a retransmission of a \emph{specific} SU packet in the CD graph, and thus explicitly characterizes the labeling policy $\lambda_S$ of the SU.
It is based on four basic rules,
stated in Sec. \ref{sec:chdecprot}.
In Sec. \ref{compactstsp}, we will show that such optimal labeling policy makes it possible to define a compact state space representation
of the system,
which takes into account only some features of the CD graph in the decision process, rather than the complete structure of the CD graph.
This compact representation
 lends itself to an efficient optimization of the SU access policy $\mu_S^*$
via a Markov decision process formulation.

As a result, the SU access policy
and the CD protocol are decoupled:
the former specifies whether the SU should access the channel or remain idle,
depending on the state of the system in the compact state space representation;
the latter, should the SU decide to access the channel,
specifies which SU packet needs to be transmitted according to  four CD rules.

\section{Chain Decoding Protocol}\label{sec:chdecprot}
Let $\mathcal G_n$ be the CD graph at the beginning of slot $n$, and $l_{P,n}$ be the 
label of the PU packet transmitted in slot $n$.
Note  that the SU does not know whether the PU transmits or remains idle in slot $n$, due to the randomized PU
access policy. However, if a PU transmission occurs in slot $n$, then the SU pair knows the corresponding label $l_{P,n}$ (see Lemma \ref{lem1}),
\emph{i.e.}, whether PUtx is about to perform a new transmission (if $t_{P,n}=0$), or a retransmission (if $t_{P,n}>0$).

The CD protocol defines which packet the SU should transmit at any given time, in those slots where $a_{S,n}=1$, and is defined by the following four rules:
\begin{description}
 \item[{\bf R1})] If $\kappa_{P,n}(l_{P,n})=0$
  (the current PU packet is unknown by SUrx),
  $\rho_S(\mathcal G_n)\not\Rightarrow l_{P,n}$ and $l_{P,n}\not\Rightarrow\rho_S(\mathcal G_n)$,
  so that $l_{P,n}$ cannot be reached from the root $\rho_S(\mathcal G_n)$ in the CD graph, and vice versa,
 then $l_{S,n}=\rho_S(\mathcal G_n)$, \emph{i.e.}, the root of $\mathcal G_n$ is transmitted;\footnote{Note that this implies that, if $v_S^*(\mathcal G_n)=1$, then $l_{S,n}=\rho_S(\mathcal G_n)=n_S$, so that the SU transmits a new packet.}
\item[{\bf R2})]
If $\kappa_{P,n}(l_{P,n})=0$
  (the current PU packet is unknown by SUrx)
  and $\rho_S(\mathcal G_n)\Rightarrow l_{P,n}$,
   or $\kappa_{P,n}(l_{P,n})=0$ and $l_{P,n}\Rightarrow\rho_S(\mathcal G_n)$,
     so that $l_{P,n}$ can be reached from the root $\rho_S(\mathcal G_n)$ in the CD graph, or vice versa,
 then $l_{S,n}=n_S$ (a new SU packet is transmitted);
\item[{\bf R3})] 
If $\kappa_{P,n}(l_{P,n})=1$  (the current PU packet is known by SUrx),
 then $l_{S,n}=\rho_S(\mathcal G_n)$, \emph{i.e.}, the root of $\mathcal G_n$ is transmitted;
\item[{\bf R4})] 
upon starting a new ARQ cycle ($t_{P,n}=0$), 
the portion of the graph reachable from $\rho_S(\mathcal G_n)$ is retained;
whereas the remaining portion of the graph is discarded.
\end{description}
\begin{remark}
Note that, according to {\bf R4}, SU packets may be discarded at the end of the slot,
and thus reliability is not guaranteed.
However, reliability can still be achieved by higher layer protocols, \emph{i.e.}, 
by forcing a retransmission at the upper layer.
The lower levels of the protocol considered in this paper are oblivious to the retransmission process
enforced at the upper levels, and thus, this information is not exploited for IC.
\end{remark}

The aim of rule {\bf R1} is to connect the current PU packet with label $l_{P,n}$
to the graph, in order to build CD potential. In particular, if 
$l_{P,n}\leftrightarrow\rho_S(\mathcal G_n)$ or $l_{P,n}\rightarrow\rho_S(\mathcal G_n)$, then $l_{P,n}$ \emph{inherits} the CD potential of the root $\rho_S(\mathcal G_n)$;
on the other hand, if 
$l_{P,n}\leftrightarrow\rho_S(\mathcal G_n)$ or $l_{P,n}\leftarrow\rho_S(\mathcal G_n)$, then $l_{P,n}$ augments the
CD potential of  the root $\rho_S(\mathcal G_n)$, by making the nodes directly reachable from  $l_{P,n}$ reachable from $\rho_S(\mathcal G_n)$ as well.
The aim of rule {\bf R2} is to build the CD graph and increase the CD potential of $l_{P,n}$,
by connecting new SU packets to the current PU packet in the graph.
The aim of rule {\bf R3}  is to release the CD potential
and deliver secondary throughput, respectively,
by taking advantage of the knowledge of the current PU packet at SUrx.
The aim of rule {\bf R4} is to retain the portion of the graph with the
largest CD potential, while dismissing those packets which cannot be recovered via CD.

Notice that the four CD rules instruct SUtx to either transmit a new packet, with label $l_{S,n}=n_S$, 
or retransmit the root of the CD graph $\mathcal G_n$, with label $l_{S,n}=\rho_S(\mathcal G_n)$.
No other SU packets may be transmitted at any time. The intuition behind this result is that,
if the root is successfully decoded, then the highest CD potential is released, leading to the largest number of SU packets being decoded in the CD graph.
In contrast, if any other packet in the graph is retransmitted, a lower CD potential is released,
yielding lower transmission efficiency.
 We remark that, while the PU uses retransmissions as part of the ARQ mechanism to improve reliability,
the SU does not use retransmissions to improve reliability but to
build the CD graph
and release the CD potential,
in order to achieve the largest SU throughput possible.

\begin{figure}[t]
    \centering
    \subfigure[]
    {
\includegraphics[width=.22\linewidth,trim = 1mm 1mm 1mm 1mm,clip=true,frame]{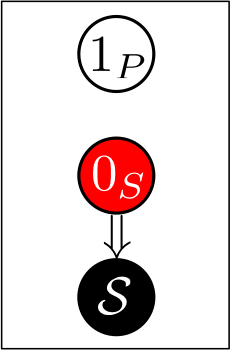}
}
    \subfigure[]
    {
\includegraphics[width=.22\linewidth,trim = 1mm 1mm 1mm 1mm,clip=true,frame]{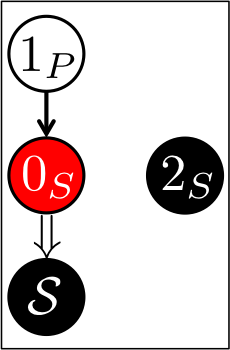}
}
\subfigure[]
    {
\includegraphics[width=.22\linewidth,trim = 1mm 1mm 1mm 1mm,clip=true,frame]{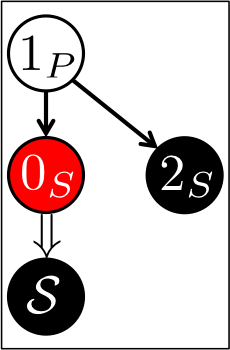}
}
\subfigure[]
    {
\includegraphics[width=.22\linewidth,trim = 1mm 1mm 1mm 1mm,clip=true,frame]{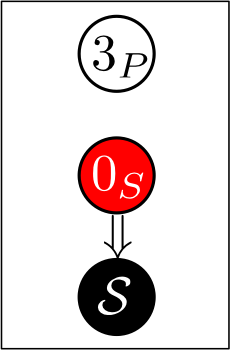}
}
\subfigure[]
    {
\includegraphics[width=.22\linewidth,trim = 1mm 1mm 1mm 1mm,clip=true,frame]{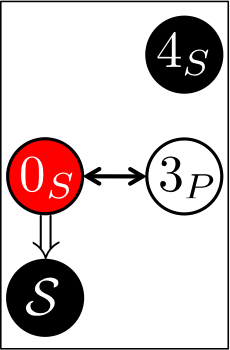}
}
    \subfigure[]
    {
\includegraphics[width=.22\linewidth,trim = 1mm 1mm 1mm 1mm,clip=true,frame]{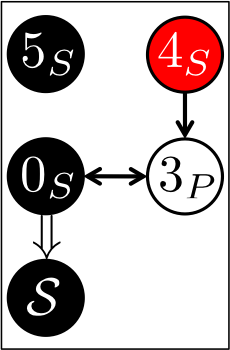}
}
    \subfigure[]
{
\includegraphics[width=.22\linewidth,trim = 1mm 1mm 1mm 1mm,clip=true,frame]{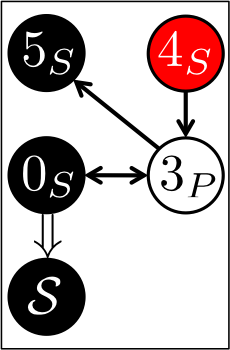}
}
\subfigure[]
    {
\includegraphics[width=.22\linewidth,trim = 1mm 1mm 1mm 1mm,clip=true,frame]{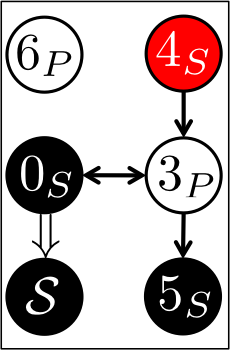}
}
\caption{Example of CD protocol and construction of CD graph. The white, black, and red nodes denote PU packets, SU packets,
and the root of the CD graph, respectively, numbered with the corresponding label.
\\
\small
(a) slot 1: new ARQ cycle; PU transmits $1_P$, SU transmits the root $0_S$ ({\bf R1})
\\
(b) slot 2: PU transmits $1_P$, SU transmits a new packet $2_S$ ({\bf R2})
\\
(c) end of slot 2
\\
(d) slot 3: new ARQ cycle, $1_P$ and $2_S$ are dropped from the graph ({\bf R4});
PU transmits $3_P$, SU transmits $0_S$ ({\bf R1})
\\
(e) slot 4: PU transmits $3_P$, SU transmits $4_S$ ({\bf R2})
\\
(f) slot 5: PU transmits $3_P$, SU transmits $5_S$ ({\bf R2})
\\
(g) end of slot 5
\\ (h) slot 6: new ARQ cycle; PU transmits $6_P$, SU transmits the root $4_S$ ({\bf R1})
}
\label{fig:example}
\end{figure}

Whether $l_{S,n}=n_S$ or $l_{S,n}=\rho_S(\mathcal G_n)$ is a function of 
 $\kappa_{P,n}(l_{P,n})\in\{0,1\}$,
 of $\chi(\rho_S(\mathcal G_n)\Rightarrow l_{P,n})$ and of $\chi(l_{P,n}\Rightarrow\rho_S(\mathcal G_n))$,
 \emph{i.e.}, whether or not the current PU packet is known by SUrx, and whether or not the current PU packet is reachable from the root of the CD graph, and vice versa.
Therefore,
 the labeling policy $\lambda_S$ can be expressed as
$\lambda_S(\kappa_{P,n}(l_{P,n}),\chi(\rho_S(\mathcal G_n)\Rightarrow l_{P,n}),\chi(l_{P,n}\Rightarrow\rho_S(\mathcal G_n)))$, rather than $\lambda_{S,n}(\cdot|y_{P,0}^{n-1},y_{S,0}^{n-1},a_{S,0}^{n-1},l_{S,0}^{n-1})$.

We now describe the application of these rules with the following example,
depicted in Fig.~\ref{fig:example}.
\begin{example}\label{example2}
 The structure of the CD graph at the beginning of the ARQ cycle, in slot $1$, is depicted 
in Fig.~\ref{fig:example}.a, where 
 $0_S$ is the root of the graph and $\mathcal S$ is the set of SU packets
  reachable from the root $0_S$, so that $0_S\Rightarrow\mathcal S$ and  $v_S(0_S;\mathcal G_1)=1+|\mathcal S|$.
 The current PU packet $1_P$ is unknown by SUrx and is not connected to the graph,
hence, according to rule {\bf R1}, the SU retransmits the root of the graph $0_S$ in slot $1$.
Assume that the decoding outcome at SUrx is such that $1_P\rightarrow 0_S$.
$1_P$ thus becomes
connected to the CD graph, as in Fig.~\ref{fig:example}.b,
and its CD potential is inherited by $0_S$, so that $v_P(1_P;\mathcal G_2)=v_S(0_S;\mathcal G_1)$.
 Hence, according to rule {\bf R2},
 in the next slot $2$
the SU transmits a new data packet $2_S$.
PUtx fails its transmission in slot $1$, hence it retransmits $1_P$ in slot $2$.
 Assume that 
the decoding outcome is such that $1_P\rightarrow 2_S$,
so that  $2_S$ becomes connected to the
CD graph, as depicted in Fig.~\ref{fig:example}.c.
Assume also that PUrx successfully decodes $1_P$, so that a new ARQ cycle begins in slot $3$.
Note that, at the end of slot $2$, $0_S$ has
 the highest CD potential (Fig.~\ref{fig:example}.c).
In fact, by initiating the CD process from $0_S$, all the SU packets
in $\mathcal S$ are recovered; on the other hand, no CD can be initiated from $2_S$.
Therefore, applying rule {\bf R4}, nodes $1_P$ and $2_S$ are trimmed from the CD graph, whose structure in the next slot $3$ is as depicted 
 in Fig.~\ref{fig:example}.d. 
  In fact, $1_P$ is no longer retransmitted by PUtx, and thus cannot be decoded by SUrx in the future, and $2_S$ cannot initiate the CD process, since it is a leaf in the graph.
In slot $3$,
SU transmits $0_S$ and PU transmits $3_P$, according to rule {\bf R1}.
Assume that $0_S\leftrightarrow 3_P$.
Then, the structure of the CD graph in slot $4$ is as depicted
in Fig.~\ref{fig:example}.e.
Now, $3_P$ is connected to the root of the graph, hence, according to rule {\bf R2}
and assuming a PU retransmission is requested,
the SU transmits a new packet $4_S$ and PUtx retransmits $3_P$.
Assume that $4_S\rightarrow 3_P$.
Then, the structure of the CD graph at the
beginning of slot $5$ is as depicted in Fig.~\ref{fig:example}.f.
According to rule {\bf R2}, in slot $5$ the SU transmits $5_S$ and the PU retransmits $3_P$.
Assume that $3_P\rightarrow 5_S$ and the PU successfully decodes $3_P$.
The structure of the CD graph at the end of slot $5$ is depicted in
Fig.~\ref{fig:example}.g.
Note that, at this point, the SU packet with the highest CD
potential is $4_S$. In fact, if the CD process is initiated from $4_S$,
in sequence,
 $3_P$, $5_S$, $0_S$ and the set of SU packets $\mathcal S$ are decoded,
and thus its CD
potential is $v_S(4_S;\mathcal G_n)=|\mathcal S|+3$.
In contrast, if the CD process were initiated from $0_S$, then
only the SU packets in $\mathcal S$ and $5_S$ would be decoded, and thus its CD potential is $v_S(0_S;\mathcal G_n)=|\mathcal S|+2$.
Therefore, according to rule {\bf R4}, in the new ARQ round
$4_S$ becomes the root of the CD graph, as depicted in Fig.~\ref{fig:example}.h.
\qed
\end{example}

The following theorem establishes the optimality of the CD protocol.
\begin{thm}
\label{CD}
The CD protocol defines one optimal labeling policy $\lambda_S^*$ solving the optimization problem $\mathbf{P1}$ under any SU access policy $\mu_S$.
\end{thm}
\begin{proof}
See Appendix \ref{proofofCD}.
\end{proof}

Since Theorem \ref{CD} proves the optimality of the CD rules, we can assume that $\lambda_S^*$ is generated according to these rules.
We denote the corresponding labeling policy  as $\lambda_S^{(CD)}$. Therefore, the original optimization problem $\mathbf{P1}$
in (\ref{opt}) can be restated as 
\begin{align}\label{opt2}
\mathbf{P2:}\ \mu_S^*=\arg\max_{\mu_S} \bar T_S(\mu_S,\lambda_S^{(CD)})\ \text{s.t.\ }
\bar{\mathbf R}_P(\mu_S,\lambda_S^{(CD)})
\geq
\bar{\mathbf R}_{P,\min},
\end{align}
so that only the SU access policy $\mu_S$ needs to be optimized.
Under the CD labeling policy, it can be proved that the SU throughput $\bar T_S(\mu_S,\lambda_S^{(CD)})$ 
achieves an upper bound $\bar T_S^{(up)}(\mu_S)$.
This is stated in the following theorem, which follows as a corollary of the proof of Theorem \ref{CD} (see Appendix \ref{proofofCD}).

The upper bound $\bar T_S^{(up)}(\mu_S)$ is composed of three components.
The first component, $\bar T_{S,j}^{(GA)}(\mu_S)$, is the genie-aided SU throughput, assuming that the PU packets are known in advance and their interference can be removed.
The second term (\ref{2ndterm}) is a throughput degradation term which accounts for the case when
SUrx cannot decode the PU packet within the PU retransmission cycle,
even in the genie-aided case where the packet is decoded after removing the interference from the SU packets;
in this case, such PU packet cannot be decoded, its interference cannot be removed, hence 
the only way for SUrx to decode SU packets is to treat the PU signal as noise, resulting in the impossibility to decode those SU packets such that $y_{S,n}\in\{5,7\}$.
The third term  (\ref{3rdterm}) is a throughput degradation term which accounts for the
case when the SU needs to retransmit the root of the CD graph, rather than transmitting new data packets;
with this last term, we guarantee that the root of the CD graph is counted only once in the throughput accrual.

\begin{thm}
\label{thm4}
Under the labeling policy defined by the CD rules, $\lambda_S^{(CD)}$, we have
\begin{align}
\bar T_S(\mu_S,\lambda_S^{(CD)})=
\bar T_{S}^{(up)}(\mu_S),
\end{align}
where 
\begin{align}
\label{1stterm}
&\bar T_S^{(up)}(\mu_S)
=
\lim\inf_{j\to\infty}
\left\{
\vphantom{-\mathbb E\left[\frac{1}{\nu_P(j+1)}\sum_{k=0}^{j}
\left(1-\kappa_{P,k}^{(GA)}\right)
\sum_{n=\nu_P(k)}^{\nu_P(k+1)-1}a_{P,n}a_{S,n}\chi(y_{S,n}\in\{5,7\})\right]
}
\bar T_{S,j}^{(GA)}(\mu_S)\right.
\\&\label{2ndterm}
-\mathbb E\left[\frac{1}{\nu_P(j+1)}\sum_{k=0}^{j}
\left(1-\kappa_{P,k}^{(GA)}\right)
\sum_{n=\nu_P(k)}^{\nu_P(k+1)-1}a_{P,n}a_{S,n}\chi(y_{S,n}\in\{5,7\})\right]
\\&\label{3rdterm}
-\mathbb E\left[\frac{1}{\nu_P(j+1)}\sum_{k=0}^{j}\left(1-\prod_{n=\nu_P(k)}^{\nu_P(k+1)-1}\left[1-a_{P,n}a_{S,n}\chi(y_{S,n}=7)\right]\right)
\right.
\\&\nonumber
\left.\left.\times
\prod_{n=\nu_P(k)}^{\nu_P(k+1)-1}\left(
\vphantom{\prod_{n=\nu_P(k)}^{\nu_P(k+1)-1}}
1
- a_{P,n}\chi(y_{S,n}\in\{1,3,6,7\})
+a_{P,n}a_{S,n}\chi(y_{S,n}=7)
\right)\right]\right\}.
\end{align}
Above, 
\begin{align}
\bar T_{S,j}^{(GA)}(\mu_S)&=
\mathbb E\left[\frac{1}{\nu_P(j+1)}\sum_{n=0}^{\nu_P(j+1)-1}a_{S,n}\chi(y_{S,n}\in\{1,2,5,7\})\right]
\nonumber\\&
=
\mathbb E\left[\frac{1}{\nu_P(j+1)}\sum_{n=0}^{\nu_P(j+1)-1}a_{S,n}\right]\left(\delta_{sp}+\delta_{s}+\upsilon_{s}+\upsilon_{sp}\right)
\end{align}
is the genie-aided (GA) throughput accrued over the first $j+1$ PU ARQ cycles, assuming SUrx knows the PU packet in advance and thus removes its interference,
and
\begin{align}
\label{kappaga2}
&\kappa_{P,k}^{(GA)}
\triangleq
1-\prod_{n=\nu_P(k)}^{\nu_P(k+1)-1}\left[1-a_{P,n}\chi(y_{S,n}\in\{1,3,6,7\})\right]
    \end{align}
    is the genie-aided (GA)
decoding outcome  at SUrx for the PU packet transmitted in the $k$th ARQ cycle,
 assuming that the SU packets are known in advance and thus their interference can be removed.
\end{thm}
\begin{proof}
From the proof of Theorem \ref{CD} in Appendix \ref{proofofCD},
under Assumption \ref{assumptionalltx}, \emph{i.e.}, assuming that both PUtx and SUtx always transmit,
the CD labeling policy asymptotically achieves the upper bound 
$T_{S,\infty}^{(up)}\triangleq \lim_{j\to\infty}
T_{S,j}^{(up)}$, as given by (\ref{ccc}),
where $ T_{S,j}^{(up)}$ is defined in (\ref{tsup}).
This result is defined under  Assumption \ref{assumptionalltx}.
In order to map it to the general case where either PUtx or SUtx may remain idle,
we  apply Procedure~\ref{lemreplace} to the expression of $T_{S,j}^{(up)}$ in (\ref{tsup}),
thus yielding the expression given by (\ref{1stterm}) after taking the expectation
with respect to the SU and PU access policies and the SNR process.
\end{proof}

Using Theorem \ref{thm4}, the optimization problem $\mathbf{P2}$ can then be expressed as
\begin{align}\label{opt3}
\mathbf{P3:}\ \mu_S^*=\arg\max_{\mu_S} \bar T_{S}^{(up)}(\mu_S)\ \text{s.t.\ }
\bar{\mathbf R}_P(\mu_S,\lambda_S^{(CD)})
\geq
\bar{\mathbf R}_{P,\min},
\end{align}
which is developed in the next section.

\section{Compact State Space representation}\label{compactstsp}
In this section, 
we investigate the solution of the optimization problem $\mathbf{P3}$, and we show that it yields a compact state space representation of the CD protocol,
and thus can be solved efficiently via dynamic programming.
To this end, in the following theorem we derive an alternative expression of the SU throughput under the CD labeling policy, $\bar T_{S}^{(up)}$.

This alternative expression expresses the SU throughput as a long-term time average of 
a \emph{virtual instantaneous throughput},
which not only counts the SU packets physically decoded in each slot, but also the CD potential, \emph{i.e.}, those SU packets that are reachable from the root of the CD graph.
Intuitively, since the root of the CD graph is transmitted infinitely often by SUtx as part of the CD rules, it will be decoded with probability one within finite time, thus releasing the full CD potential.
From the 
perspective of the SU throughput, there is no difference between counting such CD potential
as soon as it is created, rather than at the future time when the root is decoded.

As it will be seen in Theorem \ref{thm5}, the virtual instantaneous throughput is expressed as the sum of five quantities:
the genie-aided throughput, assuming that the interference from the PU packet can be removed (see (\ref{g}));
a throughput degradation term due to the fact that the current PU packet may be unknown and thus its interference cannot be removed  (see (\ref{g2}));
a throughput term due to the fact that,
if the PU packet is decoded or it becomes reachable from the root in the CD graph ($y_{S,n}\in\{1,3,6,7\}$), then those $b_{S,n}$ SU packets that
can be reached from the PU packet are \emph{virtually} decoded (see (\ref{g3}));
a throughput term due to the fact that, if the current PU packet is connected to the root of the CD graph as 
$l_{P,n}\leftrightarrow\rho_S(\mathcal G_n)$ ($\iota_{P,n}=\hat\kappa_{P,n}^{(GA)}=1$) and the PU packet is physically decoded, then the root is decoded as well  (see (\ref{g4}));
finally, a throughput term due to the fact that, if the current PU packet is connected to the root of the CD graph as 
$l_{P,n}\leftrightarrow\rho_S(\mathcal G_n)$ and the new transmissions in slot $n$ are such that 
$l_{S,n}\rightarrow l_{P,n}$, then $l_{S,n}$ becomes the new root of the graph and the previous root $\rho_S(\mathcal G_n)$ is \emph{virtually} decoded  (see (\ref{g5})).

\begin{thm}
\label{thm5}
\begin{align}
\bar T_{S}^{(up)}=\lim\inf_{N\to\infty}\mathbb E\left[\frac{1}{N}\sum_{n=0}^{N-1}g(a_{S,n},a_{P,n},y_{S,n},\hat\kappa_{P,n}^{(GA)},\iota_{P,n},b_{S,n})\right],
\end{align}
where $g(\cdot)$ is the \emph{virtual instantaneous throughput}  for the SU pair, defined as
\begin{align}
\label{g}
&g(a_{S,n},a_{P,n},y_{S,n},\hat\kappa_{P,n}^{(GA)},\iota_{P,n},b_{S,n})
=
a_{S,n}\chi(y_{S,n}\in\{1,2,5,7\})
\\&
\label{g2}
- \left(1-\hat\kappa_{P,n}^{(GA)}\right)a_{S,n}a_{P,n}\chi(y_{S,n}\in\{5,7\})
\\&
\label{g3}
+a_{P,n}\chi(y_{S,n}\in\{1,3,6,7\})b_{S,n}
\\&
\label{g4}
+\iota_{P,n}\hat\kappa_{P,n}^{(GA)}a_{P,n}[(1-a_{S,n})\chi(y_{S,n}\in\{1,3,6,7\})+a_{S,n}\chi(y_{S,n}\in\{1,3\})]
\\&
\label{g5}
+\iota_{P,n}\hat\kappa_{P,n}^{(GA)}a_{P,n}a_{S,n}\chi(y_{S,n}=6),
\end{align}
where we have defined
\begin{align}
\label{kappapnga}
&\hat\kappa_{P,n}^{(GA)}\triangleq 1-\prod_{m=\nu_P(k)}^{n-1}\left[1-a_{P,m}\chi(y_{S,m}\in\{1,3,6,7\})\right],
\\&
\label{iota}
\iota_{P,n}\triangleq
\prod_{m=\nu_P(k)}^{n-1}\left[1-a_{P,m}\chi(y_{S,m}\in\{1,3,6,7\})+a_{P,m}a_{S,m}\chi(y_{S,m}=7)\right],
\\&
\label{bsn}
b_{S,n}\triangleq\left(1-\hat\kappa_{P,n}^{(GA)}\right)\sum_{m=\nu_P(k)}^{n-1}a_{P,m}a_{S,m}\chi(y_{S,m}=5).
    \end{align}
\end{thm}
\begin{proof}
Using Lemma \ref{lemn} in the Appendix, we can express $\bar T_{S}^{(up)}$ as
\begin{align}
\bar T_{S}^{(up)}=\lim\inf_{j\to\infty}\mathbb E\left[\frac{M_{S,\nu_P(j+1)}+v_{S,\nu_P(j+1)}-1}{\nu_P(j+1)}\right],
\end{align}
where $M_{S,n}$, defined in (\ref{Msn}), is the number of SU packets successfully decoded up to the beginning of slot $n$,
and  $v_{S,n}$, defined in (\ref{vsn}),
is the CD potential of the root of the CD graph at the beginning of slot $n$. $M_{S,n}+v_{S,n}$ is defined recursively via (\ref{tobeagreed}) in Lemma \ref{lemmamecoioni},
so that we can interpret $M_{S,n+1}+v_{S,n+1}-M_{S,n}-v_{S,n}$ as the throughput reward accrued in slot $n$,
thus yielding the expression of $g(\cdot)$ in (\ref{g}) after
 applying Procedure~\ref{lemreplace} in order to map
 the expression of $M_{S,n+1}+v_{S,n+1}-M_{S,n}-v_{S,n}$, derived under   Assumption \ref{assumptionalltx},
 to the general case where either PUtx or SUtx may remain idle.
 \end{proof}

In Theorem \ref{thm5}, similarly to  (\ref{kappaga2}), $\hat\kappa_{P,n}^{(GA)}$
is the genie-aided (GA)
decoding outcome  at SUrx for the PU packet
$l_{P,n}=\nu_P(k)$ up to slot $n$, assuming that the SU packets are known in advance and thus their interference can be removed;
$\iota_{P,n}=1$
denotes the event that either the PU packet has not been decoded yet in slot $n$ in the genie-aided case, \emph{i.e.}, even after removing the interference from the SU packets (and thus $\hat\kappa_{P,n}^{(GA)}=0$),
or it is connected to the root of the CD graph as $l_{P,n}\leftrightarrow\rho_S(\mathcal G_n)$ (\emph{i.e.}, $\exists \nu_P(k)\leq m<n:a_{P,m}=1,a_{S,m}=1,y_{S,m}=7$);
finally, $b_{S,n}$
denotes the number of SU packets that can be decoded after removing the interference from the SU packet, under the assumption that
$\hat\kappa_{P,n}^{(GA)}=0$ (the PU packet is unknown in the genie-aided case).
Note that $1-\hat\kappa_{P,n}^{(GA)}\leq\iota_{P,n}$, and therefore, if $\hat\kappa_{P,n}^{(GA)}=0$, necessarily $\iota_{P,n}=1$.
It follows that three configurations are possible: 
$\hat\kappa_{P,n}^{(GA)}=0$, $\iota_{P,n}=1$ and $b_{S,n}\geq 0$;
$\hat\kappa_{P,n}^{(GA)}=1$, $\iota_{P,n}=1$ and $b_{S,n}=0$;
$\hat\kappa_{P,n}^{(GA)}=1$, $\iota_{P,n}=0$ and $b_{S,n}=0$.

As is apparent from the proof of Theorem \ref{thm5},
the reward function $g(\cdot)$ includes two quantities: 
the quantity $M_{S,n+1}-M_{S,n}$ representing the number of SU packets successfully decoded in slot $n$, due to a direct decoding operation or via CD;
and the quality $v_{S,n+1}-v_{S,n}$, representing the increase (or decrease, if negative) of the CD potential of the root of the CD graph.
Therefore, the inclusion of $v_{S,n+1}-v_{S,n}$ in the instantaneous throughput accrual presumes that the CD potential of the CD graph is released immediately, rather than
at the future time when the root is actually decoded.

Thus, Theorem \ref{thm5} defines a \emph{virtual system} where the CD potential is immediately released, \emph{i.e.},
all the SU and PU packets which are reachable from the root of the CD graph
are decoded \emph{virtually instantaneously} by SUrx,
rather than in the future slot when the actual successful decoding of the root
occurs.
In particular, the SU packets in the CD graph (except the root itself, which needs to be retransmitted)
contribute instantaneously
to the SU throughput accrual.
Moreover, if the current PU packet can be reached from the root of the graph (and thus $\hat\kappa_{P,n}^{(GA)}=1$),
it is virtually instantaneously decoded, hence, it is \emph{virtually known} by SUrx. Such virtual knowledge
can in turn be exploited in the following primary ARQ retransmission 
attempts to create a "clean" channel by using \emph{virtual FIC}.
As a result,
we can represent the state of the CD protocol 
by the tuple $(\Phi,b_S)$, corresponding to different configurations of $(\hat\kappa_{P,n}^{(GA)},\iota_{P,n},b_{S,n})$,
where
\begin{itemize}
 \item $\Phi$ is the \emph{virtual} knowledge of the current PU packet (with label $l_P$) at SUrx,
 and takes values $\Phi\in\{\stackrel{\leftrightarrow}{\K},\stackrel{\rightarrow}{\K},\U\}$,
 where "U" denotes \emph{$l_P$ virtually unknown} ($\hat\kappa_{P,n}^{(GA)}=0$) and  "K" denotes \emph{$l_P$ virtually known} ($\hat\kappa_{P,n}^{(GA)}=1$).
 The unidirectional or bidirectional arrow above "K" indicates the 
 type of edge connecting $l_P$ to the root of the CD graph.
 In particular, $\Phi=\stackrel{\rightarrow}{\K}$ indicates that
 $l_P$ is decodable after decoding the root of the CD graph, but the root is not decodable after decoding $l_P$, as in Fig. \ref{fig:Virtconf}.c, or that $l_P$ is known ($\kappa_{P,n}(l_{P})=1$);
 on the other hand, $\Phi=\stackrel{\leftrightarrow}{\K}$ indicates that
 $l_P$ and the root are mutually decodable 
 after decoding the other packet, \emph{i.e.}, $l_P\leftrightarrow \rho_S(\mathcal G_n)$, as in Fig. \ref{fig:Virtconf}.b.
Finally, $\Phi=\U$ indicates that $l_P$ is virtually unknown to SUrx, \emph{i.e.}, it is not connected to 
the CD graph or it has not been \emph{virtually} decoded by SUrx yet;
 note that $\Phi=\U$ includes the case where the root of the CD graph is decodable after decoding $l_P$, but $l_P$ is not decodable after decoding the root (otherwise, it would be virtually known and $\hat\kappa_{P,n}^{(GA)}=1$);
\item $b_S$ is the number of SU packets 
directly reachable from the current PU packet
in the CD graph but not \emph{virtually} decoded, since $l_P$ is virtually unknown; in particular, $b_S\in\mathbb N(0,R_{\max}-1)$, since each PU packet is transmitted at most $R_{\max}$ times.
Therefore, if the current PU packet is successfully (or virtually) decoded, all the $b_S$ SU packets are decoded as well.
Note that, if $\Phi=\stackrel{\leftrightarrow}{\K}$ or $\Phi=\stackrel{\rightarrow}{\K}$,
then the current PU packet  is virtually known by SUrx, hence all the SU packets reachable from it in the CD graph 
are virtually decoded, resulting in $b_S=0$.
\end{itemize}

\begin{figure}[t]
    \centering
        \subfigure[$(\U,b)$]
    {
\includegraphics[width=.22\linewidth,trim = 1mm 1mm 1mm 1mm,clip=true]{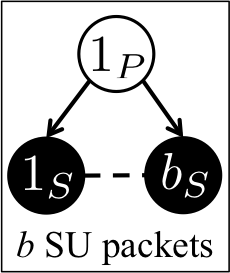}
}    
        \subfigure[$(\stackrel{\leftrightarrow}{\K},0)$]
    {
\includegraphics[width=.22\linewidth,trim = 1mm 1mm 1mm 1mm,clip=true]{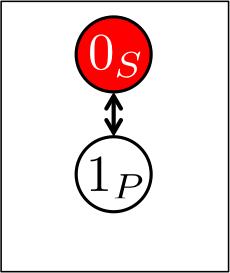}
}    
        \subfigure[$(\stackrel{\rightarrow}{\K},0)$]
    {
\includegraphics[width=.22\linewidth,trim = 1mm 1mm 1mm 1mm,clip=true]{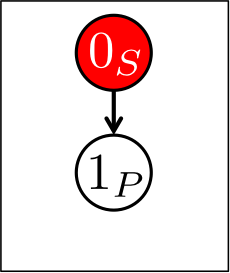}
}    
\caption{States of the virtual system}
\label{fig:Virtconf}
\end{figure}

In these different configurations, the expected virtual instantaneous throughput\footnote{We redefine the virtual instantaneous throughput
as $g(a_{S},a_{P},y_{S},\Phi,b)$, in order to reflect one of these possible configurations.} is as follows:
\begin{itemize}
\item If $(\Phi_n,b_{S,n})=(\U,b_{S,n})$, then
$\hat\kappa_{P,n}^{(GA)}=0$, $\iota_{P,n}=1$ and $b_{S,n}\geq 0$, 
hence, from (\ref{g}),
\begin{align}
g(a_{S,n},a_{P,n},y_{S,n},\U,b_{S,n})
=&
a_{S,n}\chi(y_{S,n}\in\{1,2\})
+a_{S,n}(1-a_{P,n})\chi(y_{S,n}\in\{5,7\})
\nonumber\\&
+a_{P,n}\chi(y_{S,n}\in\{1,3,6,7\})b_{S,n},
\end{align}
and, by taking the expectation with respect to $y_{S,n}$,
\begin{align}
\mathbb E\left[g(a_{S,n},a_{P,n},y_{S,n},\U,b_{S,n})\right]
=&
a_{S,n}(\delta_{sp}+\delta_{s})
+a_{S,n}(1-a_{P,n})(\upsilon_{sp}+\upsilon_{s})
\nonumber\\&
+a_{P,n}(\delta_{sp}+\delta_{p}+\upsilon_{sp}+\upsilon_{p})b_{S,n},
\end{align}
where the first two terms correspond to the successful decoding of the current SU packet, and the last term refers  to the event
that $l_{P,n}$ is virtually decoded, hence the $b_{S,n}$ SU packets are virtually decoded as well.
\item
If $(\Phi_n,b_{S,n})=(\stackrel{\leftrightarrow}{\K},0)$, then
$\hat\kappa_{P,n}^{(GA)}=1$, $\iota_{P,n}=1$ and $b_{S,n}=0$.
Since $\iota_{P,n}=1$ implies 
$a_{P,m}\chi(y_{S,m}\in\{1,3,6,7\})-a_{P,m}a_{S,m}\chi(y_{S,m}=7)=0,\forall \nu_P(k)\leq m<n$,
and 
 $\hat\kappa_{P,n}^{(GA)}=1$ excludes
 $a_{P,m}\chi(y_{S,m}\in\{1,3,6,7\})=0,\forall \nu_P(k)\leq m<n$,
 it follows that  there exists $\nu_P(k)\leq m<n$ such that 
 $a_{P,m}\chi(y_{S,m}\in\{1,3,6,7\})=1$ and $a_{P,m}a_{S,m}\chi(y_{S,m}=7)=1$,
 \emph{i.e.},
 $a_{P,m}=1$, $a_{S,m}=1$ and $\chi(y_{S,m}=7)$.
Therefore, the current PU packet is connected to the root of the CD graph as $l_{P,n}\leftrightarrow \rho_S(\mathcal G_n)$.
From (\ref{g}), we thus obtain
    \begin{align}
g(a_{S,n},a_{P,n},y_{S,n},\stackrel{\leftrightarrow}{\K},0)
=&
a_{S,n}\chi(y_{S,n}\in\{1,2,5,7\})
\nonumber\\&
+a_{P,n}\chi(y_{S,n}\in\{1,3,6\})+a_{P,n}(1-a_{S,n})\chi(y_{S,n}=7),
\end{align}
and, by taking the expectation with respect to $y_{S,n}$,
    \begin{align}
\mathbb E\left[g(a_{S,n},a_{P,n},y_{S,n},\stackrel{\leftrightarrow}{\K},0)\right]
=&
a_{S,n}(\delta_{sp}+\delta_{s}+\upsilon_{sp}+\upsilon_{s})
\nonumber\\&
+a_{P,n}(\delta_{sp}+\delta_{p}+\upsilon_{p})+a_{P,n}(1-a_{S,n})\upsilon_{sp},
\end{align}
where the first term refers to a successful decoding operation of the current SU packet (since $l_{P,n}$ is virtually known, its interference can be removed),
and the last two terms refer to the event when either $l_{P,n}$ is successfully decoded, so that the current root, which is reachable from $l_{P,n}$, is decoded as well,
or $l_{S,n}\rightarrow l_{P,n}$, so that $l_{S,n}$ becomes the new root of the graph and the previous root is virtually decoded.
\item
  Finally, if  $(\Phi_n,b_{S,n})=(\stackrel{\rightarrow}{\K},0)$, then
$\hat\kappa_{P,n}^{(GA)}=1$, $\iota_{P,n}=0$ and $b_{S,n}=0$.
Since $\iota_{P,n}=0$
excludes $a_{P,m}\chi(y_{S,m}\in\{1,3,6,7\})-a_{P,m}a_{S,m}\chi(y_{S,m}=7)=0,\forall \nu_P(k)\leq m<n$,
 there exists some $\nu_P(k)\leq m<n$ such that 
$a_{P,m}\chi(y_{S,m}\in\{1,3,6,7\})-a_{P,m}a_{S,m}\chi(y_{S,m}=7)=1$,
so that either the current PU packet is decoded successfully, or it is connected to the root of the CD graph as $\rho_S(\mathcal G_{n})\rightarrow l_{P,n}$.
From (\ref{g}), we thus obtain
   \begin{align}
&g(a_{S,n},a_{P,n},y_{S,n},\stackrel{\rightarrow}{\K},0)=a_{S,n}\chi(y_{S,n}\in\{1,2,5,7\}),
\end{align}
and, by taking the expectation with respect to $y_{S,n}$,
    \begin{align}
\mathbb E\left[g(a_{S,n},a_{P,n},y_{S,n},\stackrel{\rightarrow}{\K},0)\right]
=&
a_{S,n}(\delta_{sp}+\delta_{s}+\upsilon_{sp}+\upsilon_{s}).
\end{align}
In fact, since $l_{P,n}$ is virtually known, its interference can be removed,
so that the current SU packet can be virtually decoded if $y_{S,n}\in\{1,2,5,7\}$.
\end{itemize}

Note that
the distinction between $\Phi=\stackrel{\leftrightarrow}{\K}$ and $\Phi=\stackrel{\rightarrow}{\K}$ is necessary,
since in the configuration $\Phi=\stackrel{\leftrightarrow}{\K}$
the root may become reachable by a new root with larger CD potential, 
as shown in the example provided in Fig. \ref{fig:example}.e-f;
 on the other hand, if $\Phi=\stackrel{\rightarrow}{\K}$, then no SU transmission can achieve higher CD potential than the current root.
Also, note that virtual FIC can be employed by SUrx
in states $\Phi=\stackrel{\leftrightarrow}{\K}$ and $\Phi=\stackrel{\rightarrow}{\K}$
to perform its new  transmissions, since the current PU packet
is virtually known.
Thus, in these states the SU takes advantage of  a "clean"  transmission channel.
This fact does not hold when $\Phi=\U$, since the PU packet is virtually unknown.

The state space of the CD protocol in the virtual system is thus given by
\begin{align}
& \mathcal W=
\left\{(\U,b):b\in\mathbb N(0,R_{\max}-1)\right\}\nonumber\\&
\cup
\left\{(\stackrel{\leftrightarrow}{\K},0),(\stackrel{\rightarrow}{\K},0)\}
\right\}
\end{align}
with finite cardinality $|\mathcal W_{\mathrm{V}}|=R_{\max}+2$,
 as opposed to the original formulation, where the state space is infinite.
Therefore,
the virtual system allows
a compact state space representation of the CD protocol,
such that the specific structure of the CD graph, \emph{e.g.},
the decoding relationships between the packets in the graph,  need not be taken into account.
This compact representation thus lends itself to complexity reduction in the operation and optimization of the SU access policy.

\begin{figure}[t]
    \centering
    \subfigure[]
    {
\includegraphics[width=.22\linewidth,trim = 1mm 1mm 1mm 1mm,clip=true,frame]{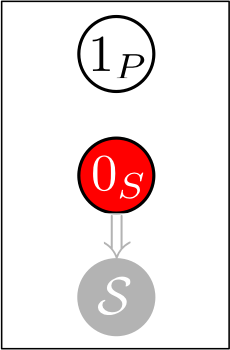}
}
    \subfigure[]
    {
\includegraphics[width=.22\linewidth,trim = 1mm 1mm 1mm 1mm,clip=true,frame]{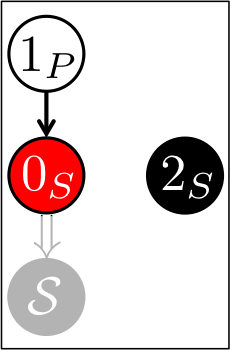}
}
\subfigure[]
    {
\includegraphics[width=.22\linewidth,trim = 1mm 1mm 1mm 1mm,clip=true,frame]{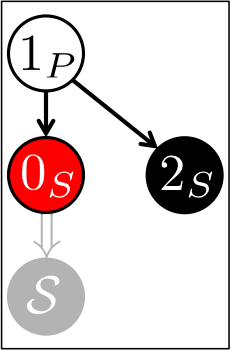}
}
\subfigure[]
    {
\includegraphics[width=.22\linewidth,trim = 1mm 1mm 1mm 1mm,clip=true,frame]{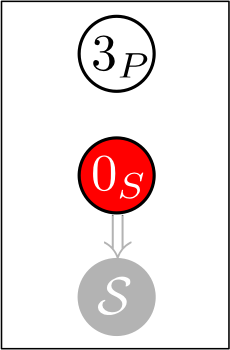}
}
\subfigure[]
    {
\includegraphics[width=.22\linewidth,trim = 1mm 1mm 1mm 1mm,clip=true,frame]{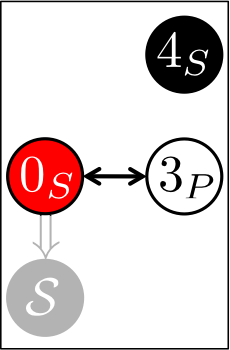}
}
    \subfigure[]
    {
\includegraphics[width=.22\linewidth,trim = 1mm 1mm 1mm 1mm,clip=true,frame]{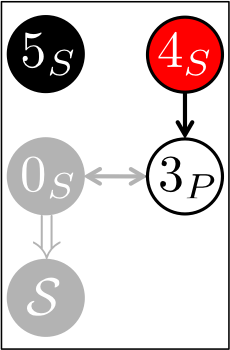}
}
    \subfigure[]
{
\includegraphics[width=.22\linewidth,trim = 1mm 1mm 1mm 1mm,clip=true,frame]{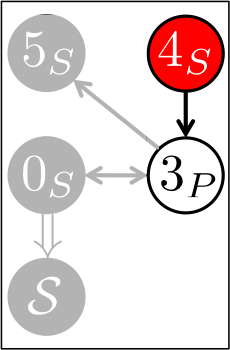}
}
\subfigure[]
    {
\includegraphics[width=.22\linewidth,trim = 1mm 1mm 1mm 1mm,clip=true,frame]{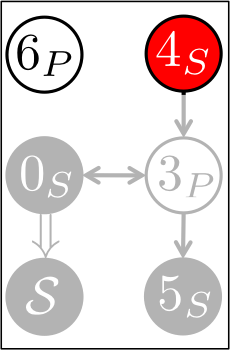}
}
\caption{Virtual CD protocol corresponding to Example  \ref{fig:example}. The portion of the graph that differentiates the physical graph from the virtual one is shaded,
so that the remaining portion of the graph captures the essential features of the CD protocol.
\\
\small
(a) slot 1: new ARQ cycle; PU transmits $1_P$, SU transmits the root $0_S$ ({\bf R1}); $\mathcal S$ has been virtually decoded; CD state $(\U,0)$
\\
(b) slot 2: PU transmits $1_P$, SU transmits a new packet $2_S$ ({\bf R2}); CD state $(\U,1)$
\\
(c) end of slot 2; CD state $(\U,2)$
\\
(d) slot 3: new ARQ cycle, $1_P$ and $2_S$ are dropped from the graph ({\bf R4});
PU transmits $3_P$, SU transmits $0_S$ ({\bf R1});  CD state $(\U,0)$
\\
(e) slot 4: PU transmits $3_P$, SU transmits $4_S$ ({\bf R2});  CD state $(\stackrel{\leftrightarrow}{\K},0)$
\\
(f) slot 5: PU transmits $3_P$, SU transmits $5_S$ ({\bf R2}); $0_S$ virtually decoded; CD state $(\stackrel{\rightarrow}{\K},0)$
\\
(g) end of slot 5; $5_S$ virtually decoded; CD state $(\stackrel{\rightarrow}{\K},0)$
\\
(h) slot 6: new ARQ cycle; PU transmits $6_P$, SU transmits the root $4_S$ ({\bf R1}); CD state $(\U,0)$
}
\label{fig:examplevir}
\end{figure}

The virtual system corresponding to Example \ref{example2}
is depicted in Fig. \ref{fig:examplevir}, along with the state of the CD protocol.
In particular, in slot $1$, $\mathcal S$ is virtually decoded and the CD potential is immediately released,
hence the $|\mathcal S|$ SU packets in $\mathcal S$
instantaneously contribute to the throughput accrual.
In slot $4$, $3_P$ is virtually known at SUrx, hence $5_S$ in slot $5$ is decoded via \emph{virtual FIC} at SUrx.
Moreover, in slot $5$, $0_S$ is virtually decoded as well,
since  $4_S$ becomes the root of the CD graph.

\subsection{Markov decision process formulation}
Based on this compact state space representation, it is possible to reformulate
 problem (\ref{opt3}) as a Markov decision process.
 The state of the system at the beginning of slot $n$ is denoted as
 \begin{align}
\mathbf s_n=(\mathbf s_{CD,n},t_{P,n},d_{P,n},\beta_n),
 \end{align}
 where $\mathbf s_{CD,n}\in\mathcal W$ is the state of the CD protocol,
$t_{P,n}\in\mathbb N(0,R_{\max}-1)$ is the ARQ state, $d_{P,n}\in\mathbb N(t_{P,n},D_{\max}-1)$ is the delay state at the PU pair (both $t_{P,n}$ and $d_{P,n}$ can be tracked by the SU pair, as per Lemma \ref{lem1}),
and $\beta_n$ is the belief available at the SU pair on the value of the internal queue state of PUtx,
based on the history collected, so that $\beta_n(q_{P})$ is the probability that 
$q_{P,n}=q_P$, as seen from the SU pair.
The following theorem establishes that $\mathbf s_n$ is an information state, \emph{i.e.},
$\mathbf s_n$ is a sufficient statistic for decision making at time $n$, so that we can redefine the SU access policy as a function of 
$\mathbf s_n$ only, and
$\mu_S(\mathbf s_n)$ is the probability that $a_{S,n}=1$ in slot $n$.
\begin{thm}
\label{thm6}
$\mathbf s_n$ is an information state.
\end{thm}
\begin{proof}
See Appendix \ref{proofofthm6}.
\end{proof}

A decoupling principle thus follows:
the \emph{secondary access decision}, \emph{i.e.}, whether the SU should transmit or stay idle,
is determined from $\mu_S(\mathbf s_n)$, based on the compact state information $\mathbf s_n$;
on the other hand, \emph{packet selection}, \emph{i.e.}, which packet should be sent if  a transmission is made,
is done based on the four CD rules of Sec. \ref{sec:chdecprot}, based on the state of the CD graph.

We can define the \emph{expected virtual instantaneous throughput}
and the expected PU reward
 under a given state $\mathbf s_n$ and SU access decision
$a_{S,n}$ as
\begin{align}
&\bar g(\mathbf s_n,a_{S,n})\triangleq\mathbb E\left[
\left.g(a_{S,n},a_{P,n},y_{S,n},\Phi_n,b_{S,n})
\right|\mathbf s_n,a_{S,n}
\right],
\\
&\bar{\mathbf r}_P(\mathbf s_n,a_{S,n})\triangleq\mathbb E\left[
\left.
\mathbf r_P(\mathbf s_{P,n},b_{P,n},\boldsymbol{\gamma}_{P,n},a_{P,n},a_{S,n})
\right|\mathbf s_n,a_{S,n}
\right],
\end{align}
where the expectation is with respect to $(\mathbf s_{P,n},b_{P,n},\boldsymbol{\gamma}_{P,n},y_{S,n},a_{P,n})$,
so that we can rewrite
\begin{align}
&\bar g(\mathbf s_n,a_{S,n})=
\sum_{i=1}^7\mathbb P(y_{S,n}=i)
\sum_{q_P}\beta_n(q_{P})\mu_P(t_{P,n},d_{P,n},q_P)g(a_{S,n},1,i,\Phi_n,b_{S,n})
\nonumber\\&
+\sum_{i=1}^7\mathbb P(y_{S,n}=i)
\sum_{q_P}\beta_n(q_{P})(1-\mu_P(t_{P,n},d_{P,n},q_P))g(a_{S,n},0,i,\Phi_n,b_{S,n}).
\end{align}
and 
\begin{align}
&\bar{\mathbf r}_P(\mathbf s_n,a_{S,n})=
\sum_{q_P}\beta_n(q_P)
\mu_P(t_{P,n},d_{P,n},q_{P})
\mathbb E\left[
\mathbf r_P(t_{P,n},d_{P,n},q_{P},b_{P,n},\boldsymbol{\gamma}_{P,n},1,a_{S,n})
\right]
\nonumber\\&
+\sum_{q_P}\beta_n(q_P)(1-\mu_P(t_{P,n},d_{P,n},q_{P}))
\mathbb E\left[
\mathbf r_P(t_{P,n},d_{P,n},q_{P},b_{P,n},\boldsymbol{\gamma}_{P,n},0,a_{S,n})
\right],
\end{align}
where the expectation is with respect to the realization of  $b_{P,n}$ and $\boldsymbol{\gamma}_{P,n}$, which 
are i.i.d. over time.

From Theorem \ref{thm5}, we can thus express $\bar T_{S,\infty}^{(up)}$ as
\begin{align}
\bar T_{S,\infty}^{(up)}(\mu_S)=\lim\inf_{N\to\infty}\mathbb E\left[\frac{1}{N}\sum_{n=0}^{N-1}\bar g(\mathbf s_n,a_{S,n})\right],
\end{align}
and the PU reward as
\begin{align}
&\bar{\mathbf R}_P(\mu_S,\lambda_S^{(CD)})=\lim\inf_{N\to\infty}
\mathbb E\left[\frac{1}{N}\sum_{n=0}^{N-1}\bar{\mathbf r}_P(\mathbf s_n,a_{S,n})\right],
\end{align}
where the expectation is with respect to $a_{S,n}$, generated according to policy $\mu_S(\mathbf s_n)$,
and to the state sequence $\{\mathbf s_n\}$ induced by $\mu_S$.

In the special case where $\beta_n$ takes values from a finite set $\mathcal B$,\footnote{This happens, for instance,
if $b_{P,n}=Q_{\max}$ with probability $1$, so that $q_{P,n}=Q_{\max},\forall n$ and the data queue is saturated; in this case,
$\beta_n(q_P)=\chi(q_P=Q_{\max})$.}
$\mathbf s_n$ takes values from a finite set. Thus, assuming the SU access policy $\mu_S$ induces an irreducible Markov chain $\{\mathbf s_n,n\geq 0\}$, and letting $\pi_{\mu_S}(\mathbf s_n)$ be its steady-state distribution
under the SU access policy $\mu_S$,
the average long-term SU  throughput and PU reward can be rewritten as
\begin{align}\label{throughput}
&\bar T_{S}^{(up)}(\mu_S)=
\sum_{\mathbf s}\pi_{\mu_S}(\mathbf s)[\mu_S(\mathbf s)\bar g(\mathbf s,1)+(1-\mu_S(\mathbf s))\bar g(\mathbf s,0)],
\\
&\bar{\mathbf R}_P(\mu_S,\lambda_S^{(CD)})
=
\sum_{\mathbf s}\pi_{\mu_S}(\mathbf s)[\mu_S(\mathbf s)\bar{\mathbf r}_P(\mathbf s,1)+(1-\mu_S(\mathbf s))\bar{\mathbf r}_P(\mathbf s,0)].
\end{align}

The optimization problem $\mathbf{P3}$ in (\ref{opt3}) can then be solved efficiently using dynamic programming tools, such as policy iteration \cite{Bertsekas}.

\section{Numerical Results}\label{sec:numres}

We now present some numerical results.
We consider Rayleigh fading channels with average SNR $\bar\gamma_s$, $\bar\gamma_{ps}$,
$\bar\gamma_p$ and $\bar\gamma_{sp}$.
For a given set of average SNR parameters,
the transmission rate for the PU system, $R_p$, is chosen so as to maximize
 the primary throughput when the SU is idle.
Similarly,
the transmission rate for the SU system, $R_s$, is chosen so as to maximize the secondary throughput
when the PU is idle. 
Such choice of $R_s$ is driven by the fact that
IC of the PU packet is potentially enabled by the CD protocol,
hence the SU may benefit from a clean channel between its transmitter-receiver pair.
The primary ARQ deadline is set to $R_{\max}=5$, and the delay deadline to $D_{\max}=5$.
The PU is assumed to be backlogged, and it always transmits ($\mu_P(\mathbf s_P)=1,\forall \mathbf s_P$).
The performance metric considered for the PU pair is the throughput, \emph{i.e.}, 
$\mathbf r_P(\mathbf s_P,b_P,\boldsymbol{\gamma}_P,a_P,a_S)=a_P\chi\left(\boldsymbol{\gamma}_P\in\Gamma_P(a_{S})\right)$.
The maximum throughput achieved by the PU pair when the SU is idle is thus denoted as $\bar T_{P,\max}=\chi\left(\boldsymbol{\gamma}_P\in\Gamma_P(0)\right)$.

We consider the following policies in addition to "chain decoding":
 "FIC/BIC", which employs both FIC and BIC, but does not
use the CD mechanism (see \cite{MichelusiJSAC});
"FIC only", which employs only FIC, \emph{i.e.},
once the current PU packet is decoded by SUrx, its knowledge
is exploited in the following primary ARQ retransmissions to perform
IC at SUrx, see \cite{MichelusiITA11};
"no FIC/BIC", which employs neither BIC nor FIC. In this case, the SU packet is
 decoded by leveraging the PU codebook structure \cite{rateRegion,Taranto};
however, possible knowledge of the PU packet gained during the decoding operation
is only used in the slot where the PU packet is acquired, whereas it is neglected in past/future
slots.



\begin{figure}[t]
\centering  
\includegraphics[width=.8\linewidth,trim = 0mm 3mm 0mm 12mm,clip=true]{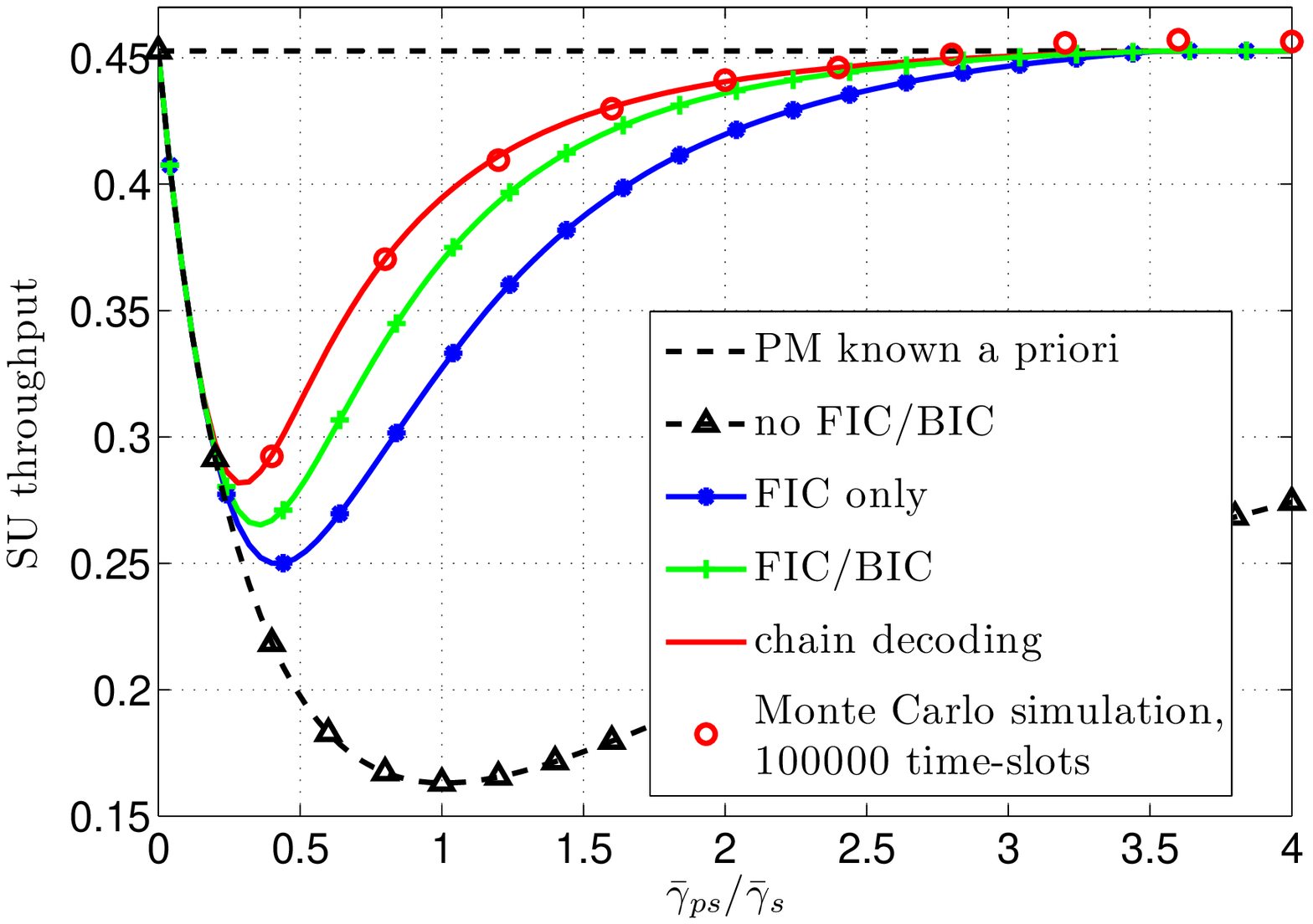}
\caption{SU throughput  vs SNR ratio $\bar\gamma_{ps}/\bar\gamma_s$.
PU throughput constraint $\bar{\mathbf R}_{P,\min}=0.8\bar T_{P,\max}$. $\bar\gamma_s=5$, $\bar\gamma_{sp}=2$, $\bar\gamma_p=10$.}
\label{TSvsGps}
\end{figure}

We point out  that the CD protocol,
by allowing the SU to intelligently perform selective retransmissions
of SU data packets, best exploits the 
primary ARQ process,
among all the schemes with fixed transmission rate and power and no cooperation between
PU and SU.
 Hence, it represents the optimal scheme for the system considered in this paper.

In Fig.~\ref{TSvsGps}, we plot the SU throughput given by (\ref{throughput}) versus the
 SNR ratio $\bar\gamma_{ps}/\bar\gamma_s$, where $\bar\gamma_s=5$.
For selected values of the SNR ratio, we plot also
the results of Monte Carlo simulations over 100000 slots,
 based on the real (not virtual) system.
In this regard, we notice a good match between the simulation curve
and the throughput curve (\ref{throughput}) based on the virtual system,
despite the finite time-horizon of the simulation.

We notice that, when $\bar\gamma_{ps}=0$, the upper bound,
corresponding to the case where the current PU packet is known a priori by SUrx,
 is achieved with equality by all mechanisms,
since the SU operates under no interference from the PU.
The upper bound is approached also for $\bar\gamma_{ps}\gg \bar\gamma_s$,
corresponding to a strong interference regime where, with high probability,
SUrx can successfully decode the PU packet,
remove its interference from the received signal, and then
 attempt to decode the SU packet.
The worst performance is attained when $\bar\gamma_{ps}\simeq\bar\gamma_{s}/2$
($\bar\gamma_{ps}\simeq\bar\gamma_{s}$ for "no FIC/BIC"),
since the interference from the PU is neither weak enough
 to be simply treated as noise, nor
strong enough to be successfully decoded and then removed.
We observe that, for $\bar\gamma_{ps}< 0.2\bar\gamma_s$,
"no FIC/BIC" is sufficient to achieve optimality.
This is because the signal from PUtx to SUrx is very weak,
hence, with high probability, $R_p\geq C(\gamma_{ps})$, so that 
a successful decoding operation of the current PU packet by SUrx is unlikely to
occur (even if the SU remains idle), hence the CD graph does not build up.
On the other hand, "FIC/BIC" is  sufficient to achieve optimality in the regime
 $\bar\gamma_{ps}>2\bar\gamma_s$.
This is because the signal from PUtx to SUrx is strong, hence,
with high probability, the PU packet is decoded before 
the ARQ retransmission window terminates, thus enabling the decoding
of the buffered SU packets via BIC.
Therefore, the CD graph does not build up over subsequent primary ARQ cycles.
Instead, a throughput improvement of the CD protocol
over the other mechanisms
can be noticed in the range $\bar\gamma_{ps}\in(0.2\bar\gamma_s,2\bar\gamma_s)$,
with a throughput gain of $\sim$10\% with respect to "FIC/BIC"
and $\sim$25\% with respect to "FIC only".
Even though these throughput gains may seem modest, they represent the maximum improvement that can be achieved by any scheme under our assumptions,
showing that CD is able to completely close the gap between the previous schemes and optimality, and is therefore the ultimate scheme.

\begin{figure}[t]
\centering  
\includegraphics[width=.8\linewidth,trim = 0mm 12mm 0mm 10mm,clip=true]{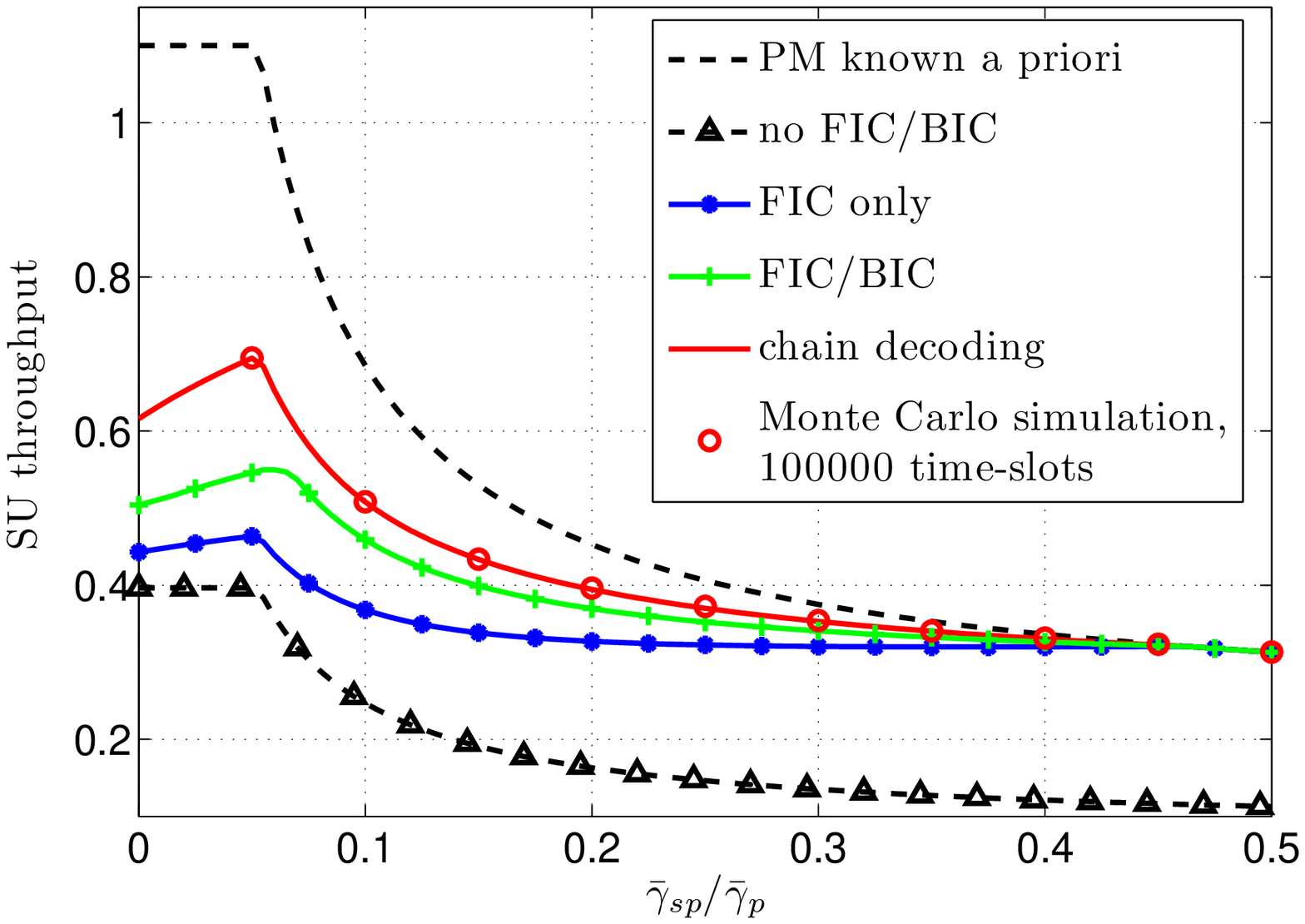}
\caption{SU throughput  vs SNR ratio $\bar\gamma_{sp}/\bar\gamma_p$.
PU throughput constraint $\bar{\mathbf R}_{P,\min}=0.8\bar T_{P,\max}$.
$\bar\gamma_s=\bar\gamma_{ps}=5$, $\bar\gamma_p=10$.}
\label{TSvsGsp}
\end{figure}

In Fig.~\ref{TSvsGsp}, we plot the SU throughput versus the
  SNR ratio $\bar\gamma_{sp}/\bar\gamma_p$, where $\bar\gamma_p=5$. 
Note that, for $\bar\gamma_{sp}/\bar\gamma_p\leq 0.5$, the SU throughput increases.
In fact, in this regime the activity of the SU causes little
harm to the PU, and the constraint on the PU throughput loss
is inactive. The SU thus maximizes its own throughput. As $\bar\gamma_{sp}$ increases
from $0$ to $0.5\bar\gamma_{p}$,
the activity of the SU induces
more frequent primary ARQ retransmissions (still within the constraint imposed), hence there are more IC
opportunities available and the SU throughput augments.
On the other hand, as  $\bar\gamma_{sp}$ grows beyond $0.5\bar\gamma_p$,
the constraint on the PU throughput loss becomes active, 
 SU accesses become more and more harmful to the PU
and take place more and more sparingly, hence the SU throughput degrades.
As in the previous figure, we notice a good match between Monte Carlo simulations
and the numerical throughput curve,
and a throughput benefit of CD over the other mechanisms,
$\sim$20\% with respect to "FIC/BIC" and $\sim$30\% for small $\bar\gamma_{sp}$.

\section{Conclusions}\label{sec:remarks}
We have studied the problem of secondary access in
a cognitive radio network, where the primary user pair employs Type-I Hybrid ARQ.
We have proposed a CD protocol,
such that the SU receiver (SUrx) buffers the secondary signals that are not successfully decoded,
 and selectively performs retransmissions of previously failed transmission attempts.
In fact primary (due to ARQ) and secondary (according to the rules of the CD protocol) retransmissions introduce temporal
redundancy in the channel, which can be exploited for interference cancellation at the SU receiver.
We have shown that the CD protocol is defined by four simple
rules, which enable a compact state representation of the protocol
and its efficient numerical optimization.
Finally, we have shown numerically the throughput benefits of CD,
with respect to other strategies proposed in the literature.

\appendices

\section{}
\label{applemma1}
\begin{lemma}
\label{lem1}
Let $t_{P,0}=0$ in slot $0$ (no active retransmission session).
Then, $(a_{P,n-1},t_{P,n},d_{P,n},l_{P,n-1})$
is a function of $y_{P,0}^{n-1}$.
\end{lemma}
\begin{proof}
The proof is by induction, with the help of  Secs. \ref{sec:intstate} and \ref{sec:evol}.
Assume that $(a_{P,k},t_{P,k+1},d_{P,k+1},l_{P,k})$ is a function of $y_{P,0}^{k}$ (induction hypothesis).
This is true for $k=0$, since:
\begin{itemize}
\item $y_{P,0}=\emptyset$ implies
$a_{P,0}=0$, $l_{P,0}=0_P$, hence $t_{P,1}=d_{P,1}=0$ (\emph{i.e.}, no transmission performed and no existing active session in slot $0$);
\item $y_{P,0}=\text{ACK}$  implies
$a_{P,0}=1$, $l_{P,0}=0_P$, $t_{P,1}=0$ and $d_{P,1}=0$, since the transmission is successful and the session ends;
\item $y_{P,0}=\text{NACK}$  implies
$a_{P,0}=1$, $l_{P,0}=0_P$, $t_{P,1}=1$ and $d_{P,1}=1$, since a retransmission needs to be performed in the next slot.
\end{itemize}
For $k\geq 0$, we show that the induction hypothesis implies  that 
$(a_{P,k+1},t_{P,k+2},d_{P,k+2},l_{P,k+1})$ is a function of $y_{P,0}^{k+1}$, thus proving the lemma.
We have that $o_{P,k+1}=\sigma(t_{P,k+1},d_{P,k+1},y_{P,k+1})$ from (\ref{sigma}),
so that $o_{P,k+1}$  is a function of $y_{P,0}^{k+1}$.
Therefore, given $y_{P,0}^{k+1}$, we have the following:
\begin{align}
a_{P,k+1}=\chi(y_{P,k+1}\in\{\text{ACK},\text{NACK}\}),
\end{align}
since $y_{P,k+1}=\emptyset$ if and only if the PU remains idle in slot $k+1$;
from (\ref{tp}) and (\ref{dp}),
\begin{align}
   &t_{P,k+2}=(1-o_{P,k+1})(t_{P,k+1}+a_{P,k+1}),
\\&d_{P,k+2}=(1-o_{P,k+1})\left[d_{P,k+1}+\chi(t_{P,k+1}>0)+\chi(t_{P,k+1}=0)a_{P,k+1}\right],
\end{align}
so that $t_{P,k+2}$ and $d_{P,k+2}$ are functions of  $y_{P,0}^{k+1}$;
finally, $l_{P,k+1}=(k+1)_P$ if $a_{P,k+1}=0$,
and $l_{P,k+1}=(k+1-d_{P,k+1})_P$ if $a_{P,k+1}=1$.
The induction step and the lemma are thus proved.
\end{proof}

\section{Proof of Theorem \ref{CD}}
\label{proofofCD}
\begin{proof}
Let $\mathcal A_P=\{a_{P,n},n\geq0\}$ and $\mathcal A_S=\{a_{S,n},n\geq0\}$ be a realization of the PU and SU access sequences;
$\mathcal L_P=\{l_{P,n},n\geq0\}$ and $\mathcal L_S=\{l_{S,n},n\geq0\}$ be a realization of the PU and SU label sequences;
$\mathcal Y_S=\{y_{S,n},n\geq0\}$  be the SNR region sequence, where $y_{S,n}=i$ if and only if
$\boldsymbol{\gamma}_S(n)\in\Gamma_{S,i}$.
In the following proof, we keep $(\mathcal A_P,\mathcal A_S,\mathcal L_P,\mathcal Y_S)$ fixed
and we vary only the SU labeling sequence $\mathcal L_S$. Therefore, we express the dependence on $\mathcal L_S$ only.
Let  $\kappa_{S,n}(l_S;\mathcal L_S)\in\{0,1\},\ l_S\geq 0$ be the decoding outcome for the SU packets
at SUrx under $\mathcal L_S$, \emph{i.e.}, $\kappa_{S,n}(l_S;\mathcal L_S)=1$ if $l_S$ has been successfully decoded by SUrx
before and not including slot $n$, and $\kappa_{S,n}(l_S;\mathcal L_S)=0$ otherwise.
Similarly, 
let  $\kappa_{P,n}(l_P;\mathcal L_S)\in\{0,1\},\ l_P\geq 0_P$ be the decoding outcome for the PU packets
at SUrx under $\mathcal L_S$, \emph{i.e.}, $\kappa_{P,n}(l_P;\mathcal L_S)=1$ if $l_P$ has been successfully decoded by SUrx
before and not including slot $n$, and $\kappa_{P,n}(l_P;\mathcal L_S)=0$ otherwise.
Note that $\kappa_{S,n}(\cdot;\mathcal L_S)$ and $\kappa_{P,n}(\cdot;\mathcal L_S)$ are univocally determined by $(\mathcal A_P,\mathcal A_S,\mathcal L_P,\mathcal L_S,\mathcal Y_S)$, by applying recursively CD.
With this definition, $\sum_{l_S=0}^{\infty}(\kappa_{S,n+1}(l_S;\mathcal L_S)-\kappa_{S,n}(l_S;\mathcal L_S))$
is the number of SU packets successfully decoded in slot $n$, 
as a consequence of a direct decoding operation or via CD.

We define the sample average secondary throughput up to slot $N$ under the sequence $\mathcal Z$ as
\begin{align}
\label{qweqweqwe}
\bar T_{S,N}(\mathcal L_S)
=\frac{1}{N} \sum_{l_S=0}^{\infty}\kappa_{S,N}(l_S;\mathcal L_S),
\end{align}
and the expected PU reward up to slot $N$ under the sequence $\mathcal Z$ as
\begin{align}
\bar{\mathbf R}_{P,N}
=\frac{1}{N}\sum_{n=0}^{N-1}
\mathbb E\left[\left.\mathbf r_P(\mathbf s_{P,n},b_{P,n},\boldsymbol{\gamma}_P(n),a_{P,n},a_{S,n})\right|
\mathcal A_P,\mathcal A_S,\mathcal L_P,\mathcal Y_S
\right],
\end{align}
where 
the expectation is computed with respect to $\boldsymbol{\gamma}_P(n)$ given $(\mathcal A_P,\mathcal A_S,\mathcal L_P,\mathcal Y_S)$, and
we have used the fact that $\mathbf r_P(\mathbf s_{P},b_{P},\boldsymbol{\gamma}_P,a_{P},a_{S})$ is independent of the SU label $l_S$,
so that  $\bar{\mathbf R}_{P,N}$ is independent of $\mathcal L_S$.

We now solve the following optimization problem:
\begin{align}
\label{optimallabeling}
\mathcal L_S^*=\arg\max_{\mathcal L_S} \bar T_{S,N}(\mathcal L_S),
\end{align}
for a given $(\mathcal A_P,\mathcal A_S,\mathcal L_P,\mathcal Y_S)$.
Note that this optimization problem does not affect the PU reward 
$\bar{\mathbf R}_{P,N}$, since the latter is independent of $\mathcal L_S$, given $(\mathcal A_P,\mathcal A_S,\mathcal L_P,\mathcal Y_S)$.
Moreover, $\mathcal L_S^*$ is the optimal \emph{offline} labeling scheme, which assumes that the sequence 
$(\mathcal A_P,\mathcal A_S,\mathcal L_P,\mathcal Y_S)$ is known non-causally at SUtx. 
Indeed, in the following proof, we will show that, when $N\to\infty$, $\mathcal L_S^*$ is defined by the CD rules, and
can be implemented causally, \emph{i.e.}, it does not require non-causal knowledge of $(\mathcal A_P,\mathcal A_S,\mathcal L_P,\mathcal Y_S)$.

Note that, if $a_{P,n}=0$ or $a_{S,n}=0$, then PUtx or SUtx are idle, respectively;
from the perspective of decoding the PU and SU packets at SUrx and initiating CD, 
the same outcome can be obtained in a new system where both PUtx and SUtx transmit, but the channel gain to SUrx is in a different region.
For instance, 
if $a_{P,n}=1$, $a_{S,n}=0$, $y_{S,n}\in\{1,3,6,7\}$, then the current PU packet is decoded by SUrx;
the same outcome is obtained in another system where both PUtx and SUtx transmit ($a_{P,n}=a_{S,n}=1$) but 
$\tilde y_{S,n}=3$, so that the current PU packet is decoded by treating the SU packet as noise, but the SU packet cannot be decoded.
Thus, we can add transmissions at the PU and SU that are not adding any positive SU throughput, by proper mapping of the channel gains.
This is formalized in the following lemma.
\begin{lemma}
\label{generateseq}
For a given sequence  $({\mathcal A}_P,{\mathcal A}_S,{\mathcal L}_P,{\mathcal L}_S,{\mathcal Y}_S)$,
there exists a sequence  $(\tilde{\mathcal A}_P,\tilde{\mathcal A}_S,\tilde{\mathcal L}_P,\tilde{\mathcal L}_S,\tilde{\mathcal Y}_S)$
achieving the same SU throughput as $({\mathcal A}_P,{\mathcal A}_S,{\mathcal L}_P,{\mathcal L}_S,{\mathcal Y}_S)$,
where both the PU and the SU always transmit.
For such sequence, 
$\tilde a_{P,n}=\tilde a_{S,n}=1,\ \forall n$, and
$\tilde{\mathcal Y}_S$, $\tilde{\mathcal L}_P$ and $\tilde{\mathcal L}_S$ are univocally defined as follows:
 \begin{itemize}
 \item If $a_{P,n}=a_{S,n}=1$, then
  $\tilde y_{S,n}=y_{S,n}$, $\tilde l_{P,n}=l_{P,n}$, $\tilde l_{S,n}=l_{S,n}$;
  \item If $a_{P,n}=1$, $a_{S,n}=0$, $y_{S,n}\in\{1,3,6,7\}$, then
   $\tilde y_{S,n}=3$, $\tilde l_{P,n}=l_{P,n}$, $\tilde l_{S,n}=n_S$;
     \item If $a_{P,n}=1$, $a_{S,n}=0$, $y_{S,n}\in\{2,4,5\}$, then
   $\tilde y_{S,n}=4$, $\tilde l_{P,n}=l_{P,n}$, $\tilde l_{S,n}=n_S$;
     \item If $a_{P,n}=0$, $a_{S,n}=1$, $y_{S,n}\in\{1,2,5,7\}$, then
   $\tilde y_{S,n}=2$, $\tilde l_{P,n}=\tilde l_{P,n-1}$, $\tilde l_{S,n}=l_{S,n}$;
     \item If $a_{P,n}=0$, $a_{S,n}=1$, $y_{S,n}\in\{3,4,6\}$, then
   $\tilde y_{S,n}=4$, $\tilde l_{P,n}=\tilde l_{P,n-1}$, $\tilde l_{S,n}=l_{S,n}$;
        \item If $a_{P,n}=0$, $a_{S,n}=0$, then
   $\tilde y_{S,n}=4$, $\tilde l_{P,n}=\tilde l_{P,n-1}$, $\tilde l_{S,n}=n_S$.
   \qed
 \end{itemize}
 \end{lemma}
 
We can thus exploit Lemma \ref{generateseq} and  proceed as follows: (1) given a sequence 
$(\mathcal A_P,\mathcal A_S,\mathcal L_P,\mathcal L_S,\mathcal Y_S)$, 
we define $(\tilde{\mathcal A}_P,\tilde{\mathcal A}_S,\tilde{\mathcal L}_P,\tilde{\mathcal L}_S,\tilde{\mathcal Y}_S)$
as per Lemma \ref{generateseq}, which preserves the SU throughput;
(2) Given $(\tilde{\mathcal A}_P,\tilde{\mathcal A}_S,\tilde{\mathcal L}_P,\tilde{\mathcal L}_S,\tilde{\mathcal Y}_S)$, 
we then solve the optimization problem (\ref{optimallabeling}) to determine the optimal labeling sequence $\tilde{\mathcal L}_S^*$;
(3) Given the optimal labeling sequence $\tilde{\mathcal L}_S^*$, we then define the optimal labeling sequence ${\mathcal L}_S^*$
for the original sequence $(\mathcal A_P,\mathcal A_S,\mathcal L_P,\mathcal Y_S)$ as
\begin{align}
l_{S,n}^*=\tilde l_{S,n}^*,\ \text{if }a_{S,n}=1,
\\
l_{S,n}^*=n_S,\ \text{if }a_{S,n}=0.
\end{align}
Notice that ${\mathcal L}_S^*$ generated with this approach is indeed the optimal labeling sequence solving the optimization problem (\ref{optimallabeling}) under the original sequence $(\mathcal A_P,\mathcal A_S,\mathcal L_P,\mathcal Y_S)$.
This can be seen by contradiction: if there exists $\hat{\mathcal L}_S$ such that 
$\bar T_{S,N}(\hat{\mathcal L}_S)>\bar T_{S,N}(\mathcal L_S^*)$, then we can define via Lemma \ref{generateseq}
a sequence $(\tilde{\mathcal A}_P,\tilde{\mathcal A}_S,\tilde{\mathcal L}_P,\hat{\tilde{\mathcal L}}_S,\tilde{\mathcal Y}_S)$
achieving the SU throughput $\bar T_{S,N}(\hat{\mathcal L}_S)$; it follows that 
a higher SU throughput is achieved under 
$(\tilde{\mathcal A}_P,\tilde{\mathcal A}_S,\tilde{\mathcal L}_P,\hat{\tilde{\mathcal L}}_S,\tilde{\mathcal Y}_S)$
than under 
$(\tilde{\mathcal A}_P,\tilde{\mathcal A}_S,\tilde{\mathcal L}_P,{\tilde{\mathcal L}}_S^*,\tilde{\mathcal Y}_S)$,
thus contradicting the optimality of the labeling sequence ${\tilde{\mathcal L}}_S^*$; necessarily,
${\mathcal L}_S^*$ is optimal for the original sequence.

It thus remains to determine the optimal labeling sequence $\tilde{\mathcal L}_S^*$
for a given $(\tilde{\mathcal A}_P,\tilde{\mathcal A}_S,\tilde{\mathcal L}_P,\tilde{\mathcal Y}_S)$ with the property that
$\tilde a_{P,n}=\tilde a_{S,n}=1,\forall n$.
Throughout the following proof, we thus make the following assumption.
\begin{assumption}
\label{assumptionalltx}
The sequence $(\mathcal A_P,\mathcal A_S,\mathcal L_P,\mathcal L_S,\mathcal Y_S)$ is such that
\begin{align}
&a_{P,n}=a_{S,n}=1,\forall n.
\end{align}
\qed
\end{assumption}

We proceed as follows. First, in Theorem \ref{thm1}, we determine an upper bound to
$\bar T_{S,N}(\mathcal L_S)$,
 which is independent of $\mathcal L_S$ (but does depend on $(\mathcal L_P,\mathcal Y_S)$)
  and holds for any SU labeling sequence $\mathcal L_S$.
 Then, in Theorem \ref{thm2}, we determine a lower bound to $\bar T_{S,N}(\mathcal L_S^{(CD)})$,
 the throughput obtained under the  labeling sequence $\mathcal L_S^{(CD)}$ generated according to the CD rules.
Finally, we show that  the lower bound, determined via the CD rules, converges to the upper bound  for $N\to\infty$, thus proving the optimality of the CD rules.
Note that $\{\kappa_{S,N}(l_S;\mathcal L_S),N\geq 0\}$ is a non-decreasing and bounded sequence, therefore its limit,
denoted as $\kappa_{S}^*(l_S;\mathcal L_S)\triangleq\lim_{N\to\infty}\kappa_{S,N}(l_S;\mathcal L_S)$,
 exists and $\kappa_{S}^*(l_S;\mathcal L_S)=1$ if and only if $l_S$ is decoded under the sequence $\mathcal L_S$.
 
 The following results are derived for a sequence satisfying Assumption \ref{assumptionalltx}. We can map back to the original sequence where
 $a_{P,n},a_{S,n}\in\{0,1\}$ using the following procedure, which is obtained by inspecting all the different cases arising in Lemma \ref{generateseq}.
 \begin{proced}
\label{lemreplace}
The general case where $a_{S,n}\in\{0,1\}$ and $a_{P,n}\in\{0,1\}$ is obtained by replacing:
\begin{itemize}
\item  $\chi(\tilde y_{S,n}=1)$ with $a_{P,n}a_{S,n}\chi(y_{S,n}=1)$;
\item  $\chi(\tilde y_{S,n}=2)$ with $a_{S,n}\chi(y_{S,n}=2)+(1-a_{P,n})a_{S,n}\chi(y_{S,n}\in\{1,5,7\})$;
\item  $\chi(\tilde y_{S,n}=3)$ with $a_{P,n}\chi(y_{S,n}=3)+a_{P,n}(1-a_{S,n})\chi(y_{S,n}\in\{1,6,7\})$;
\item  $\chi(\tilde y_{S,n}=4)$ with 
$
\chi(y_{S,n}=4)+(1-a_{S,n})\chi(y_{S,n}\in\{2,5\})+(1-a_{P,n})\chi(y_{S,n}\in\{3,6\})+(1-a_{P,n})(1-a_{S,n})\chi(y_{S,n}\in\{1,7\})$;
\item  $\chi(\tilde y_{S,n}=i)$ with $a_{P,n}a_{S,n}\chi(y_{S,n}=i)$, for $i\in\{5,6,7\}$.
\qed
\end{itemize}
\end{proced}

 We use the following definitions:
 \begin{itemize}
 \item Let $u_{P,n}\in\{0,1\}$ be the state variable denoting the beginning of a new PU ARQ cycle, \emph{i.e.},
 $u_{P,n}=1$ if a new PU transmission occurs in slot $n$, and $u_{P,n}=0$ otherwise;
 we let 
 $\mathcal U_P=\{u_{P,n},n\geq0\}$ be a realization of this process; 
   \item Let $\nu_P(j),\ j\geq 0$ be the slot index corresponding to the beginning of the $j$th primary ARQ cycle;
  mathematically, $\nu_P(0)=0$  and, for $j>0$, $\nu_P(j)=\min\{n:u_{P,n}=1,n>\nu_P(j-1)\}$;
  note that, owing to the labeling scheme employed by the PU pair, $l_{P,n}=\nu_P(j),\forall \nu_P(j)\leq n <\nu_P(j+1)$;
   \item Let $J_P(n),n\geq 0$, be the index of the primary ARQ cycle slot $n$ belongs to; 
  mathematically,  
 $J_P(n)=j$ if and only if $\nu_P(j)\leq n <\nu_P(j+1)$.
 \end{itemize}

In the next theorem, we upper bound $\bar T_{S,N}(\mathcal L_S)$.
We show that the upper and lower bounds are composed of three components.
The first component, $\bar T_{S,\nu_P(j+1)}^{(GA)}$, is the genie-aided SU throughput, obtained by assuming that the PU packets are known in advance and their interference can be removed.
The second term (\ref{tsup2}) is a throughput degradation term which accounts for the fact that the PU packet cannot be decoded by SUrx within its retransmission cycle,
even in the genie-aided case where the interference from SUtx is removed;
thus, those SU packets with $y_{S,n}\in\{5,7\}$, which are decodable in the genie-aided case and are counted in the genie-aided throughput $\bar T_{S,\nu_P(j+1)}^{(GA)}$, cannot be decoded due to the impossibility to remove the interference from PUtx.
The third term  (\ref{tsup3}) is a throughput degradation term which accounts for the retransmission of the root of the CD graph, rather than transmitting new data packets.

\begin{thm}
\label{thm1}
For a given $(\mathcal U_P,\mathcal Y_S)$, $\forall \mathcal L_S$, $\forall N>0$, 
\begin{align}
\label{upbound}
&\bar T_{S,N}(\mathcal L_S)
\leq
\frac{\nu_P(J_P(N-1)+1)}{N}\bar T_{S,\nu_P(J_P(N-1)+1)}^{(up)},\ \forall \mathcal L_S,
\end{align}
where we have defined, for $j\geq 0$,
\begin{align}
\label{tsup}
&\bar T_{S,\nu_P(j+1)}^{(up)}
=
\bar T_{S,\nu_P(j+1)}^{(GA)}
\\&\label{tsup2}
-\frac{1}{\nu_P(j+1)}\sum_{k=0}^{j}
\left[1-\kappa_{P}^{(GA)}(\nu_P(k))\right]
\sum_{n=\nu_P(k)}^{\nu_P(k+1)-1}\chi(y_{S,n}\in\{5,7\})
\\&\label{tsup3}
-\frac{1}{\nu_P(j+1)}\sum_{k=0}^{j}\left(1-\prod_{n=\nu_P(k)}^{\nu_P(k+1)-1}\chi(y_{S,n}\neq 7)\right)\prod_{n=\nu_P(k)}^{\nu_P(k+1)-1}\chi(y_{S,n}\in\{2,4,5,7\}).
\end{align}
Above, 
\begin{align}
\bar T_{S,\nu_P(j+1)}^{(GA)}=
\frac{1}{\nu_P(j+1)}
\sum_{n=0}^{\nu_P(j+1)-1}\chi(y_{S,n}\in\{1,2,5,7\})
\end{align}
is the genie-aided (GA) throughput up to slot $\nu_P(j+1)$ at SUrx, assuming SUrx knows the PU packet in advance and thus removes its interference,
and
\begin{align}
\label{kappaga}
&\kappa_{P}^{(GA)}(\nu_P(k))
\triangleq
1-\prod_{n=\nu_P(k)}^{\nu_P(k+1)-1}\chi(y_{S,n}\in\{2,4,5\})
    \end{align}
    is the genie-aided (GA)
decoding outcome  at SUrx for the PU packet
$l_{P,n}=\nu_P(k)$ transmitted in slots $\nu_P(k)\leq n<\nu_P(k+1)$, assuming that the SU packets are known in advance and thus their interference can be removed.
\end{thm}
\begin{proof}
See Appendix \ref{proofofthm1}.
\end{proof}

In the second part of the proof, we lower bound $\bar T_{S,N}(\mathcal L^{(CD)})$
where $\mathcal L^{(CD)}$ is defined via the CD rules. We have the following theorem.
\begin{thm}
\label{thm2}
For a given $(\mathcal U_P,\mathcal Y_S)$, 
let the label sequence $\mathcal L^{(CD)}$ be generated according to the CD rules.
Then, $\forall N$, we have
\begin{align}
\label{lowerbound}
&\bar T_{S,N}(\mathcal L_S^{(CD)})
\geq
\frac{\nu_P(\underline J(N)+1)}{N}\bar T_{S,\nu_P(\underline J(N)+1)}^{(up)}
\end{align}
where
\begin{align}
\label{underlineJn}
\underline J(n)=
\max\{k:n\geq\nu_P(k+1)\cap Q_k=1\}\cup\{-1\},
\end{align}
and we have defined $Q_k\in\{0,1\}$ referred to the $k$th ARQ cycle as
\begin{align}
\label{Qk}
Q_k=\chi
\left(\exists\ \nu_P(k)\leq m_1<m_2<\nu_P(k+1):y_{S,m_1}\in\{1,3\},y_{S,m_2}\in\{1,2,5,7\}\right).
\end{align}
\end{thm}
\begin{proof}
See Appendix \ref{proofofthm2}.
\end{proof}
$\underline J(n)$ is the index of the last ARQ cycle, finishing before slot $n$, with the following properties:
there exist two slots $m_1$ and $m_2$ in the $k$th ARQ cycle such that
$y_{S,m_1}\in\{1,3\}$, so that the PU packet is successfully decoded in slot $m_1$ and its interference can be removed.
Moreover, $y_{S,m_2}\in\{1,2,5,7\}$. Since the PU packet is known by SUrx in slots
$m_1<n<\nu_P(k+1)$, the root of the CD graph is transmitted according to CD rule {\bf R3}, until it is decoded successfully, which is
guaranteed by the condition $y_{S,m_2}\in\{1,2,5,7\}$.
Therefore, $Q_k=1$  guarantees that the CD potential of the CD graph is fully released by the end of the $k$th ARQ cycle.
However, note that $Q_k=1$ is a sufficient, but not necessary, condition for the release of the CD potential, hence the lower bound may be loose for general $N$.

By combining Theorems \ref{thm1} and \ref{thm2}, and generating the labeling sequence according to the CD rules,
we obtain
\begin{align}
\label{fjdjsjdj}
\frac{\nu_P(\underline J(N)+1)}{N}\bar T_{S,\nu_P(\underline J(N)+1)}^{(up)}
\leq
\bar T_{S,N}(\mathcal L_S^{(CD)})
\leq
\frac{\nu_P(J_P(N-1)+1)}{N}\bar T_{S,\nu_P(J_P(N-1)+1)}^{(up)}.
\end{align}
Notice that, under the assumption that
the PU starts a new ARQ cycle infinitely often (so that $\nu_P(J(N-1)+1)=N$ infinitely often when $N\to\infty$),
and that the condition $Q_k=1$ occurs infinitely often (so that $\nu_P(\underline J(N)+1)=N$ infinitely often when $N\to\infty$),
we obtain the limits
\begin{align}
\label{aaa}
\lim_{N\to\infty}\frac{\nu_P(J(N-1)+1)}{N}=
\lim_{N\to\infty}\frac{\nu_P(\underline J(N)+1)}{N}=1,
\end{align}
and
\begin{align}
\label{bbb}
\lim_{N\to\infty}
\bar T_{S,\nu_P(\underline J(N)+1)}^{(up)}(\mathcal L_S^{(CD)})
=
\lim_{N\to\infty}
\bar T_{S,\nu_P(J(N-1)+1)}^{(up)}
\triangleq
\bar T_{S,\infty}^{(up)}.
\end{align}
Letting $N\to\infty$ in (\ref{fjdjsjdj}), we thus obtain 
\begin{align}
\label{ccc}
 \lim_{j\to\infty}\bar T_{S,\nu_P(j+1)}(\mathcal L_S^{(CD)})
=\bar T_{S,\infty}^{(up)},
\end{align}
so that the label sequence $\mathcal L^{(CD)}$ converges to the upper bound, and is thus optimal.
Theorem \ref{CD} is thus proved.
\end{proof}

\section{Proof of Theorem \ref{thm1}}
\label{proofofthm1}
\begin{proof}
Since $\{\kappa_{S,n}(l_S;\mathcal L_S),n\geq 0\}$ is a  non-decreasing sequence,
and $N\leq \nu_P(J_P(N-1)+1)$,
 we have
\begin{align}
\label{}
\sum_{l_S=0}^{\infty}\kappa_{S,N}(l_S;\mathcal L_S)\leq
\sum_{l_S=0}^{\infty}\kappa_{S,\nu_P(J_P(N-1)+1)}(l_S;\mathcal L_S),
\end{align}
and therefore, from (\ref{qweqweqwe}),
\begin{align}
\label{ineq10}
\bar T_{S,N}(\mathcal L_S)
\leq\frac{\nu_P(J_P(N-1)+1)}{N}
\bar T_{S,\nu_P(J_P(N-1)+1)}(\mathcal L_S).
\end{align}
In the following, we prove that
\begin{align}
\label{ineq11}
\bar T_{S,\nu_P(j+1)}(\mathcal L_S)
\leq
\bar T_{S,\nu_P(j+1)}^{(up)},\ \forall j\geq0,
\end{align}
which proves the theorem. In fact,
using the inequality (\ref{ineq11}) in (\ref{ineq10}) with $j=J_P(N-1)$,
we obtain the inequality (\ref{upbound}) in the statement of the theorem,
 so that (\ref{ineq10}) and (\ref{ineq11}) imply (\ref{upbound}).

Therefore, in the following we prove the inequality  (\ref{ineq11}) for a generic $j\geq 0$.
Since we consider fixed $j\geq 0$ and $\mathcal L_S$,
in the following analysis we drop the dependence on $\nu_P(j+1)$ and on $\mathcal L_S$.

\subsection{Necessary condition for the decodability of the PU packets}
We first determine a necessary condition for the decodability of the PU packets by SUrx. Let $k\leq j$
and consider the PU packet $l_{P,n}=\nu_P(k)$ transmitted by PUtx in slots $\nu_P(k)\leq n<\nu_P(k+1)$.
$l_{P,n}=\nu_P(k)$ cannot be decoded
if
\begin{align}
A_{P,k}\equiv y_{S,n}\in\{2,4,5\},\ \forall \nu_P(k)\leq n<\nu_P(k+1)
\end{align}
holds true.
This event is independent of $\mathcal L_S$.
In fact, if $A_{P,k}$ holds true, the channel conditions are such that the PU packet $\nu_P(k)$ cannot be decoded even in the genie-aided case where
 the interference from SUtx is known and is removed.
 Therefore, 
 \begin{align}
A_{P,k}\Rightarrow\kappa_{P}(\nu_P(k))=0,
\end{align}
yielding the inequality
\begin{align}
\label{ineqdecodp}
&\kappa_{P}(\nu_P(k))\leq
1-\chi(A_{P,k})
\triangleq\kappa_{P}^{(GA)}(\nu_P(k)),
    \end{align}
    where $\kappa_{P}^{(GA)}(\nu_P(k))$ is the genie-aided (GA)
decoding outcome for the PU packets at SUrx,
which can be explicitly written as in (\ref{kappaga}).

\subsection{Necessary condition for the decodability of the SU packets}
We  now analyze the decodability of the SU packets at SUrx.
Let
\begin{align}
\mathcal N(l_S)\equiv\{n<\nu_P(j+1):l_{S,n}=l_S\}
\end{align}
be the set of slots where  $l_S$ is transmitted.
Then,
$l_S$ cannot be decoded within the first $\nu_P(j+1)$ slots 
if, for all $n\in\mathcal N(l_S)$,
either of the following events occur:
\begin{itemize}
  \item  $\kappa_{P}(l_{P,n})=0$ and $y_{S,n}\notin\{1,2\}$, \emph{i.e.}, the PU packet transmitted
  in slot $n$ cannot be decoded by slot $\nu_P(j+1)$, its interference cannot be removed,
and $l_S$ cannot be decoded jointly with $l_{P,n}$, nor by treating $l_{P,n}$ as noise.
 \item $\kappa_{P}(l_{P,n})=1$ and $y_{S,n}\in\{1,2,5,7\}$, \emph{i.e.}, the PU packet transmitted
  in slot $n$ can be decoded by slot $\nu_P(j+1)$, its interference can be removed,
but the direct link $\gamma_s(n)$ is too weak to make $l_S$ decodable, even after removing the interference from $l_{P,n}$.
    \end{itemize}
    Mathematically, 
        \begin{align}
&
\left[ (1-\kappa_{P,}(l_{P,n}))\chi(y_{S,n}\notin\{1,2\})\right.\nonumber\\&\left.+
     \kappa_{P}(l_{P,n})\chi(y_{S,n}\notin\{1,2,5,7\})
     \right]=1, \forall n\in\mathcal N(l_S)
     \nonumber\\&
     \Rightarrow 
     \kappa_{S}(l_S)=0,
    \end{align}
or equivalently 
        \begin{align}
        \label{conditiondellaminchia}
&      \kappa_{S}(l_S)=1
     \\&
     \Rightarrow 
\exists n\in\mathcal N(l_S):
\left[\chi(y_{S,n}\in\{1,2\})+\kappa_{P}(l_{P,n})\chi(y_{S,n}\in\{5,7\})
     \right]=1.
     \nonumber
    \end{align}
    Letting
    \begin{align}
    \label{tau}
\tau_{S}(l_S)\equiv\left\{n\in\mathcal N(l_S):\left[\chi(y_{S,n}\in\{1,2\})+\kappa_{P}(\nu_P(k))\chi(y_{S,n}\in\{5,7\})
     \right]=1
\right\},
\end{align}
    we can rewrite the logical relationship (\ref{conditiondellaminchia}) as 
    \begin{align}
&      \kappa_{S}(l_S)=1
     \Rightarrow 
|\tau_{S}(l_S)|>0,
     \nonumber
    \end{align}
so that $|\tau_{S}(l_S)|>0$ is a necessary condition for the decodability of $l_S$ by SUrx.
We thus obtain the inequality 
    \begin{align}
    \label{firtstqnet}
&      \kappa_{S}(l_S)\leq \chi(|\tau_{S}(l_S)|>0).
    \end{align}
    Note that $\tau_{S}(l_S)$ represents the set of slots where $l_S$ can possibly be decoded; 
    outside of this set, \emph{e.g.}, in slot $n\in\mathcal N(l_S)\setminus\tau_{S}(l_S)$, $l_S$ cannot be decoded due either to the fact that 
    $l_{P,n}$ cannot be decoded by SUrx ($\kappa_{P}(l_{P,n})=0$) and $y_{S,n}\notin\{1,2\}$,
    or to the fact that 
    $l_{P,n}$ can be decoded by SUrx ($\kappa_{P}(l_{P,n})=1$),
     its interference can be removed from the received signal, but $y_{S,n}\notin\{1,2,5,7\}$.
    In particular, if $\tau_{S}(l_S)\equiv\emptyset$, then  the SU packet $l_{S}$ cannot be decoded successfully by SUrx
    and $\kappa_{S}(l_S)=0$.

We can further bound (\ref{firtstqnet}) as follows. Let
    \begin{align}
\tau_{S}^{(GA)}(l_S)\equiv\left\{n\in\mathcal N(l_S):
\left[\chi(y_{S,n}\in\{1,2\})+\kappa_{P}^{(GA)}(\nu_P(k))\chi(y_{S,n}\in\{5,7\})
     \right]=1
\right\},
\end{align}
obtained by replacing the decodability of the PU packet $\nu_P(k)$, 
$\kappa_{P}(\nu_P(k))$, with the genie-aided decodability $\kappa_{P}^{(GA)}(\nu_P(k))$.
Using (\ref{ineqdecodp}), we have that 
\begin{align}
\tau_{S}(l_S)\subseteq\tau_{S}^{(GA)}(l_S),
\end{align}
and thus, using (\ref{firtstqnet}), 
\begin{align}
\label{necessarycondS}
&      \kappa_{S}(l_S)\leq \chi(|\tau_{S}(l_S)|>0)\leq \chi(|\tau_{S}^{(GA)}(l_S)|>0).
    \end{align}

Let, for each slot $\nu_P(k)\leq n<\nu_P(k+1)$,
\begin{align}
\label{osn}
o_{S,n}=\chi(y_{S,n}\in\{1,2\})+\kappa_{P}^{(GA)}(\nu_P(k))\chi(y_{S,n}\in\{5,7\}).
\end{align}
Then, we can rewrite 
    \begin{align}
    \label{ciaociao}
\tau_{S}^{(GA)}(l_S)\equiv\left\{n\in\mathcal N(l_S):o_{S,n}=1
\right\}
\end{align}
and $|\tau_{S}^{(GA)}(l_S)|=\sum_{n\in\mathcal N(l_S)}o_{S,n}$, so that
\begin{align}
\label{ineqdellaminchia}
&      \kappa_{S}(l_S)
\leq \chi\left(\sum_{n\in\mathcal N(l_S)}o_{S,n}>0\right)
\leq \sum_{n\in\mathcal N(l_S)}o_{S,n}.
    \end{align}

\subsection{Analysis of the upper bound on the throughput}
Let $k\leq j$
and consider the PU packet $l_{P,n}=\nu_P(k)$ transmitted by PUtx in slots $\nu_P(k)\leq n<\nu_P(k+1)$.
Consider the condition
\begin{align}
 \label{case2}
B_{P,k}\equiv\left\{y_{S,n}\in\{2,4,5,7\},\ \forall \nu_P(k)\leq n<\nu_P(k+1)\cap \bar A_{P,k}\right\},
\end{align}
where $\bar A_{P,k}$ denotes the complement of the event $A_{P,k}$.
Note that the condition (\ref{case2}) excludes the genie-aided condition for non-decodability, $A_{P,k}$,
so that $l_{P,n}=\nu_P(k)$ may indeed be decoded in the genie-aided case, if $B_{P,k}$ holds.
In fact, $B_{P,k}$ implies the existence of $n$ such that $y_{S,n}=7$, so that, in principle, $l_{P,n}$ can be decoded after removing the interference from
the SU packet $l_{S,n}$.

Let
\begin{align}
\mathcal N_{k}^{(B,7)}\equiv \{\nu_P(k)\leq n<\nu_P(k+1):y_{S,n}=7\cap B_{P,k}\}
\end{align}
be the set of slots in the $k$th ARQ window such that $B_{P,k}$ holds and $y_{S,n}=7$.
Therefore, $l_{P,n}=\nu_P(k)$ can only be decoded by removing the interference from $l_{S,n}$, for some $n\in\mathcal N_{B,k}$.
Equivalently, in order, first such $l_{S,n}$ is decoded; then, its interference is removed from slot $n$; finally $l_{P,n}=\nu_P(k)$
is decoded without interference from the SU signal.
Note that, if $B_{P,k}$ does not hold for the $k$th ARQ cycle, then necessarily $\mathcal N_{k}^{(B,7)}\equiv\emptyset$.
Finally, let 
\begin{align}
\mathcal L_{S,k}^{(B,7)}\equiv \{l_{S,n},\forall n\in \mathcal N_{k}^{(B,7)}\}
\end{align}
be the set of SU packets transmitted in these slots,
where $\mathcal L_{S,k}^{(B,7)}\equiv\emptyset$ if the condition $B_{P,k}$ does not hold.
Then, we have
\begin{align}
\label{docoinistelabels}
&\sum_{l_S}\kappa_{S}(l_S)
=
\sum_{l_S\in\cup_k\mathcal L_{S,k}^{(B,7)}}\kappa_{S}(l_S)
+
\sum_{l_S\notin\cup_k\mathcal L_{S,k}^{(B,7)}}\kappa_{S}(l_S)
\nonumber\\&
\leq
\sum_{l_S\in\cup_k\mathcal L_{S,k}^{(B,7)}}\kappa_{S}(l_S)
+
\sum_{l_S\notin\cup_k\mathcal L_{S,k}^{(B,7)}}\sum_{n\in\mathcal N(l_S)}o_{S,n},
\end{align}
where we have used the inequality (\ref{ineqdellaminchia}).
We now analyze the decodability of $l_S\in\mathcal L_{S,k}^{(B,7)}$, $\kappa_{S,\nu_P(j+1)}(l_S)$.
We define, for $l_S\in\mathcal L_{S,k}^{(B,7)}$,
$\kappa_{S}^{-k}(l_S)$ as the decodability of $l_S$ by restricting the observation interval to
 the set of slots $\{0,1,\dots,\nu_P(j+1)-1\}\setminus \mathcal N_{k}^{(B,7)}$, \emph{i.e.},
 $\kappa_{S}^{-k}(l_S)=1$ if $l_S$ can be decoded successfully without the signal received in slots $\mathcal N_{k}^{(B,7)}$,
 and $\kappa_{S}^{-k}(l_S)=0$ otherwise.
 Clearly,
\begin{align}
\kappa_{S}^{-k}(l_S)\leq \kappa_{S}(l_S),
\end{align}
since the decodability improves if a larger number of slots is considered in the decoding process.
Then, we have the following cases. If
\begin{align}
\label{casep1}
\sum_{l_S\in\mathcal L_{S,k}^{(B,7)}}\kappa_{S}^{-k}(l_S)=0,
\end{align}
hence $\kappa_{S}^{-k}(l_S)=0,\forall l_S\in\mathcal L_{S,k}^{(B,7)}$,
then no packets $l_S\in\mathcal L_{S,k}^{(B,7)}$ can be decoded without considering the signals received in slots $\mathcal N_{k}^{(B,7)}$.
It follows that the PU packet transmitted in the $k$th ARQ cycle, $\nu_P(k)$, cannot be decoded, hence 
its interference cannot be removed, and therefore
\begin{align}
\kappa_{S}(l_S)=0,\forall l_S\in\mathcal L_{S,k}^{(B,7)}.
\end{align}
Note that, since $n\in\tau_{S}^{(GA)}(l_{S,n}),\forall n\in\mathcal N_{k}^{(B,7)}$,
hence $|\tau_{S}^{(GA)}(l_{S})|\geq 1,\forall l_S\in\mathcal L_{S,k}^{(B,7)}$,
it follows that 
\begin{align}
0=\sum_{l_S\in\mathcal L_{S,k}^{(B,7)}}\kappa_{S}(l_S)
\leq
\sum_{l_S\in\mathcal L_{S,k}^{(B,7)}}|\tau_{S}^{(GA)}(l_{S})|
-1
=
\sum_{l_S\in\mathcal L_{S,k}^{(B,7)}}\sum_{n\in\mathcal N(l_S)}o_{S,n}
-1.
\end{align}
On the other hand, if 
\begin{align}
\label{casep2}
\sum_{l_S\in\mathcal L_{S,k}^{(B,7)}}\kappa_{S}^{-k}(l_S)>0,
\end{align}
it follows that there exists some $\bar l_S\in\mathcal L_{S,k}^{(B,7)}$, transmitted in slot $\bar n\in\mathcal N_{k}^{(B,7)}$, such that 
$\kappa_{S}^{-k}(\bar l_S)=1$.
This SU packets can thus be decoded successfully without considering the slots $\mathcal N_{k}^{(B,7)}$ in the decoding process.
If this condition holds, then the PU packet transmitted in the $k$th ARQ cycle, $\nu_P(k)$, can be decoded after removing the interference from
such $\bar l_S$. All the SU packets $l_S\in\mathcal L_{S,k}^{(B,7)}$ can then be decoded after removing the interference from the
PU packet $\nu_P(k)$, since $y_{S,n}=7,\forall n\in\mathcal N_{k}^{(B,7)}$. It follows that
\begin{align}
\kappa_{S}(l_S)=1,\ \forall l_S\in\mathcal L_{S,k}^{(B,7)},
\end{align}
and therefore
\begin{align}
\label{asfsgdfghfsg}
\sum_{l_S\in\mathcal L_{S,k}^{(B,7)}}\kappa_{S}(l_S)
=
\sum_{l_S\in\mathcal L_{S,k}^{(B,7)}\setminus\{\bar l_S\}}\kappa_{S}(l_S)
+
\kappa_{S}(\bar l_S)
\leq
\sum_{l_S\in\mathcal L_{S,k}^{(B,7)}\setminus\{\bar l_S\}}\sum_{n\in\mathcal N(l_S)}o_{S,n}
+
\kappa_{S}(\bar l_S),
\end{align}
where we have used the inequality (\ref{ineqdellaminchia}),
and  $\kappa_{S}(\bar l_S)=1$.
Note that, for the SU packet $\bar l_S$, we have  $|\tau_{S}^{(GA)}(\bar l_{S})|>1$.
In fact, assuming that $\bar l_S$ is transmitted in slot $\bar n\in\mathcal N_{k}^{(B,7)}$,
then necessarily $\bar n\in\tau_{S}^{(GA)}(\bar l_{S})$, and thus  $\mathcal N_{k}^{(B,7)}\cap \tau_{S}^{(GA)}(\bar l_{S})\not\equiv\emptyset$.
Moreover, $\tau_{S}^{(GA)}(\bar l_{S})\setminus\mathcal N_{B,k}\not\equiv\emptyset$,
since $\bar l_S$ must be decodable without considering the slots $\mathcal N_{B,k}$, hence it 
 must satisfy the necessary condition (\ref{necessarycondS}) outside this set.
It follows that $1=\kappa_{S}(\bar l_S)\leq \sum_{n\in\mathcal N(\bar l_S)}o_{S,n} -1$,
hence, substituting in (\ref{asfsgdfghfsg}), we obtain
 the inequality
\begin{align}
\sum_{l_S\in\mathcal L_{S,k}^{(B,7)}}\kappa_{S}(l_S)
\leq
\sum_{l_S\in\mathcal L_{S,k}^{(B,7)}}\sum_{n\in\mathcal N(l_S)}o_{S,n}
-1,
\end{align}
which thus holds for both cases (\ref{casep1}) and (\ref{casep2}).
In general, since $\mathcal L_{S,k}^{(B,7)}\equiv\emptyset$ if the condition $B_{P,k}$ does not hold, for each $0\leq k\leq j$ we obtain the inequality
\begin{align}
\sum_{l_S\in\mathcal L_{S,k}^{(B,7)}}\kappa_{S}(l_S)
\leq
\sum_{l_S\in\mathcal L_{S,k}^{(B,7)}}\sum_{n\in\mathcal N(l_S)}o_{S,n}
-\chi(B_{P,k}).
\end{align}
Finally, substituting in (\ref{docoinistelabels}), we obtain
\begin{align}
&\sum_{l_S}\kappa_{S}(l_S)
=
\sum_{l_S\in\cup_k\mathcal L_{S,k}^{(B,7)}}\kappa_{S}(l_S)
+
\sum_{l_S\notin\cup_k\mathcal L_{S,k}^{(B,7)}}\kappa_{S}(l_S)
\nonumber\\&
\leq
\sum_{l_S\in\cup_k\mathcal L_{S,k}^{(B,7)}}\sum_{n\in\mathcal N(l_S)}o_{S,n}
-\sum_{k=0}^j\chi(B_{P,k})
+
\sum_{l_S\notin\cup_k\mathcal L_{S,k}^{(B,7)}}\sum_{n\in\mathcal N(l_S)}o_{S,n}
\nonumber\\&
=
\sum_{l_S}\sum_{n\in\mathcal N(l_S)}o_{S,n}
-\sum_{k=0}^j\chi(B_{P,k}).
\end{align}
Using the fact that $\sum_{l_S}\sum_{n\in\mathcal N(l_S)}o_{S,n}=\sum_{n=0}^{\nu_P(j+1)-1}o_{S,n}$,
we then obtain the inequality
\begin{align}
\label{docoioni}
&\sum_{l_S}\kappa_{S}(l_S)
\leq\sum_{n=0}^{\nu_P(j+1)-1}o_{S,n}-\sum_{k=0}^j\chi(B_{P,k}).
\end{align}
The expression (\ref{tsup}) is finally obtained by expressing 
 the condition $B_{P,k}$ as 
\begin{align}
\chi(B_{P,k})
=
\left(1-\prod_{n=\nu_P(k)}^{\nu_P(k+1)-1}\chi(y_{S,n}\neq 7)\right)\prod_{n=\nu_P(k)}^{\nu_P(k+1)-1}\chi(y_{S,n}\in\{2,4,5,7\}),
\end{align}
and by replacing the expression (\ref{osn}) in (\ref{docoioni}).
By dividing each side of (\ref{docoioni}) by $\nu_P(j+1)$, we finally obtain the inequality (\ref{ineq11}), thus proving the theorem.
\end{proof}

\section{Proof of Theorem \ref{thm2}}
\label{proofofthm2}
\begin{proof}
Note that $\mathcal L_S^{(CD)}$ is a causal function of $(\mathcal U_P,\mathcal Y_S)$, according to the CD rules, \emph{i.e.},
$l_{S,n}$ is a function of $\{(u_{P,m},y_{S,m}),m=0,1,\dots,n\}$.

Since $\{\kappa_{S,n}(l_S;\mathcal L_S^{(CD)}),n\geq 0\}$ is a  non-decreasing sequence,
and $N\geq \nu_P(\underline J(N)+1)$ from the definition of $\underline J(n)$ in (\ref{underlineJn}),
 we have
\begin{align}
\label{}
\sum_{l_S=0}^{\infty}\kappa_{S,N}(l_S;\mathcal L_S^{(CD)})\geq
\sum_{l_S=0}^{\infty}\kappa_{S,\nu_P(\underline J(N)+1)}(l_S;\mathcal L_S^{(CD)}),
\end{align}
and therefore, from (\ref{lowerbound}),
\begin{align}
\label{staroba2}
&\bar T_{S,N}(\mathcal L_S^{(CD)})
\geq
\frac{\nu_P(\underline J(N)+1)}{N}\bar T_{S,\nu_P(\underline J(N)+1)}(\mathcal L_S^{(CD)}).
\end{align}
In the following, we show that 
\begin{align}
\label{remaintoprove}
&\bar T_{S,\nu_P(\underline J(N)+1)}(\mathcal L_S^{(CD)})
\geq
\bar T_{S,\nu_P(\underline J(N)+1)}^{(up)},\ \forall N\geq 0.
\end{align}
By combining (\ref{staroba2}) and (\ref{remaintoprove}), the inequality (\ref{lowerbound}) then directly follows and the theorem is proved.
In order to prove (\ref{remaintoprove}), we use the following lemma.
\begin{lemma}
\label{lemn}
Let $j\geq 0$ and $\nu_P(j)<n\leq\nu_P(j+1)$. Let
\begin{align}
\label{vsn}
v_{S,n}\triangleq v_S(\rho_S(\mathcal G_n);\mathcal G_n)
\end{align}
be the CD potential of the root of the CD graph at the beginning of slot $n$
and
\begin{align}
\label{Msn}
M_{S,n}=\sum_{l_S}\kappa_{S,n}(l_S;\mathcal L_S^{(CD)})
\end{align}
be the number of SU packets successfully decoded up to the beginning of slot $n$. Then,
\begin{align}
\label{fkakgj}
&M_{S,n}+v_{S,n}-1
=
\sum_{m=0}^{n-1}\chi(y_{S,m}\in\{1,2,5,7\})
\\&\nonumber
-\sum_{k=0}^{j-1}\left[1-\kappa_{P}^{(GA)}(\nu_P(k))\right]\sum_{m=\nu_P(k)}^{\nu_P(k+1)-1}\chi(y_{S,m}\in\{5,7\})
\\&\nonumber
-\prod_{m=\nu_P(j)}^{n-1}\chi(y_{S,m}\in\{2,4,5\})\sum_{m=\nu_P(j)}^{n-1}\chi(y_{S,m}\in\{5,7\})
\\&\nonumber
-\sum_{k=0}^{j-1}\left(1-\prod_{m=\nu_P(k)}^{\nu_P(k+1)-1}\chi(y_{S,m}\neq 7)\right)\prod_{m=\nu_P(k)}^{\nu_P(k+1)-1}\chi(y_{S,m}\in\{2,4,5,7\})
\\&\nonumber
-\left(1-\prod_{m=\nu_P(j)}^{n-1}\chi(y_{S,m}\neq 7)\right)\prod_{m=\nu_P(j)}^{n-1}\chi(y_{S,m}\in\{2,4,5,7\}).
\end{align}
Moreover, in the special case $n=\nu_P(j+1)$,
\begin{align}
\label{byinspection}
\bar T_{S,\nu_P(j+1)}^{(up)}
=\frac{M_{S,\nu_P(j+1)}+v_{S,\nu_P(j+1)}-1}{\nu_P(j+1)}.
\end{align}
\qed
\end{lemma}
\begin{proof}
The expression (\ref{fkakgj}) is obtained by using Lemma \ref{lemmamecoioni} in Appendix \ref{proofoflemmamecoioni}  and induction on $n$.
The expression (\ref{byinspection}) is obtained by letting $n=\nu_P(j+1)$ and by inspection of (\ref{tsup}).
\end{proof}

Using the definition of $M_{S,n}$ in (\ref{Msn}) and (\ref{qweqweqwe}),
 we can rewrite 
\begin{align}
\bar T_{S,n}(\mathcal L_S^{(CD)})=\frac{M_{S,n}}{n}.
\end{align}
Therefore, using (\ref{byinspection}), (\ref{remaintoprove}) is equivalent to
\begin{align}
\bar T_{S,\nu_P(\underline J(N)+1)}(\mathcal L_S^{(CD)})=
\frac{M_{S,\nu_P(\underline J(N)+1)}}{\nu_P(\underline J(N)+1)}
\geq
\bar T_{S,\nu_P(\underline J(N)+1)}^{(up)}
=\frac{M_{S,\nu_P(\underline J(N)+1)}+v_{S,\nu_P(\underline J(N)+1)}-1}{\nu_P(\underline J(N)+1)},
\end{align}
or equivalently, 
\begin{align}
v_{S,\nu_P(\underline J(N)+1)}\leq 1,\forall N\geq 0.
\end{align}
Using the definition of $\underline J(N)$ in (\ref{underlineJn}), this is equivalent to proving
that 
\begin{align}
v_{S,\nu_P(k+1)}\leq 1,\forall k\geq 0\text{ s.t. }Q_k=1,
\end{align}
where $Q_k$ is defined in (\ref{Qk}).
This can be readily shown by following the recursions in the proof of Lemma \ref{lemmamecoioni} in Appendix \ref{proofoflemmamecoioni}.
In fact, for $k$ such that $Q_k=1$, we have that there exists 
$\nu_P(k)\leq m_1<m_2<\nu_P(k+1)$ such that $y_{S,m_1}\in\{1,3\}$ and $y_{S,m_2}\in\{1,2,5,7\}$.
Without loss of generality, assume $y_{S,t}\in\{3,4,6\},\forall m_1<t<m_2$.
$y_{S,m_1}\in\{1,3\}$ guarantees that the PU packet is successfully decoded in slot $m_1$, hence 
$\kappa_{P,m_1+1}(l_{P,m_1+1})=1$.
Therefore, from (\ref{spec4}),
\begin{align}
&M_{S,m_1+t+1}=M_{S,m_1+t}+\chi(y_{S,m_1+t}\in\{1,2,5,7\})v_{S,m_1+t},
\\
&v_{S,m_1+t+1}=\chi(y_{S,m_1+t}\in\{1,2,5,7\})+\chi(y_{S,m_1+t}\in\{3,4,6\})v_{S,m_1+t},
\end{align}
for all $1\leq t< \nu_P(k+1)-m_1$.
Using the fact that $y_{S,t}\in\{3,4,6\},\forall m_1<t<m_2$,
we obtain $v_{S,m_2}=v_{S,m_1+1}$. Then, since $y_{S,m_2}\in\{1,2,5,7\}$,
we obtain $v_{S,m_2+1}=1$ and, for $m_2-m_1\leq t< \nu_P(k+1)-m_1$,
$v_{S,m_1+t+1}=1$,
 so that $v_{S,\nu_P(k+1)}=1$.
The theorem is thus proved.
\end{proof}

\section{}
\label{proofoflemmamecoioni}
\begin{lemma}
\label{lemmamecoioni}
Let 
\begin{align}
v_{S,n}\triangleq v_S(\rho_S(\mathcal G_n);\mathcal G_n)
\end{align}
be the CD potential of the root of the CD graph at the beginning of slot $n$
and
\begin{align}
M_{S,n}=\sum_{l_S}\kappa_{S,n}(l_S;\mathcal L_S^{(CD)})
\end{align}
be the number of SU packets successfully decoded up to the beginning of slot $n$.
Then, we have $v_{S,0}=1$, $M_{S,0}=0$ and,  $\forall k\geq 0$, $\nu_P(k)\leq n<\nu_P(k+1)$, we have the recursion
\begin{align}
\label{tobeagreed}
&M_{S,n+1}+v_{S,n+1}
=
M_{S,n}+v_{S,n}
+\chi(y_{S,n}\in\{1,2,5,7\})
\nonumber\\&
- \prod_{m=\nu_P(k)}^{n-1}\chi(y_{S,m}\in\{2,4,5\})\chi(y_{S,n}\in\{1,3,5,6,7\})
\nonumber\\&
+\prod_{m=\nu_P(k)}^{n-1}\chi(y_{S,m}\in\{2,4,5\})\chi(y_{S,n}\in\{1,3,6,7\})\sum_{m=\nu_P(k)}^{n-1}\chi(y_{S,m}=5)
\nonumber\\&
+\prod_{m=\nu_P(k)}^{n-1}\chi(y_{S,m}\in\{2,4,5,7\})\chi(y_{S,n}\in\{1,3,6\}).
\end{align}
\end{lemma}
\begin{proof}
Consider slot $n$ within the $k$th ARQ cycle, \emph{i.e.}, $\nu_P(k)\leq n<\nu_P(k+1)$.
We have the following cases.

\noindent\emph{A) PU packet $l_{P,n}$ unknown ($\kappa_{P,n}(l_{P,n})=0$) and disconnected from the CD graph}
This is true in slot $n=\nu_P(k)$, \emph{i.e.}, at the beginning of the $k$th ARQ cycle. 
Therefore, according to {\bf R1}, the root is transmitted in slot $n$, $l_{S,n}=\rho_S(\mathcal G_n)$,
and has been transmitted in the previous slots $\nu_P(k)\leq m<n$.
This condition holds if and only if
\begin{align}
\label{x1}
y_{S,m}\in\{2,4\},\forall \nu_P(k)\leq m<n.
\end{align}
In fact, if $y_{S,m}\in\{1,3\}$ for some $\nu_P(k)\leq m<n$, then the PU packet is successfully decoded in slot $m$;
similarly, if $y_{S,m}\in\{5,6,7\}$ for some $\nu_P(k)\leq m<n$, then the PU packet 
becomes connected to the root of the CD graph.
Specializing (\ref{tobeagreed}) to this case and using (\ref{x1}),
we need to prove the recursion
\begin{align}
&M_{S,n+1}+v_{S,n+1}=M_{S,n}+v_{S,n}+\chi(y_{S,n}\in\{1,2\}).
\end{align}
In fact, if $y_{S,n}\in\{1,2\}$, the root is successfully decoded and the full CD potential is released, resulting in
\begin{align}
M_{S,n+1}=M_{S,n}+v_{S,n}.
\end{align}
The new root of the CD graph becomes $\rho_S(\mathcal G_{n+1})=n+1$ (new SU packet),
with CD potential $v_{S,n+1}=1$, so that 
\begin{align}
M_{S,n+1}+v_{S,n+1}=M_{S,n}+v_{S,n}+1.
\end{align}
Otherwise, the root of the CD graph remains unchanged, $\rho_S(\mathcal G_{n+1})=\rho_S(\mathcal G_{n})$,
with CD potential $v_{S,n+1}=v_{S,n}$, and no SU packets are decoded, so that 
$M_{S,n+1}=M_{S,n}$. It follows that 
\begin{align}
M_{S,n+1}+v_{S,n+1}=M_{S,n}+v_{S,n}.
\end{align}

\noindent\emph{B) PU packet $l_{P,n}$ unknown ($\kappa_{P,n}(l_{P,n})=0$) and connected to the root of the CD graph with $l_{P,n}\rightarrow\rho_S(\mathcal G_{n})$}
Therefore, according to {\bf R2}, a new SU packet is transmitted in slot $n$, $l_{S,n}=n_S$.
This condition holds if and only if
\begin{align}
\label{x2}
y_{S,m}\in\{2,4,5\},\forall \nu_P(k)\leq m<n
\cap
\exists \nu_P(k)\leq m<n: y_{S,m}=5.
\end{align}
In fact, if $y_{S,m}\in\{1,3\}$, for some $\nu_P(k)\leq m<n$,
then the PU packet is decoded in slot $m$, thus contradicting the hypothesis $\kappa_{P,n}(l_{P,n})=0$;
if $y_{S,m}\in\{6,7\}$, for some $\nu_P(k)\leq m<n$,
then the PU packet becomes connected to the root of the CD graph with $l_{P,n}\leftrightarrow\rho_S(\mathcal G_{n})$
or $\rho_S(\mathcal G_{n})\rightarrow l_{P,n}$, thus contradicting the hypothesis $l_{P,n}\rightarrow\rho_S(\mathcal G_{n})$;
finally, if $y_{S,m}\in\{2,4\},\forall \nu_P(k)\leq m<n$, then the case (\ref{x1}) holds and  $l_{P,n}$ is not connected to the root of the CD graph.

The condition (\ref{x2}) implies that  $y_{S,t}\in\{2,4\},\forall \nu_P(k)\leq t<m$ and $y_{S,m}=5$, for some $\nu_P(k)\leq m<n$,
so that the root is transmitted in slots $\nu_P(k)\leq t\leq m$, and the PU packet becomes connected to the CD graph in slot $m$;
in the following slots $m<t\leq n$, according to {\bf R2}, new SU packets are transmitted, so that $l_{S,t}=t_S,\forall m<t\leq n$.

Specializing (\ref{tobeagreed}) to this case and using (\ref{x2}),
we need to prove the recursion
\begin{align}
\label{spec1}
&M_{S,n+1}+v_{S,n+1}
=
M_{S,n}+v_{S,n}+\chi(y_{S,n}\in\{1,2\})
\nonumber\\&
+\chi(y_{S,n}\in\{1,3,6,7\})\sum_{m=\nu_P(k)}^{n-1}\chi(y_{S,m}=5).
\end{align}
We analyze all the possible cases:
\begin{itemize}
\item If $y_{S,n}\in\{1,3\}$, then $l_{P,n}$ is decoded and $l_{S,n}$ is decoded if $y_{S,n}=1$ ; since 
$l_{P,n}\rightarrow\rho_S(\mathcal G_{n})$, all the SU packets transmitted in the previous slots $\nu_P(k)\leq t<n$ such that $y_{S,t}=5$,
including the root, are decoded after removing the interference from the PU packet, hence the full CD potential is released. We thus obtain
\begin{align}
&M_{S,n+1}
=
M_{S,n}+v_{S,n}-1+\sum_{t=\nu_P(k)}^{n-1}\chi(y_{S,t}=5)+\chi(y_{S,n}=1);
\end{align}
in the next slot, the new root is $\rho_S(\mathcal G_{n+1})=n+1$, with CD potential $v_{S,n+1}=1$, hence 
\begin{align}
&M_{S,n+1}+v_{S,n+1}
=
M_{S,n}+v_{S,n}+\sum_{t=\nu_P(k)}^{n-1}\chi(y_{S,t}=5)+\chi(y_{S,n}=1),
\end{align}
so that (\ref{spec1}) holds.
\item If $y_{S,n}=2$, then only $l_{S,n}$ is decoded, so that $M_{S,n+1}=M_{S,n}+1$;
the root and its CD potential remain unchanged, so that $v_{S,n+1}=v_{S,n}$, so that (\ref{spec1}) holds.
\item If $y_{S,n}\in\{4,5\}$, then neither $l_{P,n}$ nor $l_{S,n}$ is decoded, so that $M_{S,n+1}=M_{S,n}$; the root of the CD graph and its CD potential remain unchanged,
so that $v_{S,n+1}=v_{S,n}$, which is consistent with (\ref{spec1}).
\item If $y_{S,n}\in\{6,7\}$, then neither $l_{P,n}$ nor $l_{S,n}$ is decoded, so that $M_{S,n+1}=M_{S,n}$; however, $l_{S,n}$ becomes the new root of the CD graph,
since its CD potential is $v_S(l_{S,n};\mathcal G_{n+1})=\sum_{t=\nu_P(k)}^{n-1}\chi(y_{S,t}=5)+v_{S,n}=v_{S,n+1}$,
as opposed to the previous root, with smaller CD potential  $v_S(\rho_S(\mathcal G_{n});\mathcal G_{n+1})=v_{S,n}$;
therefore, we obtain
\begin{align}
&M_{S,n+1}+v_{S,n+1}
=
M_{S,n}+v_{S,n}
\nonumber\\&
+\sum_{m=\nu_P(k)}^{n-1}\chi(y_{S,m}=5),
\end{align}
which is consistent with (\ref{spec1}).
\end{itemize}

\noindent\emph{C) PU packet $l_{P,n}$ unknown ($\kappa_{P,n}(l_{P,n})=0$) and connected to the root of the CD graph with $l_{P,n}\leftrightarrow\rho_S(\mathcal G_{n})$}
Therefore, according to {\bf R2}, a new SU packet is transmitted in slot $n$, $l_{S,n}=n_S$.
This condition holds if and only if
\begin{align}
\label{x3}
y_{S,m}\in\{2,4,5,7\},\forall \nu_P(k)\leq m<n
\cap
\exists \nu_P(k)\leq m<n: y_{S,m}=7.
\end{align}
In fact, if $y_{S,m}\in\{1,3\}$, for some $\nu_P(k)\leq m<n$,
then the PU packet is decoded in slot $m$, thus contradicting the hypothesis $\kappa_{P,n}(l_{P,n})=0$;
if  $y_{S,m}=6$, for some $\nu_P(k)\leq m<n$,
then the PU packet becomes connected to the root of the CD graph with
$\rho_S(\mathcal G_{n})\rightarrow l_{P,n}$, thus contradicting the hypothesis $l_{P,n}\leftrightarrow\rho_S(\mathcal G_{n})$;
finally, if $y_{S,m}\in\{2,4,5\},\forall \nu_P(k)\leq m<n$, then the previous cases (\ref{x1}) or (\ref{x2}) hold.

The condition (\ref{x3}) implies that  $y_{S,t}\in\{2,4\},\forall \nu_P(k)\leq t<m$ and $y_{S,m}\in\{5,7\}$, for some $\nu_P(k)\leq m<n$,
so that the root is transmitted in slots $\nu_P(k)\leq t\leq m$, and the PU packet becomes connected to the CD graph in slot $m$;
in the following slots $m<t\leq n$, according to {\bf R2}, new SU packets are transmitted, so that $l_{S,t}=t_S,\forall m<t\leq n$.
If $y_{S,m}=5$, then condition  (\ref{x3}) implies that there exists also some $m<\tilde m<n$ such that $y_{S,\tilde m}=7$,
so that $l_{S,\tilde m}$ becomes the new root of the CD graph.

Specializing (\ref{tobeagreed}) to this case and using (\ref{x3}),
we need to prove the recursion
\begin{align}
\label{spec2}
&M_{S,n+1}+v_{S,n+1}
=
M_{S,n}+v_{S,n}
+\chi(y_{S,n}\in\{1,2,5,7\})
+\chi(y_{S,n}\in\{1,3,6\}).
\end{align}
We analyze all the possible cases:
\begin{itemize}
\item If $y_{S,n}\in\{1,3\}$, then $l_{P,n}$ is decoded and $l_{S,n}$ is decoded if $y_{S,n}=1$ ; since 
$l_{P,n}\leftrightarrow\rho_S(\mathcal G_{n})$, the root is decoded and its CD potential is released,
so that
\begin{align}
&M_{S,n+1}
=
M_{S,n}+v_{S,n}+\chi(y_{S,n}=1);
\end{align}
in the next slot, the new root is $\rho_S(\mathcal G_{n+1})=n+1$, with CD potential $v_{S,n+1}=1$, hence 
\begin{align}
&M_{S,n+1}+v_{S,n+1}
=
M_{S,n}+v_{S,n}+1+\chi(y_{S,n}=1),
\end{align}
which is consistent with (\ref{spec2}).
\item If $y_{S,n}\in\{2,4\}$, then only $l_{S,n}$ is decoded if $y_{S,n}=2$, so that $M_{S,n+1}=M_{S,n}+\chi(y_{S,n}=2)$;
the root and its CD potential remain unchanged, so that $v_{S,n+1}=v_{S,n}$ and (\ref{spec1}) holds.
\item If $y_{S,n}\in\{5,7\}$, then neither $l_{P,n}$ nor $l_{S,n}$ is decoded, so that $M_{S,n+1}=M_{S,n}$; 
however, $l_{P,n}\rightarrow l_{S,n}$ (if $y_{S,n}=5$) or $l_{P,n}\leftrightarrow l_{S,n}$ (if $y_{S,n}=7$), so that $l_{S,n}$ becomes reachable from the root of the CD graph, whose CD potential thus increases by one unit, yielding
$v_{S,n+1}=v_{S,n}+1$. This is consistent with (\ref{spec2}).
\item If $y_{S,n}=6$, then neither $l_{P,n}$ nor $l_{S,n}$ is decoded, so that $M_{S,n+1}=M_{S,n}$; however, $l_{S,n}$ becomes the new root of the CD graph,
since its CD potential is $v_S(l_{S,n};\mathcal G_{n+1})=v_{S,n}+1$,
as opposed to the previous root, with smaller CD potential  $v_S(\rho_S(\mathcal G_{n});\mathcal G_{n+1})=v_{S,n}$;
therefore, we obtain
\begin{align}
&M_{S,n+1}+v_{S,n+1}
=
M_{S,n}+v_{S,n}+1,
\end{align}
which is consistent with (\ref{spec2}).
\end{itemize}

\noindent\emph{D) PU packet $l_{P,n}$ unknown ($\kappa_{P,n}(l_{P,n})=0$) and connected to the root of the CD graph with $\rho_S(\mathcal G_{n})\rightarrow l_{P,n}$}
Therefore, according to {\bf R2}, a new SU packet is transmitted in slot $n$, $l_{S,n}=n_S$.
This condition holds if and only if
\begin{align}
\label{x4}
y_{S,m}\in\{2,4,5,6,7\},\forall \nu_P(k)\leq m<n
\cap
\exists \nu_P(k)\leq m<n: y_{S,m}=6.
\end{align}
In fact, if $y_{S,m}\in\{1,3\}$, for some $\nu_P(k)\leq m<n$,
then the PU packet is decoded in slot $m$, thus contradicting the hypothesis $\kappa_{P,n}(l_{P,n})=0$;
finally, if $y_{S,m}\in\{2,4,5,7\},\forall \nu_P(k)\leq m<n$, then the previous cases (\ref{x1}), (\ref{x2}) or (\ref{x3}) hold.

The condition (\ref{x4}) implies that  $y_{S,t}\in\{2,4\},\forall \nu_P(k)\leq t<m$ and $y_{S,m}\in\{5,6,7\}$, for some $\nu_P(k)\leq m<n$,
so that the root is transmitted in slots $\nu_P(k)\leq t\leq m$, and the PU packet becomes connected to the CD graph in slot $m$;
in the following slots $m<t\leq n$, according to {\bf R2}, new SU packets are transmitted, so that $l_{S,t}=t_S,\forall m<t\leq n$.
If $y_{S,m}\in\{5,7\}$, then condition  (\ref{x4}) implies that there exists also some $m<\tilde m<n$ such that $y_{S,\tilde m}=6$,
so that $l_{S,\tilde m}$ becomes the new root of the CD graph (in fact, it is the SU packet with the largest CD potential).

Specializing (\ref{tobeagreed}) to this case and using (\ref{x4}),
we need to prove the recursion
\begin{align}
\label{spec3}
&M_{S,n+1}+v_{S,n+1}
=
M_{S,n}+v_{S,n}+\chi(y_{S,n}\in\{1,2,5,7\}).
\end{align}
We analyze all the possible cases:
\begin{itemize}
\item If $y_{S,n}\in\{1,3\}$, then $l_{P,n}$ is decoded and $l_{S,n}$ is decoded if $y_{S,n}=1$;
then, the CD potential of $l_{P,n}$ is released, 
$v_P(l_{P,n};\mathcal G_n)$, so that 
\begin{align}
&M_{S,n+1}
=
M_{S,n}+v_P(l_{P,n};\mathcal G_n)+\chi(y_{S,n}=1);
\end{align}
note that $v_P(l_{P,n};\mathcal G_n)\leq v_{S,n}-1$, since 
$\rho_S(\mathcal G_{n})\rightarrow l_{P,n}$ and thus the root of the CD graph has CD potential strictly larger than that of $l_{P,n}$.
After decoding the SU packets reachable from $l_{P,n}$, that portion of the graph is removed,
so that the root remains unchanged but its CD potential is decreased, resulting in $v_{S,n+1}=v_{S,n}-v_P(l_{P,n};\mathcal G_n)$.
We obtain
\begin{align}
&M_{S,n+1}+v_{S,n+1}
=
M_{S,n}+v_{S,n}+\chi(y_{S,n}=1),
\end{align}
which is consistent with (\ref{spec3}).
\item If $y_{S,n}\in\{2,4\}$, then only $l_{S,n}$ is decoded if $y_{S,n}=2$, so that $M_{S,n+1}=M_{S,n}+\chi(y_{S,n}=2)$;
the root and its CD potential remain unchanged, so that $v_{S,n+1}=v_{S,n}$ and (\ref{spec3}) holds.
\item If $y_{S,n}\in\{5,7\}$, then neither $l_{P,n}$ nor $l_{S,n}$ is decoded, so that $M_{S,n+1}=M_{S,n}$; 
however, $l_{P,n}\rightarrow l_{S,n}$ (if $y_{S,n}=5$) or $l_{P,n}\leftrightarrow l_{S,n}$ (if $y_{S,n}=7$), so that $l_{S,n}$ becomes reachable from the root of the CD graph, whose CD potential thus increases by one unit, yielding
$v_{S,n+1}=v_{S,n}+1$. This is consistent with (\ref{spec3}).
\item If $y_{S,n}=6$, then neither $l_{P,n}$ nor $l_{S,n}$ is decoded, so that $M_{S,n+1}=M_{S,n}$; $l_{S,n}$ becomes the new root of the CD graph,
since its CD potential equals that of the previous root, but $l_{S,n}=n_S$ is a more recent SU packet, therefore
$v_{S,n+1}=v_{S,n}$, which is  consistent with (\ref{spec3}).
\end{itemize}

\noindent\emph{E) PU packet $l_{P,n}$ known ($\kappa_{P,n}(l_{P,n})=1$)}
Therefore, according to {\bf R3}, the root of the CD graph is transmitted in slot $n$, $l_{S,n}=\rho_S(\mathcal G_{n})$.
This condition holds if and only if
\begin{align}
\label{x5}
\exists \nu_P(k)\leq m<n: y_{S,m}\in\{1,3\}.
\end{align}
In fact, if the above condition is not satisfied, we fall in one of the cases (\ref{x1}), (\ref{x2}), (\ref{x3}) or  (\ref{x4}) analyzed before.

Specializing (\ref{tobeagreed}) to this case and using (\ref{x4}),
we need to prove the recursion
\begin{align}
\label{spec4}
&M_{S,n+1}+v_{S,n+1}
=
M_{S,n}+v_{S,n}+\chi(y_{S,n}\in\{1,2,5,7\})
\end{align}
We have the following two cases:
\begin{itemize}
\item If $y_{S,n}\in\{1,2,5,7\}$, then $l_{S,n}=\rho_S(\mathcal G_{n})$ can be decoded successfully after removing the interference from the PU packet,
and thus the CD potential of the root can be released, resulting in 
\begin{align}
&M_{S,n+1}
=
M_{S,n}+v_{S,n}.
\end{align}
In the next slot, the new root of the CD graph is the new SU packet $\rho_S(\mathcal G_{n+1})=n+1$, with CD potential
$v_{S,n+1}=1$, yielding 
\begin{align}
&M_{S,n+1}+v_{S,n+1}
=
M_{S,n}+v_{S,n}+1,
\end{align}
which is consistent with (\ref{spec4}).
\item  Otherwise,
$l_{S,n}=\rho_S(\mathcal G_{n})$ cannot be decoded successfully even after removing the interference from the PU packet,
and no CD potential can be released, resulting in 
\begin{align}
&M_{S,n+1}
=
M_{S,n}.
\end{align}
In the next slot, the root of the CD graph and its CD potential remain unchanged, so that 
$v_{S,n+1}=v_{S,n}$, which is consistent with (\ref{spec4}).
\end{itemize}
In general, for this case we have the dynamics
\begin{align}
&M_{S,n+1}=M_{S,n}+\chi(y_{S,n}\in\{1,2,5,7\})v_{S,n},
\\
&v_{S,n+1}=\chi(y_{S,n}\in\{1,2,5,7\})+\chi(y_{S,n}\in\{3,4,6\})v_{S,n}.
\end{align}
\end{proof}

\section{Proof of Theorem \ref{thm6}}
\label{proofofthm6}
\begin{proof}
We need to prove that $\mathbf s_{n+1}$ is independent of the past, given $\mathbf s_n$ and $a_{S,n}$,
and that the expected virtual instantaneous throughput $\mathbb E[g(\cdot)]$ accrued in slot $n$ is a function of $\mathbf s_n$ and $a_{S,n}$ only.
Therefore, let $\mathbf s_n$ and $a_{S,n}$ be given.
At the end of slot $n$, the SU pair overhears the PU feedback $y_{P,n}$.
The distribution of  $y_{P,n}$
depends on $\mathbf s_{P,n}$ only and is independent of the past, as in (\ref{pmfY});
in turn, the distribution of the internal PU state $\mathbf s_{P,n}=(t_{P,n},d_{P,n},q_{P,n})$
is a function of $\mathbf s_n=(\mathbf s_{CD,n},t_{P,n},d_{P,n},\beta_n)$, since $\beta_n(q_P)=\mathbb P(q_{P,n}=q_P)$.
It follows that the distribution of $y_{P,n}$ is independent of the past history, given $(\mathbf s_n,a_{S,n})$.

Given $o_{P,n}$, $\beta_{n+1}$ can be computed as
\begin{align}
\label{nextbeta}
&\beta_{n+1}(q)=\mathbb P(q_{P,n+1}=q|t_{P,n}=t,d_{P,n}=d,\beta_n,o_{P,n}=o)
\nonumber\\&
=
\sum_{\tilde q}\sum_{b}\mathbb P(q_{P,n+1}=q,q_{P,n}=\tilde q,b_{P,n}=b|t_{P,n}=t,d_{P,n}=d,\beta_n,o_{P,n}=o)
\nonumber\\&
=
\frac{
\begin{array}{l}
\sum_{\tilde q}\mathbb P(q_{P,n}=\tilde q|\beta_n)\sum_{b}
\mathbb P(b_{P,n}=b)
\mathbb P(q_{P,n+1}=q|q_{P,n}=\tilde q,b_{P,n}=b,o_{P,n}=o)\\
\times
\sum_{a\in\{0,1\}}
\mathbb P(o_{P,n}=o|q_{P,n}=\tilde q,t_{P,n}=t,d_{P,n}=d,a_{P,n}=a)
\\\times
\mathbb P(a_{P,n}=a|q_{P,n}=\tilde q,t_{P,n}=t,d_{P,n}=d)
\end{array}
}{
\begin{array}{l}
\sum_{\tilde q}\mathbb P(q_{P,n}=\tilde q|\beta_n)
\sum_{a\in\{0,1\}}
\mathbb P(o_{P,n}=o|q_{P,n}=\tilde q,t_{P,n}=t,d_{P,n}=d,a_{P,n}=a)
\\\times
\mathbb P(a_{P,n}=a|q_{P,n}=\tilde q,t_{P,n}=t,d_{P,n}=d)
\end{array}
}
\nonumber\\&
=
\frac{
\begin{array}{l}
\sum_{\tilde q}\beta_n(\tilde q)\sum_{b}
\mathbb P_B(b)
\chi\left(q=\min\{\tilde q-o+b,Q_{\max}\}\right)
\\\times
\left[\mu_P(t,d,\tilde q)\mathbb P_O(o|t,d,\tilde q,1)
+(1-\mu_P(t,d,\tilde q))\mathbb P_O(o|t,d,\tilde q,0)\right]
\end{array}
}{
\sum_{\tilde q}\beta_n(\tilde q)\sum_{b}
\mathbb P_B(b)
\left[\mu_P(t,d,\tilde q)\mathbb P_O(o|t,d,\tilde q,1)
+(1-\mu_P(t,d,\tilde q))\mathbb P_O(o|t,d,\tilde q,0)\right]
},
\end{align}
where, from (\ref{o}),
 \begin{align}
& \mathbb P_O(1|q,t,d,a)
 \triangleq
 \mathbb P(o_{P,n}=1|q_{P,n}=q,t_{P,n}=t,d_{P,n}=d,a_{P,n}=a)
\nonumber\\&
 =
  (1-a)\chi(d=D_{\max}-1)
 +
 a\mathbb P(\boldsymbol{\gamma}_P\in\Gamma_P(a))\chi(q>0)
 \\&
+a[1-\mathbb P(\boldsymbol{\gamma}_P\in\Gamma_P(a))]\chi(q>0)\chi(t=R_{\max}-1)
 \\&
+a[1-\mathbb P(\boldsymbol{\gamma}_P\in\Gamma_P(a))]\chi(q>0)\chi(t<R_{\max}-1)\chi(d=D_{\max}-1)
 \end{align}
 and $\mathbb P_O(0|q,t,d,a)=1-\mathbb P_O(1|q,t,d,a)$.
 We can thus write $\beta_{n+1}=f(t_{P,n},d_{P,n},\beta_n,y_{P,n})$ for a proper function $f(\cdot)$, as given by (\ref{nextbeta}), where, in turn,
 $o_{P,n}$ is a function of $(t_{P,n},d_{P,n},y_{P,n})$ via (\ref{sigma}).
Therefore, $\beta_{n+1}$ is independent of the past, given $(\mathbf s_n,a_{S,n})$.

From (\ref{tp}), (\ref{dp}) and (\ref{sigma}), and using the fact that
$a_{P,n}=\chi(y_{P,n}\neq\emptyset)$,
 we can write $t_{P,n+1}$ and $d_{P,n+1}$ as
\begin{align}
&t_{P,n+1}=(1-\sigma(t_{P,n},d_{P,n},y_{P,n}))(t_{P,n}+\chi(y_{P,n}\neq\emptyset)),
\\&d_{P,n+1}=(1-\sigma(t_{P,n},d_{P,n},y_{P,n}))\left[d_{P,n}+\chi(t_{P,n}>0)+\chi(t_{P,n}=0)\chi(y_{P,n}\neq\emptyset)\right],
\end{align}
so that $(t_{P,n+1},d_{P,n+1})$ is a function of $(t_{P,n},d_{P,n},y_{P,n})$ only.
Therefore, $(t_{P,n+1},d_{P,n+1})$ is independent of the past, given $(\mathbf s_n,a_{S,n})$.

Consider the CD state $\mathbf s_{CD,n}=(\Phi_n,b_{S,n})\in\mathcal W$.
 Note that there is a one-to-one mapping between $\Phi_n$ and 
 $(\hat\kappa_{P,n}^{(GA)},\iota_{P,n})$. Therefore, given $\Phi_n$, the pair  $(\hat\kappa_{P,n}^{(GA)},\iota_{P,n})$ is given.
 We have the following cases for $(\hat\kappa_{P,n+1}^{(GA)},\iota_{P,n+1},b_{S,n+1})$: if $o_{P,n}=1$,
 so that the current ARQ cycle ends and a new one begins in the next slot,  
from (\ref{kappapnga}),  (\ref{iota}) and (\ref{bsn})
we obtain $(\hat\kappa_{P,n+1}^{(GA)},\iota_{P,n+1},b_{S,n+1})=(0,1,0)$;
on the other hand, if  $o_{P,n}=0$,
  $(\hat\kappa_{P,n+1}^{(GA)},\iota_{P,n+1},b_{S,n+1})$ can be determined recursively from (\ref{kappapnga}),  (\ref{iota}) and (\ref{bsn}); we thus obtain
 \begin{align}
&\hat\kappa_{P,n+1}^{(GA)}\triangleq (1-o_{P,n})\left[1-\left(1-\hat\kappa_{P,n}^{(GA)}\right)\left(1-a_{P,n}\chi(y_{S,n}\in\{1,3,6,7\})\right)\right],
\\&
\iota_{P,n+1}=o_{P,n}+(1-o_{P,n})\iota_{P,n}\left[1-a_{P,n}\chi(y_{S,n}\in\{1,3,6,7\})+a_{P,n}a_{S,n}\chi(y_{S,n}=7)\right],
\\&
b_{S,n+1}=(1-o_{P,n})\left(1-\hat\kappa_{P,n+1}^{(GA)}\right)\sum_{m=\nu_P(k)}^{n}a_{P,m}a_{S,m}\chi(y_{S,m}=5)
\nonumber\\&
=
(1-o_{P,n})\left[1-a_{P,n}\chi\left(y_{S,n}\in\{1,3,6,7\}\right)\right]\left[b_{S,n}+\left(1-\hat\kappa_{P,n}^{(GA)}\right)a_{P,n}a_{S,n}\chi(y_{S,n}=5)\right],
    \end{align}
    so that $(\hat\kappa_{P,n+1}^{(GA)},\iota_{P,n+1},b_{S,n+1})$
are functions of 
$(\hat\kappa_{P,n}^{(GA)},\iota_{P,n},b_{S,n})$,
 $y_{S,n}$, $a_{P,n}$, $a_{S,n}$, and $o_{P,n}$.
   Since $(\hat\kappa_{P,n+1}^{(GA)},\iota_{P,n+1},b_{S,n+1})$ can be mapped to the new CD state $\mathbf s_{CD,n+1}$,
   it follows that $\mathbf s_{CD,n+1}$ is a function of $(\mathbf s_{CD,n},o_{P,n},y_{S,n},a_{P,n},a_{S,n})$.
   In turn, $o_{P,n}$ is a function of $(t_{P,n},d_{P,n},y_{P,n})$ via (\ref{sigma});
   $y_{S,n}$ is i.i.d. over time; $a_{P,n}=\chi(y_{P,n}\neq \emptyset)$, and $(\mathbf s_{CD,n},a_{S,n})$ is given.
   We conclude that $\mathbf s_{CD,n+1}$ is statistically independent of the past, given $(\mathbf s_n,a_{S,n})$.
   
Finally,
the virtual instantaneous throughput $g(\cdot)$ and PU reward $\mathbf r_P$ accrued in slot $n$ are statistically independent of the 
past, given $(\mathbf s_n,a_{S,n})$. In fact, these are functions of
$(\mathbf s_{P,n},b_{P,n},\boldsymbol{\gamma}_{P,n},a_{P,n},a_{S,n})$
and $(a_{S,n},a_{P,n},y_{S,n},\hat\kappa_{P,n}^{(GA)},\iota_{P,n},b_{S,n})$, respectively.
As shown above, $(\mathbf s_{P,n},a_{P,n})$ are independent of the past, given $(\mathbf s_n,a_{S,n})$;
$b_{P,n}$ has distribution $\beta_n$, as seen from the SU pair;
$(\boldsymbol{\gamma}_{P,n},y_{S,n})$ are i.i.d. over time;
$(\hat\kappa_{P,n}^{(GA)},\iota_{P,n},b_{S,n})$ are univocally determined by $\mathbf s_n$.
\end{proof}

\bibliographystyle{IEEEtran}
\bibliography{IEEEabrv,biblio}

\end{document}